\documentclass[10pt,A4paper]{article}
\usepackage{amssymb}
\usepackage{amsfonts}
\usepackage{amsmath}
\usepackage{amsthm}
\usepackage{latexsym}
\usepackage[dvips]{epsfig}
\usepackage{enumerate}
\usepackage{mathrsfs}
\usepackage{eufrak}
\usepackage{bm}
\usepackage{tikz}
\usepackage{authblk}
\usepackage{cancel}

\usepackage{color}

\usepackage{stmaryrd}

\usepackage[T1]{fontenc}
\usepackage[utf8]{inputenc}

\theoremstyle{plain}
\newtheorem{proposition}{Proposition}

\newtheorem{theorem}{Theorem}

\newtheorem{remark}{Remark}

\def\bma{{\bm a}}
\def\bmb{{\bm b}}
\def\bmc{{\bm c}}

\def\bml{{\bm l}}
\def\bmn{{\bm n}}
\def\bmm{{\bm m}}

\def\bmA{{\bm A}}
\def\bmB{{\bm B}}
\def\bmC{{\bm C}}

\def\bmzeta{{\bm\zeta}}

\def\bmphi{{\bm \phi}}

\def\bmGamma{{\bm \Gamma}}

\def\bmDelta{{\bm \Delta}}
\def\bmPsi{{\bm \Psi}}
\def\bmphi{{\bm \phi}}
\def\bmzeta{{\bm \zeta}}

\def\bmpartial{{\bm \partial}}

\def\nablasl{/\kern-0.58em\nabla}
\def\Deltasl{/\kern-0.7em\Delta}
\def\Dsl{/\kern-0.7em D}

\def\sigmasl{/\kern-0.58em\sigma}

\allowdisplaybreaks

\def\TiPsi{{\tilde \Psi}}

\def\Timu{{\tilde \mu}}

\def\Tizeta{{\tilde \zeta}}

\def\ulchi{{\underline{\chi}}}
\def\ulomega{{\underline{\omega}}}

\def\ulvarrho{{\underline{\varrho}}}

\newcommand{\mthorn}{\textit{þ}} 
\newcommand{\meth}{\textit{ð}} 

\newcounter{mnotecount}

\newcommand{\mnotex}[1]
{\protect{\stepcounter{mnotecount}}$^{\mbox{\footnotesize $\bullet$\themnotecount}}$
\marginpar{
\raggedright\tiny\em
$\!\!\!\!\!\!\,\bullet$\themnotecount: #1} }

\newcommand*{\rom}[1]{\expandafter\@slowromancap\romannumeral #1@}

\usepackage{hyperref}

\begin{document}
\title{\textbf{Trapped surface formation for the Einstein-Weyl spinor system}}

\author[,1]{Peng
  Zhao \footnote{E-mail address:{\tt p.zhao@bnu.edu.cn}}}
\author[,2,3]{Xiaoning Wu \footnote{E-mail address:{\tt
      wuxn@amss.ac.cn}}}
\affil[1]{Faculty of Arts and Sciences, Beijing Normal University, Zhuhai, 519087, China.}        
\affil[2]{Institute of Mathematics, Academy of Mathematics and Systems Science 
and State Key Laboratory of Mathematical Sciences, Chinese Academy of Sciences, Beijing 100190, China.}
\affil[3]{School of Mathematical Sciences, University of Chinese Academy of Sciences, 
Beijing 100049, China.}

\maketitle

\begin{abstract}
We prove trapped‐surface formation for the Einstein-Weyl spinor system (gravity coupled to a massless left-handed two-spinor) without any symmetry assumption. 
To this end we establish a semi-global solution under double null foliation and show that 
the focusing of the gravitational waves and the Weyl spinor flux leads to the formation of a trapped surface.
\end{abstract}

\section{Introduction}

Following Penrose's singularity theorem, the presence of a trapped surface-a closed 
spacelike 2-sphere on which both future-directed null expansions are negative-leads 
to the incompleteness of future causal geodesics. 
In practice this is taken as rigorous evidence for a black hole. 
Consequently, the dynamical formation of a black hole can be framed-mathematically 
and conceptually-as the problem of producing a trapped surface from physically meaningful initial data.
There is by now a substantial body of work on trapped-surface formation 
in vacuum, including foundational results of Christodoulou \cite{Chr00}, and continued in 
\cite{KlaRod09, An2012, KlainermanLukRodnianski14, AnLuk2017, An2022}
and more recent criteria in matter models (e.g. electromagnetic, scalar, and Yang-Mills fields) in 
\cite{PYu2011, LiLiu2017, An202209, AthanasiouPuskarYau25, HilValZha23, DaweiJinbo25}. 
For recent nonlinear numerical simulations of black-hole formation and related dynamics, see, e.g., 
\cite{ZhangChenLiuLuoTianWang22,NingChenTianWuZhang24,FiguerasFranca20,BantilanFiguerasRossi21}.

In this work, we study the Einstein equations minimally coupled to a left-handed Weyl spinor; 
see Sec. \ref{IntroEinsteinWeyl}.
Weyl spinors model massless, chiral spin-$1/2$ degrees of freedom that appear ubiquitously in the Standard Model. 
Understanding their dynamics under strong gravity provides a baseline 
for gravitational collapse of fermionic, e.g. neutron, proton, electron and so on. 
On the small-data, far-from-collapse side, Chen \cite{Chen-SpinorStability} 
has established the global nonlinear stability of Minkowski space for 
the massless Einstein-Dirac system in generalized wave gauge via the vector-field method, 
providing a robust baseline for the diffusion dynamics of spinor fields coupled to gravity.
On the static/spherically symmetric side, 
Finster-Smoller-Yau \cite{FSY-EDM} proved a no-hair type result for the spherically symmetric, 
static Einstein-Dirac-Maxwell system: under mild regularity of the horizon, 
the only black-hole solutions have vanishing spinors in the exterior region and equivalently, 
reduce to Reissner-Nordstr\"om. 
Beyond Finster-Smoller-Yau, stationary nonexistence results now cover charged, rotating (A)dS geometries: 
Wang-Zhang \cite{YaohuaXiao18} rule out time-periodic Dirac fields in non-extreme Kerr-Newman-AdS, 
while Zhang-Zhang \cite{HequnXiao24} exclude differentiable time-periodic Majorana fermions 
in Kerr-Newman[-(A)dS] with nonzero charge.
Numerically, nonlinear simulations by 
Ventrella-Choptuik \cite{VC-DiracCritical} for the massless Einstein-Dirac system 
display Type~II critical collapse with a characteristic mass-scaling exponent, 
indicating that spinor flux can trigger black-hole formation at threshold. 
To the best of our knowledge, however, a symmetry-free, characteristic-data theory 
that proves trapped-surface formation for Dirac/Weyl spinors-on 
par with vacuum and classical matter cases—has not yet been established.
This gap motivates our present analysis of trapped-surface formation in the Einstein-Weyl spinor system.

\paragraph{Conventions.}
In this article, lowercase Latin $a,b,c,\dots$ denote abstract tensor indices, 
while $\bma,\bmb,\bmc,\dots$ are tetrad indices taking values $0,\dots,3$; 
uppercase Latin $A,B,C,\dots$ are abstract two–spinor indices, 
and $\bmA,\bmB,\bmC,\dots$ are spin–frame indices taking values $0,1$. 
The alternating spinor $\epsilon_{AB}$ gives the pairing 
$\llbracket\xi,\eta\rrbracket:=\epsilon_{AB}\,\xi^{A}\eta^{B}$; 
indices are moved with $\epsilon^{AB}$ and $\epsilon_{AB}$, e.g. 
$\xi_{B}=\epsilon_{AB}\,\xi^{A}$. For a normalized spin dyad $\{o^{A},\iota^{A}\}$ with $o_{A}\iota^{A}=1$, 
we have $\epsilon_{AB}=o_{A}\iota_{B}-o_{B}\iota_{A}$, and we set $\epsilon_{\bm{0}}{}^{A}=o^{A}$, 
$\epsilon_{\bm{1}}{}^{A}=\iota^{A}$. We choose an orthonormal tetrad $e_{\bma}$ 
with dual $\omega^{\bma}$ so that $g_{\bma\bmb}=\eta_{\bma\bmb}$ with signature $(+,-,-,-)$. 
The Infeld–van der Waerden symbols $\sigma^{\bma}{}_{\bmA\bmA'}$ 
relate vectors to bi-spinors via 
$\epsilon_{\bmA\bmB}\,\epsilon_{\bmA'\bmB'}=\eta_{\bma\bmb}\,\sigma^{\bma}{}_{\bmA\bmA'}\,\sigma^{\bmb}{}_{\bmB\bmB'}$; 
here $\sqrt{2}\,\sigma^{\bma}{}_{\bmA\bmA'}$ corresponds to $\{\mathbf{1},\sigma_{1},\sigma_{2},\sigma_{3}\}$ 
where $\sigma_{i}$ is the Pauli matries, 
and $\sigma_{\bma}{}^{\bmA\bmA'}$ denotes the inverse 
with $\sigma^{\bma}{}_{\bmA\bmA'}\sigma_{\bma}{}^{\bmB\bmB'}=\epsilon_{\bmA}{}^{\bmB}\epsilon_{\bmA'}{}^{\bmB'}$. 
For a tensor $T_{a}{}^{b}$ we define its spinorial representative 
by $T_{\bmA\bmA'}{}^{\ \ \bmB\bmB'}:=T_{\bma}{}^{\bmb}\,\sigma^{\bma}{}_{\bmA\bmA'}\,\sigma_{\bmb}{}^{\bmB\bmB'}$, 
giving the standard correspondence between $T_{a}{}^{b}$ and $T_{A}{}^{B}$; 
we also use $g_{AA'BB'}=\epsilon_{AB}\epsilon_{A'B'}$. 
Our curvature sign convention is $\nabla_{a}\nabla_{b}\omega_{c}-\nabla_{b}\nabla_{a}\omega_{c}=-R_{abc}{}^{d}\,\omega_{d}$. 
Spinor conventions follow \cite{PenRin86,Ste91,CFEBook}.

\subsection{Einstein-Weyl System}
\label{IntroEinsteinWeyl}
In the massless limit a Dirac spinor decomposes into two decoupled Weyl spinors of opposite chirality. 
The present work retains only the left-handed component, 
i.e. a single-chirality subsector of the massless Einstein–Dirac system. 
For the massive Dirac field, we refer the reader to \cite{PengXiaoning2501}.
In the following, we denote by $\phi_A$ the left-handed Weyl spinor field. 
The Einstein-Weyl spinor system is
\begin{align}
R_{ab}-\frac{1}{2}Rg_{ab}=&T_{ab}, \label{EinsteinFE}\\
\nabla_{AA'}\phi^A=&0,
\end{align}
where $\nabla_{AA'}$ is the spinorial counterpart of the covariant derivative, 
$T_{ab}$ is the energy-momentum tensor and its spinorial counterpart $T_{ABA'B'}$ is defined by
\begin{align}
T_{ABA'B'}\equiv-2\mathrm{i}\left(\phi_A\nabla_{BB'}\phi_{A'}-\phi_{A'}\nabla_{BB'}\phi_{A}+\phi_B\nabla_{AA'}\phi_{B'}-\phi_{B'}\nabla_{AA'}\phi_{B} \right). \label{EMtensor}
\end{align}
In order to analyse the derivative term in the energy-momentum tensor and close the bootstrap argument, 
we follow the strategy in our previous paper \cite{PengXiaoning2501}, 
i.e. make the irreducible decomposition for $\nabla_{AA'}\phi_B$:
\begin{align*}
\nabla_{AA'}\phi_B=\nabla_{(A|A'|}\phi_{B)}+\frac{1}{2}\epsilon_{AB}\nabla_{CA'}\phi^C,
\end{align*}
then the equation of motion of the Weyl spinor leads to 0 for the second term on the right hand. 
Hence one defines the symmetric part be a new variable 
\begin{align}
\zeta_{ABA'}\equiv\nabla_{(A|A'|}\phi_{B)}. \label{Definitionzeta}
\end{align}
Its equation can be obtained by first commutating the covariant derivative to $\phi_A$:
\begin{align*}
\nabla_{AA'}\nabla_{BB'}\phi_C-\nabla_{BB'}\nabla_{AA'}\phi_C=
\epsilon_{A'B'}\Box_{AB}\phi_C+\epsilon_{AB}\Box_{A'B'}\phi_C
\end{align*}
where 
\begin{align*}
\Box_{AB}\equiv\nabla_{Q'(A}\nabla_{B)}^{\phantom{B)}Q'}, \quad
\Box_{A'B'}\equiv\nabla_{Q(A'}\nabla_{B')}^{\phantom{B')}Q}
\end{align*}
and 
\begin{align*}
\Box_{AB}\phi_C=\Psi_{ABCD}\phi^D-2\Lambda\phi_{(A}\epsilon_{B)C}=\Psi_{ABCD}\phi^D, \quad
\Box_{A'B'}\phi_C=\Phi_{CDA'B'}\phi^D,
\end{align*}
here $\Lambda=-\frac{R}{24}$, which vanishes in this model. 
$\Psi_{ABCD}$ is the spinorial counterpart of the Weyl curvature, 
$\Phi_{CDA'B'}$ is the spinorial counterpart of the trace-free Ricci tensor.
One then concludes that $\zeta_{ABA'}$ satisfies the following
\begin{align}
\nabla_{AA'}\zeta_{BCB'}-\nabla_{BB'}\zeta_{ACA'}=
\Psi_{DCAB}\epsilon_{A'B'}\phi^D+\Phi_{DCA'B'}\epsilon_{AB}\phi^D \label{EoMzeta}
\end{align}
Then with the definition of $\zeta_{ABA'}$ as well as the equation of motion for $\phi_A$, 
the energy-momentum tensor takes the form
\begin{align}
T_{ABA'B'}\equiv-2\mathrm{i}\left(\phi_A\bar\zeta_{B'A'B}-\phi_{A'}\zeta_{BAB'}+\phi_B\bar\zeta_{A'B'A}-\phi_{B'}\zeta_{ABA'} \right) \label{EMtensorAlt}
\end{align}
where $\bar\zeta_{A'B'A}$ is the conjugate of $\zeta_{ABA'}$.
The Einstein field equations \eqref{EinsteinFE} 
can be expressed in the spinorial form
\begin{align}
\Phi_{ABB'A'}
=\mathrm{i}\left(\phi_A\zeta_{B'A'B}-\phi_{A'}\zeta_{BAB'}+\phi_B\zeta_{A'B'A}-\phi_{B'}\zeta_{ABA'} \right).
\label{EFEspinor}
\end{align} 

\subsection{Main theorem}

In this section we show the main results of this paper:
\begin{theorem}
[\textbf{\em Existence result}]
Given a positive number~$\mathcal{I}$, there exists a sufficiently
large~$a_0=a_0(\mathcal{I})$ such that for~$a\geq a_0\geq0$ and
initial data 
\begin{align*}
\mathcal{I}_0\equiv\sum_{j=0}^1\sum_{i=0}^{15}\frac{1}{a^{\frac{1}{2}}}
||\mthorn^j(|u_{\infty}|\mathcal{D})^i(\sigma,\phi_0)||_{L^{2}(\mathcal{S}_{u_{\infty},v})}\leq\mathcal{I}
\end{align*}
along the outgoing initial null hypersurface~$u=u_{\infty}$, and
Minkowskian initial data along ingoing initial null
hypersurface~$v=0$. 
Then the Einstein-Weyl spinor system admits a unique solution in
\begin{align*}
\mathbb{D}=\{(u,v)|u_{\infty}\leq u\leq -a/4, \ \ 0\leq v\leq 1\}.
\end{align*}
\end{theorem}
Based on the existence theorem we have the following theorem on trapped-surface formation:
\begin{theorem}
[\textbf{\em Trapped surface formation}]
Given~$\mathcal{I}$, there exists a sufficiently large
$a=a(\mathcal{I})$ such that for~$a\geq a_0\geq0$ an initial data set
with
\begin{align*}
\mathcal{I}_0\equiv\sum_{j=0}^1\sum_{i=0}^{15}\frac{1}{a^{\frac{1}{2}}}
  ||\mthorn^j(|u_{\infty}|\mathcal{D})^i(\sigma,\phi_0)||_{L^{2}(\mathcal{S}_{u_{\infty},v})}
  \leq\mathcal{I}
\end{align*}
along the outgoing initial null hypersurface~$u=u_{\infty}$ and
Minkowskian initial data along the ingoing initial null hypersurface
$v=0$ such that
\begin{align*}
  \int_0^1\left(
  |u_{\infty}|^2(\sigma\bar\sigma
  +2\mathrm{i}(\phi_0\mthorn\bar\phi_0-\bar\phi_0\mthorn\phi_0))|\right)(u_{\infty},v') \mathrm{d}v'\geq a
\end{align*}
holds uniformly for any point on the initial outgoing null
hypersurface $u=u_{\infty}$, it follows that~$\mathcal{S}_{-a/4,1}$ is
a trapped surface.
\end{theorem}

\begin{figure}[t]
\centering
\includegraphics[width=0.6\textwidth]{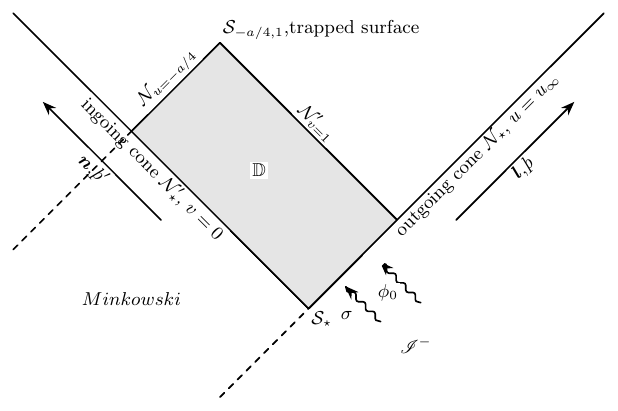}
\caption{Existence domain~$\mathbb{D}$ for the main theorem. 
We pose a characteristic problem with Minkowskian data on the incoming null hypersurface~$v=0$
and suitable outgoing data on~$u=u_\infty$ in the asymptotic region. For an appropriate choice of $a$, 
the top-vertex 2-surface~$\mathcal{S}_{-a/4,1}$ is a trapped surface.}
\label{Fig:ExistenceDomain}
\end{figure}

\begin{remark}
We work under double null foliation $(u,v)$ in the T-weight formalism; see Sec. \ref{GeometricSetup}. 
We make use of the Newman-Penrose frame~$\{\bml,\bmn,\bm{m},\bar{\bm{m}}\}$, 
$\sigma$ denotes the shear of outgoing light rays 
\begin{align*}
\Pi_{\mathcal{S}}(\widehat{\nabla_al_b})&\equiv2|\sigma| 
\end{align*}
where $\Pi_{\mathcal{S}}$ is the projection operator, 
$\widehat{\nabla_al_b}$ is the traceless part of $\nabla_al_b$. 
$\phi_0$ is the $o$-component of spinor field $\phi_A$, 
$\mthorn\phi_0$ is the derivative along $\bml$. 
The term $2\mathrm{i}(\phi_0\mthorn\bar\phi_0-\bar\phi_0\mthorn\phi_0)$ is $\Phi_{00}$ 
and represent the energy flux of spinor field.

\end{remark}

\subsection{Main difficulty}

Our argument is energy-based in a double-null foliation using T-weight framework. 
The point of the present analysis is to highlight that the spinor 
$\phi_A$ and its symmetric first derivative $\zeta_{ABA'}$ 
admit evolution/constraint equations whose signature is exactly compatible with 
the scale-invariant norms $L^2_{sc}$. 
This norm, whose definitions can be found in Sec. \ref{Preliminaries}, 
is designed to capture the peeling property by assigning each quantity a signature $s_2$. 
The equations for $\zeta_{ABA'}$ provide sharp, 
scale-critical control of $\phi_A$ in both the angular and null directions 
and remain aligned with the curvature hierarchy, 
allowing the Bianchi--matter couplings to be quantified and absorbed at the same critical scale.
This compatibility allows us track the $u$-weight and obtain a scale-critical result by the bootstrap, 
yielding closure without derivative loss for a quasilinear geometric system 
whose stress--energy contains derivative--field products.

In the Bianchi sector coupled to the Weyl spinor,
large borderline contributions involving $\zeta_4$ and $\zeta_5$ (see Eq. \ref{Definitionzeta}) 
arise at every derivative order once one commutes with angular derivatives. 
That is because pair $(\zeta_4,\zeta_5)$ satisfies
\begin{align*}
\mthorn'\zeta_4-\meth\zeta_5=&-2\mu\zeta_4+..., \\
\mthorn\zeta_4-\meth'\zeta_4=&-\mu\zeta_1+...
\end{align*}
and leads to the following energy estimate:
\begin{align*}
&\int_0^v\frac{1}{|u|^2}||\zeta_4||^2_{L^2_{sc}(\mathcal{S}_{u,v'})}
+\int_{u_{\infty}}^u\frac{a}{|u'|^4}||\zeta_5||^2_{L^2_{sc}(\mathcal{S}_{u',v})}
\lesssim\int_0^v\frac{1}{|u_{\infty}|^2}||\zeta_4||^2_{L^2_{sc}(\mathcal{S}_{u_{\infty},v'})}
+\int_{u_{\infty}}^u\frac{a}{|u'|^4}||\zeta_5||^2_{L^2_{sc}(\mathcal{S}_{u',0})} \nonumber\\
&+2\int_0^v\int_{u_{\infty}}^u\frac{a}{|u'|^4}
||\zeta_4||_{L^2_{sc}(\mathcal{S}_{u',v'})}
||\zeta_4\mu||_{L^2_{sc}(\mathcal{S}_{u',v'})} 
+2\int_0^v\int_{u_{\infty}}^u\frac{a}{|u'|^4}
||\zeta_5||_{L^2_{sc}(\mathcal{S}_{u',v'})}
||\zeta_1\mu||_{L^2_{sc}(\mathcal{S}_{u',v'})}+...
\end{align*}
in the scale-invariant norm. This means that $\zeta_4$ and $\zeta_5$ decay slower 
and borderline in the scale-invariant sense, i.e. we have 
\begin{align*}
|\zeta_4|\lesssim\frac{a^{\frac{1}{2}}}{|u|^2}, \quad
|\zeta_5|\lesssim\frac{a}{|u|^3}, \quad \Rightarrow \quad
||\{\zeta_4,\zeta_5\}||_{L^{\infty}_{sc}(\mathcal{S}_{u,v})}\sim
\frac{|u|}{a^{\frac{1}{2}}}.
\end{align*}
More discussions can be found in Prop. \ref{EnergyEstimatezeta45}. 
Hence terms $\zeta_4\bar\zeta_4$ and $\zeta_5\bar\phi_0\mu$ in the Bianchi identities for pair 
$(\Psi_2,\TiPsi_3)$ and $(\TiPsi_3,\Psi_4)$ produce logarithmic growth 
and obstruct scale-invariant closure under the bootstrap. 

To suppress the borderline accumulation at each order, 
we introduce the renormalized spinorial derivatives $\Tizeta_4$ and $\Tizeta_5$:
\begin{align*}
\Tizeta_4\equiv\zeta_4+\mu\phi_0, \quad
\Tizeta_5\equiv\zeta_5+2\mu\phi_1.
\end{align*}
On the one hand, from their equations 
\begin{align*}
\mthorn' \,\tilde\zeta_{4}- \meth\,\tilde\zeta_{5}
&= - 3\,\tilde\zeta_{4}\,\mu 
   - 2\,\phi_{1}\,\bigl(\meth\mu\bigr)+... ,\\
\mthorn \,\tilde\zeta_{5}- \meth'\,\tilde\zeta_{4}
&= 2\Psi_{2}\,\phi_{1}
   + 2\,\phi_{1}\,\bigl(\meth\pi\bigr) 
   - \phi_{0}\,\bigl(\meth'\mu\bigr)\,+... 
\end{align*} 
one can keep the signature bookkeeping compatible with scale-invariant norms 
and obtain 
\begin{align*}
&\int_0^v||\Tizeta_4||^2_{L^2_{sc}(\mathcal{S}_{u,v'})}
+\int_{u_{\infty}}^u\frac{a}{|u'|^2}||\Tizeta_5||^2_{L^2_{sc}(\mathcal{S}_{u',v})}
\lesssim\int_0^v||\Tizeta_4||^2_{L^2_{sc}(\mathcal{S}_{u_{\infty},v'})}
+\int_{u_{\infty}}^u\frac{a}{|u'|^2}||\Tizeta_5||^2_{L^2_{sc}(\mathcal{S}_{u',0})} 
+...
\end{align*}
which show that $\Tizeta_4$ and $\Tizeta_5$ are normal terms in the scale-invariant norm 
and lead to a sharper control of $\zeta_4$ and $\zeta_5$; see Prop. \ref{EnergyEstimateTizeta45}.
On the other hand, these new variables absorb the problematic couplings in the Bianchi identities.

The evolution equations for $\Tizeta_4$ and $\Tizeta_5$ contain curvature 
and higher-order derivatives of connection coefficients, 
so we do not attempt full top-order closure for these variables. 
Instead, we close them at the next-to-top order, uniformly in the hierarchy. 
Terms $\phi_j\meth\zeta_4$ and $\phi_j\meth\zeta_5$ in the Bianchi analysis cause 
no logarithmic accumulation even if $\zeta_4$, $\zeta_5$ 
themselves are not controlled sharply at the very top order. 
Based on that observation one can then close the bootstrap without derivative loss.

\subsection{Outline of the article}
In Section \ref{GeometricSetup}, we introduce the geometric setting, coordinate choice 
and equations in T-weight formalism. 
In Section \ref{Preliminaries}, we present the strategy of the proof and the basic estimates 
for the frame components. 
In Section \ref{L2estimate}, we estimate the $L^2(\mathcal{S})$ norm for non-top-order connections, 
$\phi_A$, $\zeta_{ABA'}$ and curvatures.
The analysis of top-order derivative of connections are shown in Section \ref{Elliptic}. 
Section \ref{EnergyEstimate} discusses the energy estimates of $\phi_A$, $\zeta_{ABA'}$ and curvatures. 
Section \ref{TrappedSurface} discusses the formation of trapped surfaces.

\subsection{Acknowledgements}
The calculations described in this article have been carried out in the suite xAct for Mathematica; see \cite{Jose14}. 
Peng Zhao was supported by the National Natural Science Foundation of China (Grant No. 12505056) and 
the start-up fund of Beijing Normal University at Zhuhai (Grant No. 310432114).
Xiaoning Wu was supported by the National Natural Science Foundation of China (Grant No. 12275350) 
and the International (Regional) Collaboration and Exchange Project (Grant No. W2421033).

\section{Geometric setup and Equations in T-weight formalism}
\label{GeometricSetup}
\subsection{Basic geometric setting}
Let $(\mathcal{M},\bm{g})$ be a 4-dimensional manifold with boundary: 
outgoing null edge $\mathcal{N}_{\star}$ and ingoing null edge $\mathcal{N}'_{\star}$. 
The intersection is denoted by $\mathcal{S}_{\star}=\mathcal{N}_{\star}\cap\mathcal{N}'_{\star}$. 
We assume the existence of the double null foliation $(u,v)$ in the future 
of $\mathcal{N}_{\star}\cup\mathcal{N}'_{\star}$, 
where the level sets $u$-surfaces $\mathcal{N}_u$ are outgoing null hypersurfaces and 
$\mathcal{N}'_v$ represent the ingoing null hypersurfaces 
with$\mathcal{N}_0=\mathcal{N}_{\star}$ and 
$\mathcal{N}'_0=\mathcal{N}'_{\star}$.
Define the region $\mathcal{D}_{u,v}$ via
\begin{align}
\mathcal{D}_{u,v}\equiv\bigcup_{0\leq v'\leq v, u_{\infty}\leq u'\leq u}\mathcal{S}_{u',v'}.
\end{align}
Denote $\mathcal{S}_{u,v}=\mathcal{N}_{u}\cap\mathcal{N}'_{v}$ be the 
spacelike topological 2-sphere. 
We also denote $\mathcal{N}_u(v_1,v_2)$ be the part of the hypersurface~$\mathcal{N}_u$ 
with $v_1\leq v\leq v_2$. Likewise~$\mathcal{N}'_v(u_1,u_2)$ has a similar
definition. 
One then adopt the same coordinate choice in \cite{HilValZha23} and 
construct a Newman-Penrose frame~$\{\bml,\bmn,\bm{m},\bar{\bm{m}}\}$ of the form
\begin{align}
  \bml=Q\bmpartial_v,\qquad
  \bmn=\bmpartial_u+C^{\mathcal{A}}\bmpartial_{\mathcal{A}}, \qquad
  \bmm=P^{\mathcal{A}} \bmpartial_{\mathcal{A}}, \label{framem}
\end{align}
where~$C^{\mathcal{A}}=0$ on~$\mathcal{N}'_{\star}$, and~$Q=1$ on~$\mathcal{N}_{\star}$.
Then we have the following properties for connection coefficients
\begin{subequations}
\begin{align}
& \kappa=\nu=\epsilon=0, \label{spinconnection1}\\
& \rho=\bar{\rho},\ \ \mu=\bar{\mu}, \label{spinconnection2}\\
& \tau=\bar\alpha+\beta \label{spinconnection3}
\end{align}
\end{subequations}
in the neighbourhood of~$\mathcal{D}_{u,v}$ and, furthermore, with
\begin{align*}
\gamma-\bar{\gamma}=0\ \ \ \textrm{on}\ \ \
\mathcal{D}_{u,v}\cap\mathcal{N}'_{\star}. 
\end{align*}
Moreover ons has the following equations for the frame coefficient~$Q$, $P^{\mathcal{A}}$
    and~$C^{\mathcal{A}}$:
\begin{subequations}
\begin{eqnarray}
&&DC^{\mathcal{A}}=(\bar{\tau}+\pi)P^{\mathcal{A}}+(\tau+\bar{\pi})\bar{P}^{\mathcal{A}}, \label{framecoefficient1} \\
&&DP^{\mathcal{A}}=\bar\rho P^{\mathcal{A}}+\sigma\bar{P}^{\mathcal{A}}, \label{framecoefficient2} \\
&&\Delta P^{\mathcal{A}}-\delta C^{\mathcal{A}}=-(\mu-\gamma+\bar{\gamma})P^{\mathcal{A}}-\bar\lambda\bar{P}^{\mathcal{A}}, \label{framecoefficient3} \\
&&\Delta Q=(\gamma+\bar{\gamma})Q,  \label{framecoefficient4}\\
&&\bar{\delta}P^{\mathcal{A}}-\delta\bar{P}^{\mathcal{A}}=(\alpha-\bar{\beta})P^{\mathcal{A}}-(\bar{\alpha}-\beta)\bar{P}^{\mathcal{A}}, \label{framecoefficient5} \\
&&\delta Q=(\tau-\bar{\pi})Q. \label{framecoefficient6}
\end{eqnarray}
\end{subequations}


\subsection{T-weight formalism and equations}
In this section we introduce the equations in the T-weight formalism 
which are designed to fit the PDE analysis. 
Based on the GHP formalism, we assign each quantity $f$ a so called T-weight $s(f)$ 
and introduce four new differential operators~$\meth$,
$\meth'$, $\mthorn$ and~$\mthorn'$
\begin{align*}
\meth f&\equiv\delta f+s(\beta-\bar\alpha)f, \quad
\meth' f \equiv\bar\delta f-s(\bar\beta-\alpha)f, \quad
\mthorn f\equiv Df+s(\epsilon-\bar\epsilon)f, \quad
\mthorn' f\equiv\Delta f+s(\gamma-\bar\gamma)f,
\end{align*}
and obtain that the norm of such derivative of T-weight quantities 
is independent of the spherical coordinates choice; see~\cite{HilValZha23}. 
Such quantity $f$ is related to a tensor $T(f)$ on $\mathcal{S}$ and one has that    
\begin{align*}
|\mathcal{D}^kf|^2\equiv\sum_{\alpha}|\mathcal{D}^{k_i}f|^2=|\nablasl^kT(f)|^2,
\end{align*}
where ~$\mathcal{D}^{k_i}f$ is a string of order~$k$ of the
operators~$\meth$ and~$\meth'$, and the sum over~$\alpha$ denotes all
such strings. This leads to the definition of norm on $\mathcal{S}$
\begin{align}
||\mathcal{D}^kf||^p_{L^p(\mathcal{S})}\equiv\int_{\mathcal{S}}|\mathcal{D}^kf|^p,\qquad 
||\mathcal{D}^kf||_{L^{\infty}(\mathcal{S})}\equiv\sup_{\mathcal{S}}|\mathcal{D}^kf|.
\end{align}
More properties can be found in~\cite{HilValZha23}.   

\smallskip
\noindent
\textbf{(1) Weyl spinor equations in T-weight formalism}

Given a spin basis $\{o,\iota\}$, the respective NP frame is defined by 
\begin{align*}
l^{AA'}\equiv o^A\bar{o}^{A'}, \quad n^{AA'}\equiv\iota^A\bar{\iota}^{A'}, \quad
m^{AA'}\equiv o^A\bar{\iota}^{A'}, \quad \bar{m}^{AA'}\equiv \iota^A\bar{o}^{A'}
\end{align*}
and the connection coefficients are defined as
\begin{align*}
\kappa=&o^ADo_A, \quad \epsilon=\iota^ADo_A, \quad
\pi=\iota^AD\iota_A, \quad \tau=o^A\Delta o_A, \quad
\gamma=\iota^A\Delta o_A, \quad \nu=\iota^A\Delta\iota_A, \\
\beta=&\iota^A\delta o_A, \quad \sigma=o^A\delta o_A, \quad
\mu=\iota^A\delta\iota_A, \quad \alpha=\iota^A\bar\delta o_A, \quad 
\rho=o^A\bar\delta o_A, \quad \lambda=\iota^A\bar\delta\iota_A
\end{align*}
where the NP derivatives are defined by 
\begin{align*}
D=l^{AA'}\nabla_{AA'}, \quad \Delta=n^{AA'}\nabla_{AA'}, \quad 
\delta=m^{AA'}\nabla_{AA'}, \quad \bar\delta=\bar{m}^{AA'}\nabla_{AA'}.
\end{align*}

Define the components of $\phi^A$ with respect to the spin basis $\{o,\iota\}$ by 
\begin{align*}
\phi_0\equiv\phi_Ao^A, \quad \phi_1\equiv\phi_A\iota^A.
\end{align*}
Define the components of $\zeta_{ABA'}$ with respect to the spin basis $\{o,\iota\}$ by 
\begin{align}
\zeta_0&\equiv\zeta_{ABA'}o^Ao^B\bar{o}^{A'}, \quad
\zeta_1\equiv\zeta_{ABA'}o^A\iota^B\bar{o}^{A'}, \quad
\zeta_2\equiv\zeta_{ABA'}\iota^A\iota^B\bar{o}^{A'}, \nonumber\\
\zeta_3&\equiv\zeta_{ABA'}o^Ao^B\bar{\iota}^{A'}, \quad
\zeta_4\equiv\zeta_{ABA'}o^A\iota^B\bar{\iota}^{A'}, \quad
\zeta_5\equiv\zeta_{ABA'}\iota^A\iota^B\bar{\iota}^{A'}. \label{Definitionzeta}
\end{align}
The T-weight of such quantities are listed as follows:
\begin{align*}
&s=-\frac{3}{2}:\quad \zeta_3 \\
&s=-\frac{1}{2}:\quad \phi_0,\quad \zeta_0,\quad \zeta_4, \\
&s=\frac{1}{2}:\quad \phi_1, \quad \zeta_1, \quad \zeta_5, \\
&s=\frac{3}{2}:\quad \zeta_2,
\end{align*}

We obtain the T-weight equations for $\nabla_{AA'}\phi^A=0$:
\begin{subequations}
\begin{align}
\mthorn\,\phi_{1}
&= \phi_{0}\,\pi
  + \phi_{1}\,\rho
  - \frac{\phi_{0}\,\bar\tau}{2}
  + \meth'\,\phi_{0}, \label{thornphi1} \\
\mthorn'\,\phi_{0}
&= \frac{\omega\,\phi_{0}}{2}
  - \phi_{0}\,\mu
  - \frac{\phi_{1}\,\tau}{2}
  + \meth\,\phi_{1}. \label{thornprimephi0}
\end{align}
\end{subequations}
The equations of the definition of $\zeta_{ABA'}$ \eqref{Definitionzeta} are shown in appendix. \ref{DefinitionEqzeta}. 
The equations for $\zeta_{ABA'}$, i.e. \eqref{EoMzeta} are presented in appendix \ref{Equationzeta}. 
The schematic form are
\begin{align*}
&\{\mthorn,\mthorn'\}\zeta-\meth\zeta=\Psi\phi+\Gamma\zeta+\zeta\phi^2, \\
&\meth'\zeta-\meth\zeta=\Gamma\zeta+\zeta\phi^2+\Psi\phi,
\end{align*}
where $\zeta$ denote the components of $\zeta_{ABA'}$, $\phi$ denote the components of $\phi_A$, 
$\Gamma$ denote the connection coefficients and $\Psi$ denote the curvatures.
As we have already shown in \cite{PengXiaoning2501}, such equations can be 
divided into two groups, one group does not contain Weyl curvature \ref{EquationzetaNoCurv}, 
and another group does contain; see \ref{EquationzetaCurv}.

\smallskip
\noindent
\textbf{(2) The Einstein field equation} 

Expand the Einstein field equation \eqref{EFEspinor} with the fields $\zeta_{ABA'}$, 
one obtains the following

\begin{subequations}
\begin{align}
\Phi_{00} &= 2\mathrm{i}\bigl(\bar\zeta_0\phi_0-\zeta_0\bar\phi_0\bigr),\label{EDeq1}\\
\Phi_{01} &= \mathrm{i}\bigl(2\bar\zeta_1\phi_0-\zeta_3\bar\phi_0-\zeta_0\bar\phi_1\bigr),\label{EDeq2}\\
\Phi_{02} &= 2\mathrm{i}\bigl(\bar\zeta_2\phi_0-\zeta_3\bar\phi_1\bigr),\label{EDeq3}\\
\Phi_{11} &= \mathrm{i}\bigl(\bar\zeta_4\phi_0-\zeta_4\bar\phi_0+\bar\zeta_1\phi_1
-\zeta_1\bar\phi_1\bigr),\label{EDeq4}\\
\Phi_{12} &= \mathrm{i}\bigl(\bar\zeta_5\phi_0+\bar\zeta_2\phi_1-2\zeta_4\bar\phi_1\bigr),\label{EDeq5}\\
\Phi_{22} &= 2\mathrm{i}\bigl(\bar\zeta_5\phi_1-\zeta_5\bar\phi_1\bigr),\label{EDeq6}\\
\Lambda &= 0. \label{EDeq7}
\end{align}
\end{subequations}

\smallskip
\noindent
\textbf{(3) The structure equations, Bianchi identities and the renormalised Weyl curvature} 
The structure equations have the schematic form
\begin{align*}
\{\mthorn,\mthorn'\}\Gamma-\meth\Gamma=\zeta\phi+\Gamma\Gamma+\Psi.
\end{align*}
The fully explicit expressions can be found in the appendix \ref{StructureEq}.
For the Bianchi identities, we introduce the following two 
renormalised Weyl curvature components which are defined by
\begin{subequations}
\begin{align}
\tilde\Psi_1\equiv\Psi_1-\Phi_{01}, \quad
\tilde\Psi_3\equiv\Psi_3-\Phi_{21},
\end{align}
\end{subequations}
with which one can formulate a Hodge system.
The Bianchi Identities have the schematic form:
\begin{align*}
\{\mthorn,\mthorn'\}\Psi_i-\{\meth,\meth'\}\Psi_j=&
+\zeta\phi\Gamma+\phi\meth\zeta+\Psi\phi^2+\zeta\phi^3+\zeta^2+\Gamma\Psi_k. 
\end{align*}
The fully explicit equations are shown in \ref{BianchiIdentity}. 

\subsection{Further ancillary fields}
Note that we prescribe the Minkowski data on $\mathcal{N}'_{\star}$ on which the two null expansions are 
\begin{align*}
&\theta_{\bml}(u,v=0)=-\rho-\bar\rho=-2\rho=-2\cdot\frac{1}{u-v}=-\frac{2}{u}; \\
&\theta_{\bmn}(u,v=0)=\mu+\bar\mu=2\mu=2\cdot\frac{1}{u-v}=\frac{2}{u},
\end{align*}
we anticipate that, in $\mathcal{D}_{u,v}$, the ingoing expansion $\mu$ change slightly, 
hence it is convenient do define the following new quantity
\begin{align}
\tilde\mu\equiv\mu-\frac{1}{u}. \label{DefinitionTimu}
\end{align}
whose transport equations are 
\begin{subequations}
\begin{align}
&\mthorn\tilde\mu-\meth\pi-\rho\tilde\mu=-\rho\tilde\mu+\mu\rho+\Psi_2+\lambda\sigma+\pi\bar\pi 
\label{Thorntildemu},\\
&\mthorn'\tilde\mu+2\mu\tilde\mu=\tilde\mu^2-\ulomega\mu-\lambda\bar\lambda
-2\mathrm{i}\bigl(\bar\zeta_5\phi_1-\zeta_5\bar\phi_1\bigr) \label{Thornprimetildemu}.
\end{align}
\end{subequations}

Another quantity is related to the derivative of frame coefficient $Q$, we define 
\begin{align}
\ulvarrho\equiv\mthorn\log Q \label{Definitionulvarrho}
\end{align}
to track the property of $Q$ as we did in~\cite{HilValZha23}. 
The equation for $\ulvarrho$ is 
\begin{align}
\mthorn'\ulvarrho = \Psi_2 + \bar\Psi_2
+ \ulvarrho\,\ulomega
+ 2\,\mathrm{i}\,\bar\zeta_4\,\phi_0
- 2\,\mathrm{i}\,\zeta_4\,\bar\phi_0
+ 2\,\mathrm{i}\,\bar\zeta_1\,\phi_1
- 2\,\mathrm{i}\,\zeta_1\,\bar\phi_1 
+ 2\,\pi\,\bar\pi
+ 2\,\pi\,\tau
+ 2\,\bar\pi\,\bar\tau. \label{thornprimeulvarrho}
\end{align}

In order to eliminate the logarithmic contribution 
from terms containing $\zeta_4$ and $\zeta_5$ in the energy estimate 
for Bianchi identities, we define the following renormalized quantities to absorb such terms:
\begin{align}
\Tizeta_4\equiv\zeta_4+\mu\phi_0, \quad
\Tizeta_5\equiv\zeta_5+2\mu\phi_1.
\end{align}
Making use of the Equation of motion for $\phi_A$ and the structure equation of $\mu$, 
one can then obtain the equations for pair $(\Tizeta_4,\Tizeta_5)$:
\begin{subequations}
\begin{align}
\mthorn' \,\tilde\zeta_{4}- \meth\,\tilde\zeta_{5}
&= -\,\mathrm{i}\,\bar\zeta_{2}\,\phi_{1}^{2}
   - \frac{\ulomega\,\tilde\zeta_{4}}{2}
   + 2\,\mathrm{i}\,\phi_{1}\,\bar\phi_{1}\,\tilde\zeta_{4}
   - \mathrm{i}\,\phi_{0}\,\phi_{1}\,\bar{\tilde\zeta}_{5}
   - \zeta_{2}\,\bar\lambda
   - \phi_{0}\,\lambda\,\bar\lambda
   - 3\,\tilde\zeta_{4}\,\mu \nonumber\\
&\quad
 + \frac{\tilde\zeta_{5}\,\tau}{2}
   - \phi_{1}\,\mu\,\tau 
   - 2\,\phi_{1}\,\bigl(\meth\mu\bigr)\,, \label{thornprimeTizeta4}\\
\mthorn \,\tilde\zeta_{5}- \meth'\,\tilde\zeta_{4}
&= 2\Psi_{2}\,\phi_{1}
   + \mathrm{i}\,\bar\zeta_{1}\,\phi_{1}^{2}
   + \mathrm{i}\,\zeta_{2}\,\phi_{0}\,\bar\phi_{1}
   - \mathrm{i}\,\zeta_{1}\,\phi_{1}\,\bar\phi_{1}
   - \mathrm{i}\,\bar\phi_{0}\,\phi_{1}\,\tilde\zeta_{4}
   - \mathrm{i}\,\phi_{0}\,\phi_{1}\,\tilde\zeta_{4} \nonumber\\
&\quad
+ \mathrm{i}\,\phi_{0}\,\bar\phi_{0}\,\tilde\zeta_{5}
   - \zeta_{3}\,\lambda
   + 2\,\tilde\zeta_{4}\,\pi
   + \zeta_{2}\,\bar\pi
   + 2\,\phi_{1}\,\pi\,\bar\pi
   + \tilde\zeta_{5}\,\rho
   + \phi_{1}\,\mu\,\rho \nonumber \\
&\quad
+ 2\,\phi_{1}\,\lambda\,\sigma
   + \frac{\tilde\zeta_{4}\,\bar\tau}{2}
   - \phi_{0}\,\mu\,\tau
   + 2\,\phi_{1}\,\bigl(\meth\pi\bigr) 
   - \phi_{0}\,\bigl(\meth'\mu\bigr)\,. \label{thornTizeta5}
\end{align}
\end{subequations}

\section{Preliminaries}
\label{Preliminaries}
In this section we introduce the norms, strategy of proofs and essential mathematical tools.
\subsection{Signatures and scale-invariant norms}
Following the assumptions in \cite{An2022,An202209,HilValZha23}, 
we assume that on $\mathcal{N}_{\star}$, the connection
coefficients satisfy asymptotical behaviour 
\begin{align*}
&|\rho|\lesssim\frac{1}{|u_{\infty}|}, \quad
|\sigma|\lesssim\frac{a^{\frac{1}{2}}}{|u_{\infty}|}, \quad
|\mu|\lesssim\frac{1}{|u_{\infty}|}, \quad
|\Timu|=|\mu+\frac{1}{|u_{\infty}|}|\lesssim\frac{1}{|u_{\infty}|^2}, \quad
|\lambda|\lesssim\frac{a^{\frac{1}{2}}}{|u_{\infty}|^2}, \\
&|\ulomega|\lesssim\frac{a}{|u_{\infty}|^3}, \quad
|\tau|\lesssim\frac{a^{\frac{1}{2}}}{|u_{\infty}|^2},\quad
|\pi|\lesssim\frac{a^{\frac{1}{2}}}{|u_{\infty}|^2},
\end{align*}
and that for the components of the Weyl tensor we assume
\begin{align*}
|\Psi_0|\lesssim\frac{a^{\frac{1}{2}}}{|u_{\infty}|}, \quad
|\Psi_1|\lesssim\frac{a^{\frac{1}{2}}}{|u_{\infty}|^2}, \quad
|\Psi_2|\lesssim\frac{a}{|u_{\infty}|^3}, \quad
|\Psi_3|\lesssim\frac{a^{\frac{3}{2}}}{|u_{\infty}|^4}, \quad
|\Psi_4|\lesssim\frac{a^2}{|u_{\infty}|^5}.
\end{align*}

\medskip
For the spinor components we assume
\begin{align*}
|\phi_0|\lesssim\frac{a^{\frac{1}{2}}}{|u_{\infty}|}, \quad
|\phi_1|\lesssim\frac{a}{|u_{\infty}|^2},
\end{align*}
and for the components of $\zeta_{ABA'}$ we assume
\begin{align*}
|\zeta_0|\lesssim&\frac{a^{\frac{1}{2}}}{|u_{\infty}|}, \quad
|\zeta_1|\lesssim\frac{a^{\frac{1}{2}}}{|u_{\infty}|^2}, \quad
|\zeta_2|\lesssim\frac{a}{|u_{\infty}|^3}, \\ 
|\zeta_3|\lesssim&\frac{a^{\frac{1}{2}}}{|u_{\infty}|^2}, \quad
|\zeta_4|\lesssim\frac{a^{\frac{1}{2}}}{|u_{\infty}|^2}, \quad
|\zeta_5|\lesssim\frac{a}{|u_{\infty}|^3}, \\ 
\end{align*}
From the definition of the renormalised Weyl curvature $\TiPsi_1$ and $\TiPsi_3$, we then have 
\begin{align*}
|\TiPsi_1|=&|\Psi_1-2\mathrm{i}\bar\zeta_1\phi_0+\mathrm{i}\zeta_3\bar\phi_0+\mathrm{i}\zeta_0\bar\phi_1|
\lesssim\frac{a^{\frac{1}{2}}}{|u_{\infty}|^2}, \\
|\TiPsi_3|=&|\Psi_3-2\mathrm{i}\bar\zeta_4\phi_1+\mathrm{i}\zeta_5\bar\phi_0-\mathrm{i}\bar\zeta_0\phi_1|
\lesssim\frac{a^{\frac{3}{2}}}{|u_{\infty}|^4}
\end{align*}
which show that the renormalised Weyl curvature have the same behaviour as $\Psi_1$ and $\Psi_3$. 

For the renormalized $\Tizeta_4$ and $\Tizeta_5$ we assume
\begin{align*}
\Tizeta_4\equiv\zeta_4+\mu\phi_0, \quad
\Tizeta_5\equiv\zeta_5+2\mu\phi_1.
\end{align*}
We assume
\begin{align*}
|\Tizeta_4|\lesssim\frac{a}{|u_{\infty}|^3}, \quad
|\Tizeta_5|\lesssim\frac{a^{\frac{3}{2}}}{|u_{\infty}|^4}.
\end{align*}

\begin{remark}
A single left–handed Weyl spinor admits no nontrivial spherical symmetry; 
nevertheless, our data produce a trapped $2$-sphere that is quantifiably close to round.
This can be obtained by the estimate of the Gaussian curvature $K$ at the Minkowskian level. 
In the existence theorem we show that $K\sim|u|^{-2}$ in $L^{\infty}_{sc}$ and $L^{2}_{sc}$, 
and obtain uniform bounds on $\mathcal{D}^iK$; see Prop. \ref{L2Weyl} and Remark \ref{AdvancedGauss}. 
Consequently, the terminal trapped sphere deviates from a round metric by at most $O(a^{-\frac{1}{4}})$ 
in the scale-invariant norms. 

This near-sphericity reflects the low-angular-momentum character of the initial flux: 
in Newman–Penrose variables, $\Phi_{00}=T_{ab}l^al^b$ measures outgoing null energy density, 
while $\Phi_{01}=T_{ab}l^am^b$ encodes tangential momentum/rotation flux. 
In our model these read
\begin{align*}
\Phi_{00} = 2\mathrm{i}\bigl(\bar\zeta_0\phi_0-\zeta_0\bar\phi_0\bigr),\quad
\Phi_{01} = \mathrm{i}\bigl(2\bar\zeta_1\phi_0-\zeta_3\bar\phi_0-\zeta_0\bar\phi_1\bigr),
\end{align*}
and the initial data satisfies $\|\Phi_{01}\|_{L^2(\mathcal{S})}\ll \|\Phi_{00}\|_{L^2(\mathcal{S})}$, 
our existence theorem shows that such relation is maintained in the working region 
and hence amounts to a low-angular momentum regime with predominantly radial spinor flux, 
see Remark \ref{low-angular momentum argument}. 
This bias drives focusing without appreciable twist, consistent with the Gaussian curvature control 
and the formation of an almost round trapped surface.
\end{remark}

Follow the basic design in An's work \cite{An2022}, 
we introduce the signature for decay rates to each quantity in the following

\begin{tabular}{|c|c|c|c|c|c|c|c|c|c|c|c|c|c|}
\hline
~& $\Psi_0$ & $\Psi_1$ & $\Psi_2$ & $\Psi_3$ & $\Psi_4$ 
& $\rho,\sigma$ & $\omega$ & $\vartheta$ &$\tau$ &$\pi$ 
&$ \mu,\lambda, \tilde\mu$ & $\ulomega$ & $\ulchi$ \\
\hline
$s_2$ &0&$1/2$ &1&$3/2$&2&0&0&$1/2$&$1/2$&$1/2$&1&1&0\\
\hline
\end{tabular}

\begin{tabular}{|c|c|c|c|c|c|c|c|c|c|c|c|c|c|c|c|c|c|}
\hline
~&$\phi_0,\zeta_{0}$ &$\phi_{1}$,$\zeta_1$,$\zeta_3$&$\zeta_{2}$,$\zeta_{4}$,$\Tizeta_4$&$\zeta_5$,$\Tizeta_5$\\
\hline
$s_2$&0&$1/2$&1& $3/2$\\
\hline
\end{tabular}

\bigskip
\noindent
The signature is required to satisfy
 \begin{align*}
&s_2(\mthorn^i\mthorn'^j\{\meth,\meth'\}^kf)=s_2(f)+0\times i+1\times j+\tfrac{1}{2}\times k , \\
&s_2(f_1\cdot f_2)=s_2(f_1)+s_2(f_2).
\end{align*}

\medskip
We define the following \emph{scale-invariant norms} for each quantity $f$ with its signature $s_2(f)$:
 \begin{align*}
&||f||_{L^{\infty}_{sc}(\mathcal{S}_{u,v})}\equiv a^{-s_2(f)}|u|^{2s_2(f)+1}||f||_{L^{\infty}(\mathcal{S}_{u,v})}, \\
&||f||_{L^2_{sc}(\mathcal{S}_{u,v})}\equiv a^{-s_2(f)}|u|^{2s_2(f)}||f||_{L^2(\mathcal{S}_{u,v})},\\
&||f||_{L^1_{sc}(\mathcal{S}_{u,v})}\equiv a^{-s_2(f)}|u|^{2s_2(f)-1}||f||_{L^1(\mathcal{S}_{u,v})}.
 \end{align*}
 The above norms satisfy the H\"older inequalities
 \begin{align}
&||f_1\cdot f_2||_{L^2_{sc}(\mathcal{S}_{u,v})}\leq\frac{1}{|u|}||f_1||_{L^{\infty}_{sc}(\mathcal{S}_{u,v})}||f_2||_{L^2_{sc}(\mathcal{S}_{u,v})},\label{L2holder}\\
&||f_1\cdot f_2||_{L^1_{sc}(\mathcal{S}_{u,v})}\leq\frac{1}{|u|}||f_1||_{L^{\infty}_{sc}(\mathcal{S}_{u,v})}||f_2||_{L^1_{sc}(\mathcal{S}_{u,v})},\label{L1holder}\\
&||f_1\cdot f_2||_{L^1_{sc}(\mathcal{S}_{u,v})}\leq\frac{1}{|u|}||f_1||_{L^2_{sc}(\mathcal{S}_{u,v})}||f_2||_{L^2_{sc}(\mathcal{S}_{u,v})}. \label{L1holderalt}
 \end{align}
Thus, it follows that making use of these norms that if all terms are
order $1$, then the nonlinear terms can be treated as lower order
terms if $|u|\gg1$. Finally, we define the following scale-invariant norms along the lightcones:
\begin{align*}
||f||^2_{L^2_{sc}(\mathcal{N}_u(0,v))}&
\equiv\int_0^v||f||^2_{L^2_{sc}(\mathcal{S}_{u,v'})}\mathrm{d}v',\\
||f||^2_{L^2_{sc}(\mathcal{N}'_v(u_{\infty},u))}&
\equiv\int_{u_{\infty}}^u\frac{a}{|u'|^2}||f||^2_{L^2_{sc}(\mathcal{S}_{u',v})}\mathrm{d}u'.
\end{align*}

\subsection{Bootstrap norms}
In this section we introduce the norms that will be used to set up our 
main bootstrap argument. In the following, let $\hat\Gamma\equiv \{\rho,\tau,\pi,\ulomega,\ulchi\}$, 
$\hat\Psi=\{\TiPsi_1,\Psi_2,\TiPsi_3,\Psi_4\}$.

(i) Norms at spheres of constant $u$, $v$.

 For $0\leq i\leq6$, we define
 \begin{align*}
\Gamma_{i,\infty}(u,v)\equiv&\frac{1}{a^{\frac{1}{2}}}||(a^{\frac{1}{2}}\mathcal{D})^i\sigma||_{L^{\infty}_{sc}(\mathcal{S}_{u,v})}
+||(a^{\frac{1}{2}}\mathcal{D})^i\hat\Gamma||_{L^{\infty}_{sc}(\mathcal{S}_{u,v})}
+\frac{a^{\frac{1}{2}}}{|u|}||(a^{\frac{1}{2}}\mathcal{D})^i\lambda||_{L^{\infty}_{sc}(\mathcal{S}_{u,v})} \\
&+\frac{a}{|u|^2}||(a^{\frac{1}{2}}\mathcal{D})^i\mu||_{L^{\infty}_{sc}(\mathcal{S}_{u,v})}
+\frac{a}{|u|}||(a^{\frac{1}{2}}\mathcal{D})^i\tilde\mu||_{L^{\infty}_{sc}(\mathcal{S}_{u,v})},
\end{align*}

 \begin{align*}
\Psi_{i,\infty}(u,v)\equiv\frac{1}{a^{\frac{1}{2}}}||(a^{\frac{1}{2}}\mathcal{D})^i\Psi_0||_{L^{\infty}_{sc}(\mathcal{S}_{u,v})}
+||(a^{\frac{1}{2}}\mathcal{D})^i\hat\Psi||_{L^{\infty}_{sc}(\mathcal{S}_{u,v})},
\end{align*}

 \begin{align*}
\phi_{i,\infty}(u,v)\equiv\frac{1}{a^{\frac{1}{2}}}||(a^{\frac{1}{2}}\mathcal{D})^i\phi_{0}||_{L^{\infty}_{sc}(\mathcal{S}_{u,v})}
+\frac{1}{a^{\frac{1}{2}}}||(a^{\frac{1}{2}}\mathcal{D})^i\phi_{1}||_{L^{\infty}_{sc}(\mathcal{S}_{u,v})},
 \end{align*}
 
 \begin{align*}
\zeta_{i,\infty}(u,v)\equiv&\frac{1}{a^{\frac{1}{2}}}||(a^{\frac{1}{2}}\mathcal{D})^i\zeta_0||_{L^{\infty}_{sc}(\mathcal{S}_{u,v})}
+||(a^{\frac{1}{2}}\mathcal{D})^i\{\zeta_1,\zeta_2,\zeta_3,\Tizeta_4,\Tizeta_5\}||_{L^{\infty}_{sc}(\mathcal{S}_{u,v})} \\
&+\frac{a^{\frac{1}{2}}}{|u|}||(a^{\frac{1}{2}}\mathcal{D})^i\{\zeta_4,\zeta_5\}||_{L^{\infty}_{sc}(\mathcal{S}_{u,v})}.
 \end{align*}
  
For $0\leq i\leq10$, we define
 \begin{align*}
\Gamma_{i,2}(u,v)\equiv&
\frac{1}{a^{\frac{1}{2}}}||(a^{\frac{1}{2}}\mathcal{D})^i\sigma||_{L^2_{sc}(\mathcal{S}_{u,v})}
+||(a^{\frac{1}{2}}\mathcal{D})^i\hat\Gamma||_{L^2_{sc}(\mathcal{S}_{u,v})}
+\frac{a^{\frac{1}{2}}}{|u|}||(a^{\frac{1}{2}}\mathcal{D})^i\lambda||_{L^2_{sc}(\mathcal{S}_{u,v})} \\
&+\frac{a}{|u|^2}||(a^{\frac{1}{2}}\mathcal{D})^i\mu||_{L^2_{sc}(\mathcal{S}_{u,v})}
+\frac{a}{|u|}||(a^{\frac{1}{2}}\mathcal{D})^i\tilde\mu||_{L^2_{sc}(\mathcal{S}_{u,v})},
\end{align*}

\begin{align*}
\phi_{i,2}(u,v)\equiv\frac{1}{a^{\frac{1}{2}}}||(a^{\frac{1}{2}}\mathcal{D})^i\phi_{0}||_{L^2_{sc}(\mathcal{S}_{u,v})}
+\frac{1}{a^{\frac{1}{2}}}||(a^{\frac{1}{2}}\mathcal{D})^i\phi_{1}||_{L^2_{sc}(\mathcal{S}_{u,v})},
\end{align*}

 \begin{align*}
\zeta_{i,2}(u,v)\equiv\frac{1}{a^{\frac{1}{2}}}||(a^{\frac{1}{2}}\mathcal{D})^i\zeta_0||_{L^{2}_{sc}(\mathcal{S}_{u,v})}
+||(a^{\frac{1}{2}}\mathcal{D})^i\{\zeta_1,\zeta_2,\zeta_3\}||_{L^{2}_{sc}(\mathcal{S}_{u,v})}
+\frac{a^{\frac{1}{2}}}{|u|}||(a^{\frac{1}{2}}\mathcal{D})^i\{\zeta_4,\zeta_5\}||_{L^{2}_{sc}(\mathcal{S}_{u,v})}.
 \end{align*}
 
For $0\leq i\leq9$, we define
\begin{align*}
\Psi_{i,2}(u,v)\equiv\frac{1}{a^{\frac{1}{2}}}||(a^{\frac{1}{2}}\mathcal{D})^i\Psi_0||_{L^2_{sc}(\mathcal{S}_{u,v})}
+||(a^{\frac{1}{2}}\mathcal{D})^i\hat\Psi||_{L^2_{sc}(\mathcal{S}_{u,v})},
\end{align*}
 \begin{align*}
\Tizeta_{i,2}(u,v)\equiv
||(a^{\frac{1}{2}}\mathcal{D})^i\{\Tizeta_4,\Tizeta_5\}||_{L^{2}_{sc}(\mathcal{S}_{u,v})}.
 \end{align*}
 
(ii) Norms at null hypersurfaces
 
For $0\leq i\leq10$, we define
\begin{align*}
\bmPsi_{i}(u,v)\equiv\frac{1}{a^{\frac{1}{2}}}||(a^{\frac{1}{2}}\mathcal{D})^i\Psi_0||_{L^2_{sc}(\mathcal{N}_u(0,v))}
+||(a^{\frac{1}{2}}\mathcal{D})^i\{\TiPsi_1,\Psi_2,\TiPsi_3\}||_{L^2_{sc}(\mathcal{N}_u(0,v))},
\end{align*}

\begin{align*}
\underline\bmPsi_{i}(u,v)\equiv\frac{1}{a^{\frac{1}{2}}}||(a^{\frac{1}{2}}\mathcal{D})^i\TiPsi_1||_{L^2_{sc}(\mathcal{N}'_v(u_{\infty},u))}
+||(a^{\frac{1}{2}}\mathcal{D})^i\{\Psi_2,\TiPsi_3,\Psi_4\}||_{L^2_{sc}(\mathcal{N}'_v(u_{\infty},u))}.
\end{align*}

For $0\leq i\leq11$, we define
\begin{align*}
\bmphi_{i}(u,v)\equiv\frac{1}{a^{\frac{1}{2}}}||(a^{\frac{1}{2}})^{i-1}\mathcal{D}^i\phi_{0}||_{L^2_{sc}(\mathcal{N}_u(0,v))}, \quad
\underline\bmphi_{i}(u,v)\equiv\frac{1}{a^{\frac{1}{2}}}||(a^{\frac{1}{2}})^{i-1}\mathcal{D}^i\phi_{1}||_{L^2_{sc}(\mathcal{N}'_v(u_{\infty},u))},
\end{align*}

\begin{align*}
\bmzeta_{i}(u,v)\equiv&\frac{1}{a^{\frac{1}{2}}}||(a^{\frac{1}{2}})^{i-1}\mathcal{D}^i\zeta_0||_{L^2_{sc}(\mathcal{N}_u(0,v))}
+||(a^{\frac{1}{2}})^{i-1}\mathcal{D}^i\{\zeta_1,\zeta_3\}||_{L^2_{sc}(\mathcal{N}_u(0,v))} \\
&+||\frac{1}{|u|}(a^{\frac{1}{2}})^{i}\mathcal{D}^i\zeta_4||_{L^2_{sc}(\mathcal{N}_u(0,v))},
\end{align*}
\begin{align*}
\underline\bmzeta_{i}(u,v)\equiv&\frac{1}{a^{\frac{1}{2}}}||(a^{\frac{1}{2}})^{i-1}\mathcal{D}^i\{\zeta_1,\zeta_3\}||_{L^2_{sc}(\mathcal{N}'_v(u_{\infty},u))}
+||(a^{\frac{1}{2}})^{i-1}\mathcal{D}^i\zeta_2||_{L^2_{sc}(\mathcal{N}'_v(u_{\infty},u))}\\
&+||\frac{1}{|u'|}(a^{\frac{1}{2}})^{i}\mathcal{D}^i\zeta_5||_{L^2_{sc}(\mathcal{N}'_v(u_{\infty},u))}.
\end{align*}

Here
\begin{align*}
||\frac{1}{|u|}(a^{\frac{1}{2}})^{i}\mathcal{D}^i\zeta_4||_{L^2_{sc}(\mathcal{N}_u(0,v))}\equiv
\frac{1}{|u|}\left(\int_0^v||(a^{\frac{1}{2}}\mathcal{D})^i\zeta_4||^2_{L^2_{sc}(\mathcal{S}_{u,v'})} \right)^{\frac{1}{2}},
\end{align*}

\begin{align*}
||\frac{1}{|u'|}(a^{\frac{1}{2}})^{i}\mathcal{D}^i\zeta_5||_{L^2_{sc}(\mathcal{N}'_v(u_{\infty},u))}
\equiv\left(\int_{u_{\infty}}^u\frac{a}{|u'|^4}
||(a^{\frac{1}{2}}\mathcal{D})^i\zeta_5||^2_{L^2_{sc}(\mathcal{S}_{u,v'})} \right)^{\frac{1}{2}}.
\end{align*}

Moreover for $0\leq i\leq10$, define
\begin{align*}
\tilde{\bmzeta}_{i}(u,v)\equiv&||(a^{\frac{1}{2}})^{i}\mathcal{D}^i
\Tizeta_4||_{L^2_{sc}(\mathcal{N}_u(0,v))}, \quad
\tilde{\underline\bmzeta}_{i}(u,v)\equiv||(a^{\frac{1}{2}})^{i}\mathcal{D}^i\Tizeta_5||_{L^2_{sc}(\mathcal{N}'_v(u_{\infty},u))}.
\end{align*}

\begin{remark}
We denote the supremum over $u$ and $v$ of the above defined bootstrap norms.
\begin{align*}
&\bmGamma\equiv\sum_{i\leq6}(\Gamma_{i,\infty}+\Psi_{i,\infty}+\phi_{i,\infty}+\zeta_{i,\infty})
+\sum_{i\leq10}(\Gamma_{i,2}+\phi_{i,2}+\zeta_{i,2})+\sum_{i\leq9}(\Psi_{i,2}+\tilde{\zeta}_{i,2}), \\
&\bm{\Psi}\equiv\sum_{i\leq10}(\Psi_{i}+\underline\Psi_{i}), \quad 
\bm{\phi}\equiv\sum_{i\leq11}(\bm\phi_{i}+\underline{\bm\phi}_{i}), \quad
\bm{\zeta}\equiv\sum_{i\leq11}(\bm\zeta_{i}+\underline{\bm\zeta}_{i})+
\sum_{i\leq10}(\bm\Tizeta_{i}+\underline{\bm\Tizeta}_{i}), 
\end{align*}
and let $\bmGamma_0$, $\bm\Psi_0$, $\bm\phi_0$ and $\bm\zeta_0$ denote the norm on the initial light cones.
We also make use of the \emph{initial data quantity}
\begin{align*}
\mathcal{I}_0\equiv\sup_{0\leq v\leq1}\mathcal{I}_0(v), 
\end{align*}
where
\begin{align*}
\mathcal{I}_0(v)\equiv\sum_{j=0}^1\sum_{i=0}^{15}\frac{1}{a^{\frac{1}{2}}}
||\mthorn^j(|u_{\infty}|\mathcal{D})^i(\sigma,\phi_0)||_{L^{2}(\mathcal{S}_{u_{\infty},v})}.
\end{align*}

We also use boldface to denote the supremum over $u$ and $v$ 
of the norms of a specific quantity on the light cone for convenience, 
for example we denote
\begin{align*}
\bm\phi[\phi_0]\equiv\sup_{u}\sum_{i=0}^{11}\frac{1}{a^{\frac{1}{2}}}
||(a^{\frac{1}{2}})^{i-1}\mathcal{D}^i\phi_{0}||_{L^2_{sc}(\mathcal{N}_u(0,v))}, \quad
\underline{\bm\zeta}[\zeta_5]\equiv\sup_{v}\sum_{i=0}^{11}||\frac{1}{|u'|}(a^{\frac{1}{2}})^{i}\mathcal{D}^i
\zeta_5||_{L^2_{sc}(\mathcal{N}'_v(u_{\infty},u))}.
\end{align*}
We make use of non-boldface to denote the supremum over $u$ 
and $v$ of the norm of a specific quantity on 2-sphere, 
for example we denote
\begin{align*}
\Gamma[\lambda]\equiv&\sup_{u,v}\left(\sum_{i=0}^{10}\frac{a^{\frac{1}{2}}}{|u|}
||(a^{\frac{1}{2}}\mathcal{D})^i\lambda||_{L^2_{sc}(\mathcal{S}_{u,v})}
+\sum_{i=0}^6\frac{a^{\frac{1}{2}}}{|u|}
||(a^{\frac{1}{2}}\mathcal{D})^i\lambda||_{L^{\infty}_{sc}(\mathcal{S}_{u,v})}\right), \\
\zeta[\zeta_0]\equiv&\sup_{u,v}\left(\sum_{i=0}^{10}\frac{1}{a^{\frac{1}{2}}}
||(a^{\frac{1}{2}}\mathcal{D})^i\zeta_0||_{L^{2}_{sc}(\mathcal{S}_{u,v})}
+\sum_{i=0}^6\frac{1}{a^{\frac{1}{2}}}
||(a^{\frac{1}{2}}\mathcal{D})^i\zeta_0||_{L^{\infty}_{sc}(\mathcal{S}_{u,v})}\right).
\end{align*}

In the rest of the article, for simplicity we use $A\lesssim B$ to denote that there exist a
constant $C>0$ which is independent of $a$, such that $A\leq CB$. We
do not distinguish between $\lesssim$ and $\leq$ when there is no
source of confusion. 

\end{remark}

\subsection{Strategy of the bootstrap argument}
We start with the bound 
\begin{align*}
\bmGamma_0+\bm\Psi_0+\bm\phi_0+\bm\zeta_0\lesssim\mathcal{I}_0,
\end{align*}
and want to prove that 
\begin{align}
\bmGamma(u,v)+\bm\Psi(u,v)+\bm\phi(u,v)+\bm\zeta(u,v)\lesssim
(\mathcal{I}_0)^2+\mathcal{I}_0+1 \label{Conclusion}
\end{align}
in the region
\begin{align*}
\mathbb{D}=\big\{(u,v)|u_{\infty}\leq u\leq -a/4, \ \ 0\leq v\leq 1\big\}.
\end{align*}
Then, by a last-slice argument, a solution to the Einstein–Weyl system 
exists in the domain $\mathscr{D}$, from this uniform bound.

To achieve this we make the following bootstrap assumptions 
\begin{align}
\bmGamma+\bm\Psi+\bm\phi+\bm\zeta\leq\mathcal{O} \label{Hypothesis}
\end{align}
for large $\mathcal{O}$ such that
\begin{align*}
\mathcal{I}_{0}\ll\mathcal{O},
\end{align*}
but also
\begin{align*}
\mathcal{O}^{20}\leq a^{\frac{1}{16}}.
\end{align*}
We need to show that the following set 
\begin{align*}
\bm{I}=\{u|u_{\infty}\leq u\leq -a/4,\ \eqref{Hypothesis} \ holds \ for \ every \ 0\leq v\leq 1\}
\end{align*}
is both open and closed. 

Firstly, from the local existence results of CIVP for Einstein-Weyl spinor system; see \cite{PengXiaoning2501}, 
there exists a small $\varepsilon$ such that for
$u_{\infty}\leq u\leq u_{\infty}+\varepsilon$ we have
\begin{align*}
\bmGamma+\bm\Psi+\bm\phi+\bm\zeta\lesssim2\mathcal{I}_0\ll\mathcal{O}.
\end{align*}
That means $[u_{\infty}\leq u\leq u_{\infty}+\varepsilon]\subseteq\bm{I}$, i.e. $\bm{I}$ is a closed set. 
Secondly, in sections \ref{L2estimate}, \ref{Elliptic} and \ref{EnergyEstimate} we show 
 \begin{align*}
\bmGamma(u,v)\lesssim&(\bm\Psi(u,v)+\bm\phi(u,v)+\bm\zeta(u,v))^2+\mathcal{I}_0+1, \\
\bm\phi(u,v)\lesssim&\mathcal{I}_0+1, \quad
\bm\zeta(u,v)\lesssim\mathcal{I}_0+1, \quad
\bm\Psi(u,v)\lesssim\mathcal{I}_0+1.
 \end{align*}
These results improve the upper bounds in the bootstrap assumption \eqref{Hypothesis}. 
Consequently, the continuity of solutions and local existence arguments lead to that 
$\bm{I}$ can be extended a bit forward along $u$. This implies that $\bm{I}$ is open. 
Together with the closedness, one concludes that $\bm{I}\equiv[u_{\infty}\leq u\leq -a/4]$ 
and the uniform bound in \eqref{Conclusion} holds.

\subsection{Estimates for the components of the frame}
We introduce the basic estimates on the components of the frame with the bootstrap assumption. 
As the bootstrap assumptions are the same as in paper \cite{An2022,HilValZha23}, 
the ensuing conclusions remain unchanged, 
therefore we provide only a brief outline and the resulting statements. 
For the component $Q$, apply the definition of $\ulvarrho$ one has 
\begin{align*}
|Q-1|\leq\frac{\mathcal{O}}{|u|}.
\end{align*}
Hence when $a$ is sufficiently large, $||Q,Q^{-1}||_{L^{\infty}(\mathcal{S}_{u,v'})}$ are close to $1$. 
Following the same strategy we have the following estimate for the area of $\mathcal{S}_{u,v}$:
\begin{align*}
|\mbox{Area}(S_{u,v})-2\pi u^2|\lesssim\frac{\mathcal{O}}{|u|}.
\end{align*}
under the bootstrap assumption of $\rho$.

\subsection{Important inequalities}
We mainly use the Gr\"onwall inequality and Sobolev embedding in the proof. 
We list these inequalities in scale-invariant form below:
\begin{proposition}
  \begin{subequations}
\begin{align}
  &||\phi||_{L_{sc}^2(\mathcal{S}_{u,v})}\lesssim 
  ||\phi||_{L_{sc}^2(\mathcal{S}_{u,0})}+\int_0^v
  ||\mthorn\phi||_{L_{sc}^2(\mathcal{S}_{u,v'})}\mathrm{d}v' , \label{SCldirectiongronwall}\\
  &||\phi||_{L_{sc}^2(\mathcal{S}_{u,v})}\lesssim
  ||\phi||_{L_{sc}^2(\mathcal{S}_{0,v})}+
  \int_0^u\frac{a}{|u'|^2}||\mthorn'\phi||_{L_{sc}^2(\mathcal{S}_{u',v})}\mathrm{d}u'.
\end{align}
\end{subequations}
\end{proposition}
Moreover, for the ingoing equation we have 
\begin{proposition}
\label{SCweightedgronwall}
For  transport equations of the form 
\begin{align*}
\mthorn'\Gamma=\lambda_0\mu\Gamma+F
\end{align*}
one has that 
\begin{align*}
a^{s_2(\Gamma)}|u|^{\lambda_1-2s_2(\Gamma)}||\Gamma||_{L_{sc}^2(\mathcal{S}_{u,v})}\lesssim 
a^{s_2(\Gamma)}|u_{\infty}|^{\lambda_1-2s_2(\Gamma)}||\Gamma||_{L_{sc}^2(\mathcal{S}_{u_{\infty},v})}+
\int_{u_{\infty}}^ua^{s_2(F)}|u'|^{\lambda_1-2s_2(F)}||F||_{L_{sc}^2(\mathcal{S}_{u',v})}\mathrm{d}u',
\end{align*}
where $s_2(F)=s_2(\mthorn'\Gamma)=s_2(\Gamma)+1$.
\end{proposition}

Making use of the estimate on the area $\mbox{Area}(\mathcal{S}_{u,v})\sim|u|^2$, 
we have the following Sobolev inequalities
\begin{proposition}
 Under the bootstrap assumptions, in terms of scale-invariant norms one has that
\begin{align}
||\phi||_{L_{sc}^{\infty}(\mathcal{S}_{u,v})}\lesssim
\sum_{i\leq2}||\left(a^{1/2}\mathcal{D}\right)^i\phi||_{L_{sc}^2(\mathcal{S}_{u,v})}.
\end{align}
\end{proposition}
For the details of the proof we refer reader to \cite{An2022,HilValZha23}.

\subsection{Commutators}
In the following let $H_0\equiv \mthorn f$ and
$H_k\equiv\mthorn\meth^kf$, we have that 
\begin{align*}
H_k=\sum_{i_1+i_2+i_3=k}\meth^{i_1}\Gamma(\pi,\tau)^{i_2}\meth^{i_3}H_0+
\sum_{i_1+i_2+i_3+i_4=k}\meth^{i_1}\Gamma(\tau,\pi)^{i_2}\meth^{i_3}\Gamma(\tau,\pi,\rho,\sigma)\meth^{i_4}f.
\end{align*}
Here we write $\meth^{i_1}\Gamma^{i_2}$ to
denote $\meth^{j_1}\Gamma\meth^{j_2}\Gamma...\meth^{j_{i_2}}\Gamma$
where $i_1\geq0$, $i_2\geq1$, $j_1$, $j_2$, ...,
$j_{i_2}\in\mathbb{N}$ and $j_1+j_2+...+j_{i_2}=i_1$.
Similarly, letting $G_0\equiv\mthorn'f$ and
$G_k\equiv\mthorn'\meth^kf$, we have that  
\begin{align*}
G_k=-k\mu\meth^kf+\meth^{k}G_0+
\sum_{i=1}^k\meth^i\mu\meth^{k-i}f+
\sum_{i=0}^k\meth^i\lambda\meth^{k-i}f.
\end{align*}
The above expressions for $H_k$ and $G_k$ neither show the exact
constants in each sum nor the T-weight of $f$. The expressions of $H_k$ and $G_k$ do 
not distinguish the $\meth$ and $\meth'$ and their complex
conjugates. The reason behind this can be traced back to the definition 
of the norms. Further discussion can be found in~\cite{HilValZha23}. 


\section{$L^2(\mathcal{S})$ estimates}
\label{L2estimate}
In this section, under the 
bootstrap assumptions 
\begin{align*}
\bmGamma+\bm\Psi+\bm\phi+\bm\zeta\leq\mathcal{O},
\end{align*} 
we estimate the $L^2_{sc}(\mathcal{S})$ norm of next-to-leading derivative of each quantity. 
We show that 
\begin{align*}
\bmGamma\lesssim(\bm\Psi+\bm\phi+\bm\zeta)^2+\mathcal{I}_0+1.
\end{align*}
The strategy is to make use of the transport equation and the Gr\"onwall inequalities. 
Technically, due to the property of H\"older inequality in the scale-invariant norm \eqref{L2holder}, 
products of nonlinear terms are small, 
Take $\Gamma\phi_j$ as an example, we have
\begin{align*}
&\sum_{i_1+...+i_4=k}||a^{\frac{k}{2}}\meth^{i_1}\Gamma(\tau,\pi)^{i_2}\meth^{i_3}\Gamma
\meth^{i_4}\phi_j||_{L^2_{sc}(\mathcal{S}_{u,v})}\\
\leq&\frac{(a^{\frac{1}{2}})^{i_2}}{|u|^{i_2+1}}||(a^{\frac{1}{2}}\mathcal{D})^{j_1}\Gamma||...
||(a^{\frac{1}{2}}\mathcal{D})^{j_{i_2}}\Gamma||\times
||(a^{\frac{1}{2}}\mathcal{D})^{i_3}\Gamma||\times||(a^{\frac{1}{2}}\mathcal{D})^{i_4}\bm{\phi}||\\
\leq&\frac{(a^{\frac{1}{2}})^{i_2}}{|u|^{i_2+1}}\mathcal{O}^{i_2}
||(a^{\frac{1}{2}}\mathcal{D})^{i_3}\Gamma||\times||(a^{\frac{1}{2}}\mathcal{D})^{i_4}\bm{\phi}||
\leq\frac{1}{|u|}||(a^{\frac{1}{2}}\mathcal{D})^{i}\Gamma||\times||(a^{\frac{1}{2}}\mathcal{D})^{k-i}\bm{\phi}||.
\end{align*} 
Hence in the following calculations we focus on the least multiplicative term 
at the same differential order in the absence of ambiguity. 
Moreover, because we make the same bootstrap assumption for connections and curvatures 
as in paper \cite{An2022,HilValZha23}, so we focus on the terms contain matter fields $\phi$ and $\zeta$. 

\subsection{$L^2(\mathcal{S})$ estimates of connection coefficients}
We start with looking at the estimates for the connection
coefficients. 

\begin{proposition}
\label{L2lambda}
For $0\leq k\leq10$, one has that 
\begin{align*}
\frac{a^{\frac{1}{2}}}{|u|}||(a^{\frac{1}{2}}\mathcal{D})^k\lambda||_{L^2_{sc}(\mathcal{S}_{u,v})}\lesssim&1, \quad
\frac{1}{a^{\frac{1}{2}}}||(a^{\frac{1}{2}}\mathcal{D})^{k}\sigma||_{L^2_{sc}(\mathcal{S}_{u,v})}\lesssim
\bm\Psi[\Psi_0]+1.
\end{align*}
Here 
\begin{align*}
\bm\Psi[\Psi_0]\equiv\frac{1}{a^{\frac{1}{2}}}\left(\int_0^v||(a^{\frac{1}{2}}\mathcal{D})^{k}\Psi_0||^2_{L^2_{sc}(\mathcal{S}_{u,v'})} \right)^{\frac{1}{2}}
=\frac{1}{a^{\frac{1}{2}}}||(a^{\frac{1}{2}}\mathcal{D})^i\Psi_0||_{L^2_{sc}(\mathcal{N}_u(0,v))}.
\end{align*}
\end{proposition}

\begin{proof}
Make use of \eqref{Thornprimelambda} and \eqref{Thornsigma} for $\lambda$ and $\sigma$ respectively. 
Pure gravitation effect and leads to the same analysis in the vacuum or scalar case, 
see~\cite{An2022,HilValZha23}.
\end{proof}

\begin{proposition}
\label{L2ulomega}
For $0\leq k\leq10$, one has that 
\begin{align*}
||(a^{\frac{1}{2}}\mathcal{D})^{k}\ulomega||_{L^2_{sc}(\mathcal{S}_{u,v})}\lesssim
\bm\Psi[\Psi_2]+1.
\end{align*}
\end{proposition}

\begin{proof}
Make use of \eqref{Thornulomega} and the definition of $\Tizeta_4$ we have 
\begin{align*}
\mthorn\ulomega=&2\tau\bar\tau+2\tau\pi+2\bar\tau\bar\pi
+\Psi_2+\bar\Psi_2+
2\mathrm{i}(\bar\Tizeta_4\phi_0-\Tizeta_4\bar\phi_0+\bar\zeta_1\phi_1-\zeta_1\bar\phi_1)
\end{align*}
and for $k\leq10$, commuting $\meth^{k}$ with $\mthorn$, we obtain 
\begin{align*}
\mthorn\meth^{k}\ulomega=&\meth^k\Psi_2
+\sum_{i_1+i_2+i_3=k,i_3<k}\meth^{i_1}\Gamma(\tau,\pi)^{i_2}
\meth^{i_3}\Psi_2
+\sum_{i_1+i_2+i_3+i_4=k}\meth^{i_1}\Gamma(\tau,\pi)^{i_2}
\meth^{i_3}\Gamma(\tau,\pi,\rho,\sigma)
\meth^{i_4}\Gamma(\tau,\pi,\ulomega) \\
&+\sum_{i_1+i_2+i_3+i_4=k}\meth^{i_1}\Gamma(\tau,\pi)^{i_2}
\meth^{i_3}\phi_i
\meth^{i_4}\zeta_j.
\end{align*}
Then we have
\begin{align*}
||(a^{\frac{1}{2}}\meth)^{k}\ulomega||_{L^2_{sc}(\mathcal{S}_{u,v})}\lesssim&
||(a^{\frac{1}{2}}\meth)^{k}\ulomega||_{L^2_{sc}(\mathcal{S}_{u,0})}+
\int_0^v||a^{\frac{k}{2}}\mthorn\meth^{k}\ulomega||_{L^2_{sc}(\mathcal{S}_{u,v'})} \\
\leq&Vac+\sum_{i_1+i_2+i_3+i_4=k}\int_0^v
||a^{\frac{k}{2}}\mathcal{D}^{i_1}\Gamma(\tau,\pi)^{i_2}
\mathcal{D}^{i_3}\phi_i\mathcal{D}^{i_4}\zeta_j||_{L^2_{sc}(\mathcal{S}_{u,v'})} \\
\leq&\bm\Psi[\Psi_2]+\frac{1}{|u|}a^{\frac{1}{2}}\mathcal{O}^2+1
\leq\bm\Psi[\Psi_2]+1.
\end{align*}
Here "Vac" means the vacuum case which contains curvatures and connections, analysis can be 
found in \cite{An2022,HilValZha23}. 
In what follows, we use "Vac" to denote the same object and present the results directly, 
omitting the details of the analysis.

\end{proof}

\begin{proposition}
\label{L2ulvarrho}
For $0\leq k\leq10$, one has that
\begin{align*}
||(a^{\frac{1}{2}}\mathcal{D})^k\ulvarrho||_{L^2_{sc}(\mathcal{S}_{u,v})}\lesssim
\underline{\bm\Psi}[\Psi_2]+\underline{\bm{\phi}}[\phi_1]^2+\underline{\bm{\phi}}[\phi_1]+1.
\end{align*}
\end{proposition}

\begin{proof}
Make use of 
\begin{align*}
\mthorn'\ulvarrho &= \Psi_2 + \bar\Psi_2
+ \ulvarrho\,\ulomega
+ 2\,\mathrm{i}\,\bar\zeta_4\,\phi_0
- 2\,\mathrm{i}\,\zeta_4\,\bar\phi_0
+ 2\,\mathrm{i}\,\bar\zeta_1\,\phi_1
- 2\,\mathrm{i}\,\zeta_1\,\bar\phi_1 \\[6pt]
&\quad
+ 2\,\pi\,\bar\pi
+ 2\,\pi\,\tau
+ 2\,\bar\pi\,\bar\tau \,.
\end{align*}
Commuting $\mthorn'$ with $\meth^k$, we have that
\begin{align*}
\mthorn'\meth^k\ulvarrho+k\mu\meth^k\ulvarrho=\meth^k\Psi_2
+\sum_{i=0}^k\meth^i\Gamma(\tau,\pi)\meth^{k-i}\pi
+\sum_{i=0}^k\meth^i\ulomega\meth^{k-i}\ulchi
+\sum_{i=0}^k\meth^i\phi_j\meth^{k-i}\zeta_l.
\end{align*}
Denote the righthand side of the above equation by $F$. One has that 
\begin{align*}
&s_2(\ulvarrho)=0, \ \ s_2(\meth^k\ulvarrho)=\frac{k}{2}, \ \ s_2(F)=\frac{k+2}{2}, \\
&\lambda_0=-k, \ \ \lambda_1=-\lambda_0-1=k-1, \\
& \lambda_1-2s_2(F)=-3, \ \ \lambda_1-2s_2(\meth^k\ulvarrho)=-1.
\end{align*}
Now, using Prop. \ref{SCweightedgronwall}, one obtains
\begin{align*}
\frac{1}{|u|}||(a^{\frac{1}{2}}\meth)^k\ulvarrho||_{L^2_{sc}(\mathcal{S}_{u,v})}\lesssim&
\int_{u_{\infty}}^u\frac{a}{|u'|^3}||a^{\frac{k}{2}}F||_{L^2_{sc}(\mathcal{S}_{u',v})} \\
\lesssim&\int_{u_{\infty}}^u\frac{a}{|u'|^3}||a^{\frac{k}{2}}\meth^k\Psi_2||_{L^2_{sc}(\mathcal{S}_{u',v})} 
+\sum_{i=0}^k\int_{u_{\infty}}^u\frac{a}{|u'|^3}||a^{\frac{k}{2}}\meth^i\Gamma
\meth^{k-i}\Gamma||_{L^2_{sc}(\mathcal{S}_{u',v})} \\
&+\sum_{i=0}^k\int_{u_{\infty}}^u\frac{a}{|u'|^3}||a^{\frac{k}{2}}\meth^i\phi_j
\meth^{k-i}\zeta_l||_{L^2_{sc}(\mathcal{S}_{u',v})} \\
=&Vac+I.
\end{align*}
For Vac part we have control $\frac{\underline{\bm\Psi}[\Psi_2]}{|u|}$.
For the integral $I$, we have
\begin{align*}
I\lesssim\int_{u_{\infty}}^u\frac{a}{|u'|^3}\frac{1}{|u'|}a^{\frac{1}{2}}\phi[\phi_0]
\frac{|u|}{a^{\frac{1}{2}}}\zeta[\zeta_4]+\frac{1}{|u|} 
\lesssim\frac{\phi[\phi_0]\zeta[\zeta_4]+1}{|u|}.
\end{align*}
Hence we have
\begin{align*}
||(a^{\frac{1}{2}}\meth)^k\ulvarrho||_{L^2_{sc}(\mathcal{S}_{u,v})}\lesssim&\underline{\bm\Psi}[\Psi_2]
+\phi[\phi_0]\zeta[\zeta_4]+1 \\
\lesssim&\underline{\bm\Psi}[\Psi_2]+\underline{\bm{\phi}}[\phi_1]^2+\underline{\bm{\phi}}[\phi_1]+1.
\end{align*}
Here in the last step we make use of the estimate for $\phi_0$ and $\zeta_4$ which will be 
shown in Prop. \ref{L2phiA} and Prop.\ref{L2zeta}.
\end{proof}

\begin{proposition}
\label{L2tau}
For $0\leq k\leq10$, one has that
\begin{align*}
||(a^{\frac{1}{2}}\mathcal{D})^{k}\tau||_{L^2_{sc}(\mathcal{S}_{u,v})}\lesssim
\underline{\bmzeta}[\zeta_1]\bmphi[\phi_0]+\Psi[\TiPsi_1]+\bmphi[\phi_0]+\underline{\bmzeta}[\zeta_1]+1.
\end{align*}
\end{proposition}

\begin{proof}

We make use of the structure equation \eqref{Thorntau}
\begin{align*}
\mthorn\tau&=(\tau+\bar\pi)\rho+(\bar\tau+\pi)\sigma
+\TiPsi_1+2\mathrm{i}(2\bar\zeta_1\phi_0-\zeta_3\bar\phi_0-\zeta_0\bar\phi_1).
\end{align*}
For $k\leq10$, commuting $\meth^{k}$ with $\mthorn$, we have that 
\begin{align*}
\mthorn\meth^{k}\tau=Vac
+\sum_{i_1+i_2+i_3+i_4=k}\meth^{i_1}\Gamma(\tau,\pi)^{i_2}
\meth^{i_3}\zeta_i
\meth^{i_4}\phi_l.
\end{align*}
Then we have
\begin{align*}
||(a^{\frac{1}{2}}\meth)^{k}\tau||_{L^2_{sc}(\mathcal{S}_{u,v})}\lesssim&
||(a^{\frac{1}{2}}\meth)^{k}\tau||_{L^2_{sc}(\mathcal{S}_{u,0})}+
\int_0^v||a^{\frac{k}{2}}\mthorn\meth^{k}\tau||_{L^2_{sc}(\mathcal{S}_{u,v'})} \\
\leq&\bm\Psi[\TiPsi_1]+1+\frac{a^{\frac{1}{2}}}{|u|}\mathcal{O}^2
+\frac{a}{|u|}\Gamma[\zeta_0]\Gamma[\phi_1]
\leq\bm\Psi[\TiPsi_1]+\zeta[\zeta_0]\phi[\phi_1]+1 \\
\lesssim&\underline{\bmzeta}[\zeta_1]\bmphi[\phi_0]
+\Psi[\TiPsi_1]+\bmphi[\phi_0]+\underline{\bmzeta}[\zeta_1]+1.
\end{align*}

\end{proof}

\begin{proposition}
\label{L2rho}
For $0\leq k\leq10$, one has that 
\begin{align*}
||(a^{\frac{1}{2}}\mathcal{D})^k\rho||_{L^2_{sc}(\mathcal{S}_{u,v})}\lesssim&
\bm\Psi[\Psi_0]^2+\underline{\bm{\phi}}[\phi_1]\underline{\bmzeta}[\zeta_1]+\bm\Psi[\Psi_0]
+\underline{\bmzeta}[\zeta_1]+\underline{\bm{\phi}}[\phi_1]+1.
\end{align*}
\end{proposition}
\begin{proof}
We make use of the structure equation for $\rho$ \eqref{Thornrho}
\begin{align*}
\mthorn\rho&=\rho^2+\sigma\bar\sigma+2\mathrm{i}(\bar\zeta_0\phi_0-\zeta_0\bar\phi_0).
\end{align*}
For $k\leq10$, commuting $\meth^{k}$ with $\mthorn$, we have that 
\begin{align*}
\mthorn\meth^{k}\rho=&Vac
+\sum_{i_1+i_2+i_3+i_4=k}\meth^{i_1}\Gamma(\tau,\pi)^{i_2}
\meth^{i_3}\phi_0\meth^{i_4}\zeta_0.
\end{align*}
Then we have
\begin{align*}
||(a^{\frac{1}{2}}\meth)^{k}\rho||_{L^2_{sc}(\mathcal{S}_{u,v})}\lesssim&
||(a^{\frac{1}{2}}\meth)^{k}\rho||_{L^2_{sc}(\mathcal{S}_{u,0})}+
\int_0^v||a^{\frac{k}{2}}\mthorn\meth^{k}\rho||_{L^2_{sc}(\mathcal{S}_{u,v'})} \\
\lesssim&1+\frac{a}{|u|}\left(\bmGamma(\sigma)_{\infty}\bmGamma(\sigma)_{2}
+\bmGamma(\phi_0)_{\infty}\bmGamma(\zeta_0)_{2}+\bmGamma(\zeta_0)_{\infty}\bmGamma(\phi_0)_{2}\right)+\frac{\mathcal{O}}{|u|}+\frac{a\mathcal{O}}{|u|^2} \\
\lesssim&1+\frac{a}{|u|}(\Gamma[\sigma]^2+\zeta[\zeta_0]\phi[\phi_0]) \\
\lesssim&\frac{a}{|u|}\left(\bm\Psi[\Psi_0]^2+\underline{\bm{\phi}}[\phi_1]\underline{\bmzeta}[\zeta_1]+\bm\Psi[\Psi_0]+\underline{\bmzeta}[\zeta_1]+\underline{\bm{\phi}}[\phi_1]\right)+1.
\end{align*}
The initial value of $\rho$ is $-\frac{1}{|u|}$ which is independent of spherical coordinates, 
hence when $1\leq k\leq10$,
\begin{align*}
||(a^{\frac{1}{2}}\meth)^{k}\rho||_{L^2_{sc}(\mathcal{S}_{u,0})}=0
\end{align*}
and then when $1\leq k\leq10$
\begin{align*}
||(a^{\frac{1}{2}}\meth)^{k}\rho||_{L^2_{sc}(\mathcal{S}_{u,v})}
\lesssim&\frac{a}{|u|}\left(\bm\Psi[\Psi_0]^2+\underline{\bm{\phi}}[\phi_1]\underline{\bmzeta}[\zeta_1]+\bm\Psi[\Psi_0]+\underline{\bmzeta}[\zeta_1]+\underline{\bm{\phi}}[\phi_1]\right).
\end{align*}

\end{proof}

\begin{proposition}
\label{L2mu}
For $0\leq k\leq10$, one has that 
\begin{align*}
\frac{a}{|u|}||(a^{\frac{1}{2}}\mathcal{D})^k\Timu||_{L^2_{sc}(\mathcal{S}_{u,v})}\lesssim&
\bm\Psi[\Psi_2]+1,
\end{align*}
\begin{align*}
\frac{a}{|u|^2}||(a^{\frac{1}{2}}\mathcal{D})^k\mu||_{L^2_{sc}(\mathcal{S}_{u,v})}\lesssim&1.
\end{align*}
\end{proposition}

\begin{proof}
We make use of the structure equation \eqref{Thornprimetildemu} for $\Timu$ 
and the definition of $\Tizeta_5$ and obtain
\begin{align*}
\mthorn'\tilde\mu+2\mu\tilde\mu=\tilde\mu^2-\ulomega\mu-\lambda\bar\lambda
+2\mathrm{i}(\Tizeta_5\bar\phi_1-\bar\Tizeta_5\phi_1).
\end{align*}
Commuting $\meth^k$ with $\mthorn'$, we find that 
\begin{align*}
\mthorn'\meth^k\Timu+(k+2)\mu\meth^k\Timu=&
\sum_{i=0}^k\meth^i\Gamma(\Timu,\lambda)\meth^{k-i}\Gamma(\Timu,\lambda)
+\sum_{i=0}^k\meth^i\ulomega\meth^{k-i}\mu
+\sum_{i=0}^k\meth^i\Tizeta_5\meth^{k-i}\phi_1
\end{align*}
which leads to the following estimate:
\begin{align*}
\frac{a}{|u|}||(a^{\frac{1}{2}}\meth)^k\Timu||_{L^2_{sc}(\mathcal{S}_{u,v})}\lesssim&
\frac{a}{|u_{\infty}|}||(a^{\frac{1}{2}}\meth)^k\Timu||_{L^2_{sc}(\mathcal{S}_{u_{\infty},v})}
+\int_{u_{\infty}}^u\frac{a^2}{|u'|^3}||a^{\frac{k}{2}}F||_{L^2_{sc}(\mathcal{S}_{u',v})} \\
\leq&Vac+\sum_{i=0}^k\int_{u_{\infty}}^u\frac{a^{2}}{|u'|^3}||a^{\frac{k}{2}}\mathcal{D}^i\Tizeta_5
\mathcal{D}^{k-i}\phi_1||_{L^2_{sc}(\mathcal{S}_{u',v})}\\
\leq&1+\frac{a}{|u|}\Gamma[\ulomega]+\frac{1}{|u|^{\frac{1}{2}}}\underline{\bmzeta}[\Tizeta_5]\phi[\phi_1]
\leq1+\bm\Psi[\Psi_2].
\end{align*}
Note here for $\Tizeta_5$ we can only obtain its tenth derivative on the ingoing lightcone and 
hence we translate the estimate from $\mathcal{S}$ to $\mathcal{N}'_v(u_{\infty},u)$. 
The same strategy is used for $\pi$ in the next proposition. 
With the relation between $\mu$ and $\Timu$ we have
\begin{align*}
\frac{a}{|u|^2}||(a^{\frac{1}{2}}\meth)^k\mu||_{L^2_{sc}(\mathcal{S}_{u,v})}\lesssim1.
\end{align*}

\end{proof}

\begin{proposition}
\label{L2pi}
For $0\leq k\leq10$, one has that 
\begin{align*}
||(a^{\frac{1}{2}}\mathcal{D})^k\pi||_{L^2_{sc}(\mathcal{S}_{u,v})}\lesssim
\underline{\bm\Psi}[\TiPsi_3]+\bm\Psi[\TiPsi_1]+1.
\end{align*}
\end{proposition}

\begin{proof}
We make use of the structure equation \eqref{Thornprimepi} for $\pi$
\begin{align*}
\mthorn'\pi=-(\pi+\bar\tau)\mu-(\bar\pi+\tau)\lambda-\TiPsi_3-
2\mathrm{i}(-\Tizeta_5\bar\phi_0+2\bar\Tizeta_4\phi_1-\zeta_2\bar\phi_1).
\end{align*}
Commuting $\meth^k$ with $\mthorn'$, we find that 
\begin{align*}
\mthorn'\meth^k\pi+(k+1)\mu\meth^k\pi=&-\meth^k\TiPsi_3-\mu\meth^k\bar\tau
+\sum_{i=0}^k\meth^i\Gamma(\tau,\pi)\meth^{k-i}\Gamma(\Timu,\lambda) 
+\sum_{i=0}^k\meth^i\bar\zeta_j\meth^{k-i}\phi_l. 
\end{align*}
We then obtain the following estimate
\begin{align*}
\frac{1}{|u|}||(a^{\frac{1}{2}}\meth)^k\pi||_{L^2_{sc}(\mathcal{S}_{u,v})}\lesssim&
Vac+\sum_{i=0}^k\int_{u_{\infty}}^u\frac{a}{|u'|^3}||a^{\frac{k}{2}}
\meth^i\bar\zeta_j\meth^{k-i}\phi_l||_{L^2_{sc}(\mathcal{S}_{u',v})} \\
\leq&\frac{1+\underline{\bm\Psi}[\TiPsi_3]+\bm\Gamma[\tau]}{|u|}
+\sum_{i=0}^k\int_{u_{\infty}}^u\frac{a}{|u'|^3}||a^{\frac{k}{2}}\mathcal{D}^i\Tizeta_5
\mathcal{D}^{k-i}\phi_0||_{L^2_{sc}(\mathcal{S}_{u',v})} \\
&+\int_{u_{\infty}}^u\frac{a}{|u'|^3}\frac{1}{|u'|}a^{\frac{1}{2}}\mathcal{O}^2\\
\leq&\frac{1+\underline{\bm\Psi}[\TiPsi_3]+\bm\Psi[\TiPsi_1]}{|u|}
+\frac{a^{\frac{1}{2}}\phi[\phi_0]\underline{\bmzeta}[\Tizeta_5]}{|u|^2}+\frac{\mathcal{O}^2}{|u|^{\frac{3}{2}}} \\
\leq&\frac{1+\underline{\bm\Psi}[\TiPsi_3]+\bm\Psi[\TiPsi_1]}{|u|}.
\end{align*}

\end{proof}

\subsection{$L^2(\mathcal{S})$ estimates of $\phi_A$}

\begin{proposition}
\label{L2phiA}
For $0\leq k\leq10$, one has that 
\begin{align*}
\frac{1}{a^{\frac{1}{2}}}||(a^{\frac{1}{2}}\mathcal{D})^k\phi_0||_{L^2_{sc}(\mathcal{S}_{u,v})}
\lesssim&\underline{\bm{\phi}}[\phi_1]+1, \\
\frac{1}{a^{\frac{1}{2}}}||(a^{\frac{1}{2}}\mathcal{D})^k\phi_1||_{L^2_{sc}(\mathcal{S}_{u,v})}
\lesssim&\bmphi[\phi_0]+1.
\end{align*}
\end{proposition}

\begin{proof}
We make use of equation of motion for $\phi_0$:
\begin{align*}
\mthorn'\phi_{0}-\meth\phi_{1}&=(\frac{\ulomega}{2}-\mu)\phi_{0}
-\frac{\bar\tau\phi_{1}}{2}.
\end{align*}
Commuting $\meth^k$ with $\mthorn'$, we find that 
\begin{align*}
\mthorn'\meth^k\phi_0+(k+1)\mu\meth^k\phi_0=&-\meth^{k+1}\phi_1
+\sum_{i=0}^k\meth^i\Gamma(\ulomega,\Timu,\lambda)\meth^{k-i}\phi_0 
+\sum_{i=0}^k\meth^i\tau\meth^{k-i}\phi_1. 
\end{align*}
Denoting the righthand side of the previous equation by $F$, one has that
\begin{align*}
\frac{1}{a^{\frac{1}{2}}}||(a^{\frac{1}{2}}\meth)^k\phi_0||_{L^2_{sc}(\mathcal{S}_{u,v})}\lesssim&
\frac{1}{a^{\frac{1}{2}}}||(a^{\frac{1}{2}}\meth)^k\phi_0||_{L^2_{sc}(\mathcal{S}_{u_{\infty},v})}
+\frac{1}{a^{\frac{1}{2}}}\int_{u_{\infty}}^u\frac{a}{|u'|^2}||a^{\frac{k}{2}}F||_{L^2_{sc}(\mathcal{S}_{u',v})} \\
\lesssim&1+\frac{1}{a^{\frac{1}{2}}}\int_{u_{\infty}}^u\frac{a}{|u'|^2}
||a^{\frac{k}{2}}\mathcal{D}^{k+1}\phi_1||_{L^2_{sc}(\mathcal{S}_{u',v})} \\
&+\sum_{i=0}^{k}\frac{1}{a^{\frac{1}{2}}}\int_{u_{\infty}}^u\frac{a}{|u'|^2}
||a^{\frac{k}{2}}\mathcal{D}^{i}\Gamma(\ulomega,\Timu,\lambda)\mathcal{D}^{k-i}\phi_0||_{L^2_{sc}(\mathcal{S}_{u',v})} \\
&+\sum_{i=0}^{k}\frac{1}{a^{\frac{1}{2}}}\int_{u_{\infty}}^u\frac{a}{|u'|^2}
||a^{\frac{k}{2}}\mathcal{D}^{i}\tau\mathcal{D}^{k-i}\phi_1||_{L^2_{sc}(\mathcal{S}_{u',v})} \\
\leq&1+\underline{\bmphi}[\phi_1]+
\frac{1}{a^{\frac{1}{2}}}\int_{u_{\infty}}^u\frac{a}{|u'|^2}
\frac{1}{|u'|}\frac{|u'|}{a^{\frac{1}{2}}}\mathcal{O}a^{\frac{1}{2}}\mathcal{O} \\
\leq&1+\underline{\bmphi}[\phi_1]+\frac{a^{\frac{1}{2}}}{|u|}\mathcal{O}^2.
\end{align*}

Make use of
\begin{align*}
\mthorn\phi_1=\phi_0(\pi-\frac{\bar\pi}{2})+\phi_1\rho+\meth'\phi_0
\end{align*}
and obtain
\begin{align*}
\mthorn\meth^k\phi_1=\meth^{k+1}\phi_0
+\sum_{i_1+i_2+i_3+i_4=k}\meth^{i_1}\Gamma(\tau,\pi)^{i_2}\meth^{i_3}\Gamma(\pi,\rho)
\meth^{i_4}(\phi_0,\phi_1).
\end{align*}
We then have
\begin{align*}
\frac{1}{a^{\frac{1}{2}}}||(a^{\frac{1}{2}}\meth)^{k}\phi_1||_{L^2_{sc}(\mathcal{S}_{u,v})}\lesssim&
\frac{1}{a^{\frac{1}{2}}}||(a^{\frac{1}{2}}\meth)^{k}\phi_1||_{L^2_{sc}(\mathcal{S}_{u,0})}+
\frac{1}{a^{\frac{1}{2}}}\int_0^v||a^{\frac{k}{2}}\mthorn\meth^{k}\phi_1||_{L^2_{sc}(\mathcal{S}_{u,v'})} \\
\lesssim&1+\frac{1}{a^{\frac{1}{2}}}
\int_0^v||a^{\frac{k}{2}}\mathcal{D}^{k+1}\phi_0||_{L^2_{sc}(\mathcal{S}_{u,v'})}\\
&+\frac{1}{a^{\frac{1}{2}}}\sum_{i_1+i_2+i_3+i_4=k}\int_0^v||a^{\frac{k}{2}}
\mathcal{D}^{i_1}\Gamma^{i_2}\mathcal{D}^{i_3}\Gamma
\mathcal{D}^{i_4}(\phi_0,\phi_1)||_{L^2_{sc}(\mathcal{S}_{u,v'})} \\
\leq&1+\bmphi[\phi_0]+\frac{1}{|u|}a^{\frac{1}{2}}\mathcal{O}^2
\lesssim1+\bmphi[\phi_0].
\end{align*}

\end{proof}

\begin{proposition}
\label{L2zeta}
For $0\leq k\leq10$, one has that 
\begin{align*}
\frac{1}{a^{\frac{1}{2}}}||(a^{\frac{1}{2}}\mathcal{D})^k\zeta_0||_{L^2_{sc}(\mathcal{S}_{u,v})}\lesssim&
\underline{\bmzeta}[\zeta_1]+1, \\
||(a^{\frac{1}{2}}\mathcal{D})^k\zeta_1||_{L^2_{sc}(\mathcal{S}_{u,v})}\lesssim&
1+\underline{\bmzeta}[\zeta_2], \\
||(a^{\frac{1}{2}}\mathcal{D})^k\Tizeta_4||_{L^2_{sc}(\mathcal{S}_{u,v})}\lesssim&
1+\bmzeta[\zeta_3], \quad
\frac{a^{\frac{1}{2}}}{|u|}||(a^{\frac{1}{2}}\mathcal{D})^k\zeta_4||_{L^2_{sc}(\mathcal{S}_{u,v})}\lesssim
1+\underline{\bm{\phi}}[\phi_1], \\
||(a^{\frac{1}{2}}\mathcal{D})^k\zeta_2||_{L^2_{sc}(\mathcal{S}_{u,v})}\lesssim&
\bm\Psi[\Psi_0]\underline{\bm{\phi}}[\phi_1]
+\bm\Psi[\Psi_0]+\underline{\bm{\phi}}[\phi_1]+\bmzeta[\zeta_1]+1, \\
\frac{a^{\frac{1}{2}}}{|u|}||(a^{\frac{1}{2}}\mathcal{D})^k\zeta_5||_{L^2_{sc}(\mathcal{S}_{u,v})}\lesssim&
1+\bmzeta[\zeta_4].
\end{align*}
For $0\leq k\leq9$, one has that 
\begin{align*}
||(a^{\frac{1}{2}}\mathcal{D})^k\zeta_3||_{L^2_{sc}(\mathcal{S}_{u,v})}\lesssim&\bm\Psi[\Psi_0]\bmphi[\phi_0]+\bm\Psi[\Psi_0]+\bmphi[\phi_0]+1, \\
||(a^{\frac{1}{2}}\mathcal{D})^k\Tizeta_5||_{L^2_{sc}(\mathcal{S}_{u,v})}\lesssim&1, \quad
\frac{a^{\frac{1}{2}}}{|u|}||(a^{\frac{1}{2}}\mathcal{D})^k\zeta_5||_{L^2_{sc}(\mathcal{S}_{u,v})}\lesssim
1+\underline{\bm{\phi}}[\phi_0].
\end{align*}
\end{proposition}

\begin{proof}
We start with $\zeta_0$ and make use of \eqref{thornprimezeta0}
\begin{align*}
\mthorn'\zeta_0+\mu\zeta_0-\meth\zeta_1=\Gamma\zeta_j+\zeta_j\phi_l,
\end{align*}
and obtain
\begin{align*}
\mthorn'\meth^k\zeta_0+(k+1)\mu\zeta_0=&\meth^{k+1}\zeta_1+
\sum_{i=0}^k\meth^i\Gamma(\ulomega,\Timu,\lambda,\rho,\sigma,\tau)\meth^{k-i}\zeta_{0,1,2,3,4}\\
&+\sum_{i_1+i_2+i_3=k}\meth^{i_1}\phi_{j_1}\meth^{i_2}\phi_{j_2}\meth^{i_3}\zeta_{j_3}.
\end{align*}

Denoting the righthand side of the previous equation by $F$, one has that
\begin{align*}
\frac{1}{a^{\frac{1}{2}}}||(a^{\frac{1}{2}}\meth)^k\zeta_0||_{L^2_{sc}(\mathcal{S}_{u,v})}\lesssim&
\frac{1}{a^{\frac{1}{2}}}||(a^{\frac{1}{2}}\meth)^k\zeta_0||_{L^2_{sc}(\mathcal{S}_{u_{\infty},v})}
+\frac{1}{a^{\frac{1}{2}}}\int_{u_{\infty}}^u\frac{a}{|u'|^2}||a^{\frac{k}{2}}F||_{L^2_{sc}(\mathcal{S}_{u',v})} \\
\lesssim&1+\frac{1}{a^{\frac{1}{2}}}\int_{u_{\infty}}^u\frac{a}{|u'|^2}
||a^{\frac{k}{2}}\mathcal{D}^{k+1}\zeta_1||_{L^2_{sc}(\mathcal{S}_{u',v})}\\
&+\sum_{i=0}^k\frac{1}{a^{\frac{1}{2}}}\int_{u_{\infty}}^u\frac{a}{|u'|^2}
||a^{\frac{k}{2}}\mathcal{D}^{i}\Gamma\mathcal{D}^{k-i}\zeta_{j}||_{L^2_{sc}(\mathcal{S}_{u',v})}\\
&+\sum_{i_1+i_2+i_3=k}\frac{1}{a^{\frac{1}{2}}}\int_{u_{\infty}}^u\frac{a}{|u'|^2}
||a^{\frac{k}{2}}\mathcal{D}^{i_1}\phi_{j_1}\mathcal{D}^{i_2}\phi_{j_2}
\mathcal{D}^{i_3}\zeta_{j_3}||_{L^2_{sc}(\mathcal{S}_{u',v})} \\
\leq&1+\underline{\bmzeta}[\zeta_1]+
\frac{1}{a^{\frac{1}{2}}}\int_{u_{\infty}}^u\frac{a}{|u'|^2}\frac{1}{|u'|}\frac{|u'|}{a^{\frac{1}{2}}}
a^{\frac{1}{2}}\mathcal{O}^2
+\frac{1}{a^{\frac{1}{2}}}\int_{u_{\infty}}^u\frac{a}{|u'|^2}\frac{1}{|u'|^2}
a^{\frac{1}{2}}a^{\frac{1}{2}}a^{\frac{1}{2}}\mathcal{O}^3 \\
\leq&1+\underline{\bmzeta}[\zeta_1]+\frac{a^{\frac{1}{2}}}{|u|}\mathcal{O}^2
+\frac{a^2}{|u|^3}\mathcal{O}^3
\leq1+\underline{\bmzeta}[\zeta_1].
\end{align*}

Similarly for $\zeta_1$ we make use of the ingoing equation \eqref{thornprimezeta1} and obtain 
\begin{align*}
||(a^{\frac{1}{2}}\mathcal{D})^k\zeta_1||_{L^2_{sc}(\mathcal{S}_{u,v})}\lesssim&
1+\underline{\bmzeta}[\zeta_2].
\end{align*}

Next we estimate $\Tizeta_4$ and we make use of \eqref{thornTizeta4}
\begin{align*}
\mthorn\Tizeta_4=\meth'\zeta_3+(\Psi_2+\meth\pi)\phi_0+\phi_0\mu\rho
+\zeta_i\phi_j^2+\Gamma\zeta_i+\Gamma^2\phi_j.
\end{align*}
For term $(\Psi_2+\meth\pi)\phi_0$ and $\phi_0\mu\rho$ we have
\begin{align*}
\int_0^v||\phi_0(a^{\frac{1}{2}}\meth)^i(\Psi_2+\meth\pi)||_{L^2_{sc}(\mathcal{S}_{u,v'})}&\lesssim
\frac{a^{\frac{1}{2}}}{|u|}\int_0^v||(a^{\frac{1}{2}}\mathcal{D})^i\Psi_2||_{L^2_{sc}(\mathcal{S}_{u,v'})}
+\frac{1}{a^{\frac{1}{2}}}\frac{a}{|u|}\int_0^v||a^{\frac{i}{2}}
\mathcal{D}^{i+1}\pi||_{L^2_{sc}(\mathcal{S}_{u,v'})}\\
\lesssim&\frac{a^{\frac{1}{2}}}{|u|}(\bmPsi[\Psi_2]+\Gamma[\pi])+
\frac{1}{a^{\frac{1}{2}}}\bmGamma_{11}[\pi]\lesssim1.
\end{align*}
Here the definition of $\bmGamma_{11}[\pi]$ will be introduced in \ref{EllipticEstGamma}.
\begin{align*}
\int_0^v||(a^{\frac{1}{2}}\meth)^i(\phi_0\mu\rho)||_{L^2_{sc}(\mathcal{S}_{u,v'})}\lesssim
\int_0^v\frac{1}{|u|^2}a^{\frac{1}{2}}\frac{|u|^2}{a}\mathcal{O}^3\lesssim1.
\end{align*}
Then we obtain that for $0\leq k\leq10$ the following holds
\begin{align*}
||(a^{\frac{1}{2}}\meth)^k\Tizeta_4||_{L^2_{sc}(\mathcal{S}_{u,v})}\lesssim&
1+\bmzeta[\zeta_3].
\end{align*}
By definition of $\Tizeta_4$ we have
\begin{align*}
\frac{a^{\frac{1}{2}}}{|u|}||(a^{\frac{1}{2}}\meth)^k\zeta_4||_{L^2_{sc}(\mathcal{S}_{u,v})}\lesssim&
\frac{a^{\frac{1}{2}}}{|u|}||(a^{\frac{1}{2}}\meth)^k\Tizeta_4||_{L^2_{sc}(\mathcal{S}_{u,v})}
+\frac{a^{\frac{1}{2}}}{|u|}||(a^{\frac{1}{2}}\meth)^k(\phi_0\mu)||_{L^2_{sc}(\mathcal{S}_{u,v})}\\
\lesssim&1+\frac{a^{\frac{1}{2}}}{|u|}\frac{1}{|u|}a^{\frac{1}{2}}\phi[\phi_0]\frac{|u|^2}{a}\Gamma[\mu]
\lesssim1+\underline{\bm{\phi}}[\phi_1].
\end{align*}

For $\zeta_2$ we make use of its outgoing equation \eqref{thornzeta2}
\begin{align*}
\mthorn\zeta_2=\meth'\zeta_1+\Gamma(\lambda,\rho,\sigma,\pi)\zeta_{0,1,2,4}
+\zeta_j\phi_l^2.
\end{align*}
Then we can compute 
\begin{align*}
\mthorn\meth^k\zeta_2=\meth^{k+1}\zeta_1+
\sum_{i_1+...+i_4=k}\meth^{i_1}\Gamma(\pi,\tau)^{i_2}\meth^{i_3}\Gamma\meth^{i_4}\zeta_l
+\sum_{i_1+...+i_5=k}\meth^{i_1}\Gamma(\pi,\tau)^{i_2}\meth^{i_3}\zeta_l
\meth^{i_4}\phi_{j_1}\meth^{i_5}\phi_{j_2},
\end{align*}
and obtain
\begin{align*}
||(a^{\frac{1}{2}}\mathcal{D})^k\zeta_2||_{L^2_{sc}(\mathcal{S}_{u,v})}\lesssim&
1+\bmzeta[\zeta_1]+\zeta[\zeta_4]\Gamma[\sigma] \\
\lesssim&\bm\Psi[\Psi_0]\underline{\bm{\phi}}[\phi_1]
+\bm\Psi[\Psi_0]+\underline{\bm{\phi}}[\phi_1]+\bmzeta[\zeta_1]+1.
\end{align*}

For $\zeta_3$ we first estimate its derivative up to 9 by using \eqref{thornprimezeta3alt},
\begin{align*}
\mthorn'\zeta_3+2\mu\zeta_3=\meth\Tizeta_4-\phi_0\meth\mu-\phi_1\mu\sigma
+\zeta_i\phi_j^2+\Gamma\zeta_i+\Gamma^2\phi_j,
\end{align*}
and obtain 
\begin{align*}
||(a^{\frac{1}{2}}\mathcal{D})^{\leq9}\zeta_3||_{L^2_{sc}(\mathcal{S}_{u,v})}\lesssim&1
+\int_{u_{\infty}}^u\frac{a}{|u|^2}||(a^{\frac{1}{2}}\mathcal{D})^{\leq9}F||_{L^2_{sc}(\mathcal{S}_{u',v})}\\
\lesssim&1+\int_{u_{\infty}}^u\frac{a}{|u|^2}\frac{1}{|u|^2}a^{\frac{1}{2}}
\frac{|u|^2}{a}a^{\frac{1}{2}}\phi[\phi_1]\Gamma[\mu]\Gamma[\sigma]+\frac{\mathcal{O}^3}{a^{\frac{1}{2}}}\\
\leq&1+\frac{a}{|u|}\phi[\phi_1]\Gamma[\mu]\Gamma[\sigma] \\
\lesssim&\bm\Psi[\Psi_0]\bmphi[\phi_0]+\bm\Psi[\Psi_0]+\bmphi[\phi_0]+1.
\end{align*}
We leave the analysis of its 10 derivative into the elliptic part; see Remark~\ref{Ellipticzeta3zeta1}.

For $\Tizeta_5$ we make use of \eqref{thornTizeta5},
\begin{align*}
\mthorn\Tizeta_5=\meth'\Tizeta_4+2\Psi_2\phi_1+2\phi_1\meth\pi-\phi_0\meth'\mu
+\zeta_i\phi_j^2+\Gamma\zeta_i+\Gamma^2\phi_j,
\end{align*}
and estimate its derivative up to 9. 
The reason cannot to 10 is we cannot obtain the 11 derivative of $\Tizeta_4$ on the light cone.
Follow the similar analysis we obtain
\begin{align*}
||(a^{\frac{1}{2}}\mathcal{D})^{\leq9}\Tizeta_5||_{L^2_{sc}(\mathcal{S}_{u,v})}\lesssim&
1+\frac{1}{a^{\frac{1}{2}}}\zeta[\Tizeta_4]+\frac{1}{a}\phi[\phi_0]\Gamma[\Timu]
+\frac{1}{a^{\frac{1}{2}}}\mathcal{O}^2\lesssim1.
\end{align*}
With such results we control $\zeta_5$ to 9 derivative by the definition of $\Tizeta_5$:
\begin{align*}
\frac{a^{\frac{1}{2}}}{|u|}||(a^{\frac{1}{2}}\meth)^{\leq9}\zeta_5||_{L^2_{sc}(\mathcal{S}_{u,v})}\lesssim&
\frac{a^{\frac{1}{2}}}{|u|}||(a^{\frac{1}{2}}\meth)^{\leq9}\Tizeta_5||_{L^2_{sc}(\mathcal{S}_{u,v})}
+\frac{a^{\frac{1}{2}}}{|u|}||(a^{\frac{1}{2}}\meth)^{\leq9}(\phi_1\mu)||_{L^2_{sc}(\mathcal{S}_{u,v})}\\
\lesssim&1+\frac{a^{\frac{1}{2}}}{|u|}\frac{1}{|u|}a^{\frac{1}{2}}\phi[\phi_1]\frac{|u|^2}{a}\Gamma[\mu]
\lesssim1+\underline{\bm{\phi}}[\phi_0].
\end{align*}

For the 10 derivative of $\zeta_5$ we make use of \eqref{thornzeta5},
\begin{align*}
\mthorn\zeta_5=\meth'\zeta_4-\zeta_1\mu+\zeta_i\phi_j^2+\Gamma\zeta_i+\Gamma^2\phi_j,
\end{align*}
and obtain
\begin{align*}
\frac{a^{\frac{1}{2}}}{|u|}||(a^{\frac{1}{2}}\meth)^k\zeta_5||_{L^2_{sc}(\mathcal{S}_{u,v})}\lesssim&1+
\frac{1}{|u|}\int_0^v||(a^{\frac{1}{2}}\meth)^{k+1}\zeta_4||_{L^2_{sc}(\mathcal{S}_{u,v})}
+\frac{1}{a^{\frac{1}{2}}}\mathcal{O}^2 \\
\lesssim&1+\bmzeta[\zeta_4].
\end{align*}

\end{proof}

\begin{proposition}
\label{L2Weyl}
For $0\leq k\leq9$, one has that 
\begin{align*}
\frac{1}{a^{\frac{1}{2}}}||(a^{\frac{1}{2}}\mathcal{D})^k\Psi_0||_{L^2_{sc}(\mathcal{S}_{u,v})}\lesssim&1, \\
||(a^{\frac{1}{2}}\mathcal{D})^k\Psi_{1,2,3,4}||_{L^2_{sc}(\mathcal{S}_{u,v})}\lesssim&
\bmPsi[\Psi_0]+\underline{\bmzeta}[\zeta_1]\bmphi[\phi_0]
+\underline{\bmzeta}[\zeta_1]\underline{\bm{\phi}}[\phi_1]
+\underline{\bmzeta}[\zeta_1]+\bmphi[\phi_0]+\underline{\bm{\phi}}[\phi_1]+1
\end{align*}
and
\begin{align*}
\frac{a}{|u|}||(a^{\frac{1}{2}}\mathcal{D})^kK||_{L^2_{sc}(\mathcal{S}_{u,v})}\lesssim&1.
\end{align*}
where $K$ is the Gaussian curvature.
\end{proposition}

\begin{proof}
For $\Psi_0$, we make use of \eqref{thornprimePsi0},
\begin{align*}
\mthorn'\Psi_0+\mu\Psi_0=\meth\TiPsi_1+2\mathrm{i}(\phi_0\meth\bar\zeta_1-\bar\phi_0\meth\zeta_3)+\zeta_i^2+\zeta_i\phi_j^3+\Psi_i\phi_j^2+\zeta_i\phi_j\Gamma+\phi_j^2\Gamma^2+\Gamma\Psi_i,
\end{align*}
and obtain
\begin{align*}
\frac{1}{a^{\frac{1}{2}}}||(a^{\frac{1}{2}}\meth)^k\Psi_0||_{L^2_{sc}(\mathcal{S}_{u,v})}\lesssim&
\frac{1}{a^{\frac{1}{2}}}||(a^{\frac{1}{2}}\meth)^k\Psi_0||_{L^2_{sc}(\mathcal{S}_{u_{\infty},v})}
+\frac{1}{a^{\frac{1}{2}}}\int_{u_{\infty}}^u\frac{a}{|u'|^2}||a^{\frac{k}{2}}F||_{L^2_{sc}(\mathcal{S}_{u',v})} \\
\lesssim&1+\frac{1}{|u|^{\frac{1}{2}}}\underline{\bmPsi}[\TiPsi_1]+\frac{a^{\frac{1}{2}}}{|u|}\mathcal{O}^2
\lesssim1.
\end{align*}

For $\Psi_{II}=\Psi_{1,2,3,4}$ we make use of their outgoing equation \eqref{thornPsi1}, \eqref{thornPsi2}, 
\eqref{thornPsi3}, \eqref{thornPsi4}:
\begin{align*}
\mthorn\Psi_{II}=\meth'(\Psi_0,\TiPsi)+\phi_j\meth\zeta_{0,...,4}+
\zeta_i^2+\zeta_i\phi_j^3+\Psi_i\phi_j^2+\zeta_i\phi_j\Gamma+\phi_j^2\Gamma^2+\Gamma\Psi_i.
\end{align*}
The large terms are $\mu\zeta_0\phi_i$ which can be estimated in the following:
\begin{align*}
\int_0^v||a^{\frac{k}{2}}\mu\meth^i\zeta_0\meth^{k-i}\phi_j||_{L^2_{sc}(\mathcal{S}_{u,v})}\lesssim&
\frac{1}{|u|^2}\frac{|u|^2}{a}a^{\frac{1}{2}}\zeta[\zeta_0]a^{\frac{1}{2}}\phi[\phi_j]\\
\lesssim&\underline{\bmzeta}[\zeta_1]\bmphi[\phi_0]
+\underline{\bmzeta}[\zeta_1]\underline{\bm{\phi}}[\phi_1]
+\underline{\bmzeta}[\zeta_1]+\bmphi[\phi_0]+\underline{\bm{\phi}}[\phi_1]+1.
\end{align*}
According to the results in vacuum case we obtain

\begin{align*}
||(a^{\frac{1}{2}}\meth)^k\Psi_{II}||_{L^2_{sc}(\mathcal{S}_{u,v})}\lesssim&
\bmPsi[\Psi_0]+\underline{\bmzeta}[\zeta_1]\bmphi[\phi_0]
+\underline{\bmzeta}[\zeta_1]\underline{\bm{\phi}}[\phi_1]
+\underline{\bmzeta}[\zeta_1]+\bmphi[\phi_0]+\underline{\bm{\phi}}[\phi_1]+1.
\end{align*}

For the analysis of Gaussian curvature, 
one define the complex Gaussian curvature:
\begin{align*}
\mathcal{K}\equiv&\Lambda+\Phi_{11}-\Psi_2+\mu\rho-\lambda\sigma \\
=&\mathrm{i}\bar\zeta_1\phi_1-\mathrm{i}\zeta_1\bar\phi_1
+\mathrm{i}\bar\Tizeta_4\phi_0-\mathrm{i}\Tizeta_4\bar\phi_0
-\Psi_2+\mu\rho-\lambda\sigma
\end{align*}
where the Gaussian curvature is $K=\mathcal{K}+\bar{\mathcal{K}}$. 
The initial data of $\mathcal{K}$ on the ingoing cone $\mathcal{N}'_{\star}$ is $\frac{1}{|u|^2}$, 
we anticipate that Gaussian curvature keep such property, 
hence considering the signature, 
we has that the bootstrap assumption for $\mathcal{K}$ is 
\begin{align*}
\frac{a}{|u|}||(a^{\frac{1}{2}}\mathcal{D})^k\mathcal{K}||_{L^2_{sc}(\mathcal{S}_{u,v})}\leq\mathcal{O}.
\end{align*}
Make use of its' ingoing equation
\begin{align*}
\mthorn'\mathcal{K}=&-2\mu\mathcal{K}-\meth\TiPsi_3
+\mathrm{i}(\phi_1\meth'\bar\zeta_2+\bar\phi_0\meth\zeta_5
-\bar\phi_1\meth\zeta_2-\phi_1\meth\bar\zeta_4)\\
&+\mu\meth'\tau-\lambda\meth\tau+\zeta_i^2+\zeta_i\phi_j^3
+\Psi_i\phi_j^2+\zeta_i\phi_j\Gamma+\phi_j^2\Gamma^2+\Gamma\Psi_i.
\end{align*}
Then we have
\begin{align*}
\mthorn'\meth^k\mathcal{K}+(k+2)\mu\meth^k\mathcal{K}=&-\meth^{k+1}\TiPsi_3
+\sum_{i=0}^k\meth^i\phi_j\meth^{k+1-i}\zeta_l
+\sum_{i=0}^k\meth^i\Gamma(\mu,\lambda)\meth^{k+1-i}\tau \\
&+\sum_{i=0}^k\meth^i\Gamma(\Timu,\lambda)\meth^{k-i}\mathcal{K}
+\meth^k(rest).
\end{align*}
We denote the right hand side by $F$ and
apply the ingoing transport inequality, one has
\begin{align*}
\frac{a}{|u|}||(a^{\frac{1}{2}}\meth)^k\mathcal{K}||_{L^2_{sc}(\mathcal{S}_{u,v})}\lesssim&
\frac{a}{|u_{\infty}|}||(a^{\frac{1}{2}}\meth)^k\mathcal{K}||_{L^2_{sc}(\mathcal{S}_{u_{\infty},v})}
+\int_{u_{\infty}}^u\frac{a^2}{|u'|^3}||a^{\frac{k}{2}}F||_{L^2_{sc}(\mathcal{S}_{u',v})}. 
\end{align*}
Apply the results and the bootstrap assumptions, one has
\begin{align*}
\frac{a}{|u|}||(a^{\frac{1}{2}}\meth)^k\mathcal{K}||_{L^2_{sc}(\mathcal{S}_{u,v})}\lesssim&1.
\end{align*}

\end{proof}

\subsection{Summary of $L^2$}
Under the bootstrap assumption
\begin{align*}
\bmGamma,\bmphi, \bmzeta,\bmPsi\leq\mathcal{O},
\end{align*}
we obtain the following improvement for $\bmGamma$: 
For $0\leq k\leq10$, we have
\begin{align*}
\frac{a^{\frac{1}{2}}}{|u|}||(a^{\frac{1}{2}}\mathcal{D})^k\lambda||_{L^2_{sc}(\mathcal{S}_{u,v})}\lesssim&1, \quad
\frac{1}{a^{\frac{1}{2}}}||(a^{\frac{1}{2}}\mathcal{D})^{k}\sigma||_{L^2_{sc}(\mathcal{S}_{u,v})}\lesssim
\bm\Psi[\Psi_0]+1, \\
||(a^{\frac{1}{2}}\mathcal{D})^{k}\ulomega||_{L^2_{sc}(\mathcal{S}_{u,v})}\lesssim&
\bm\Psi[\Psi_2]+1, \quad
||(a^{\frac{1}{2}}\mathcal{D})^k\pi||_{L^2_{sc}(\mathcal{S}_{u,v})}\lesssim
\underline{\bm\Psi}[\TiPsi_3]+\bm\Psi[\TiPsi_1]+1.\\
||(a^{\frac{1}{2}}\mathcal{D})^{k}\tau||_{L^2_{sc}(\mathcal{S}_{u,v})}\lesssim&
\underline{\bmzeta}[\zeta_1]\bmphi[\phi_0]+\Psi[\TiPsi_1]+\bmphi[\phi_0]+\underline{\bmzeta}[\zeta_1]+1,\\
||(a^{\frac{1}{2}}\mathcal{D})^k\rho||_{L^2_{sc}(\mathcal{S}_{u,v})}\lesssim&
\bm\Psi[\Psi_0]^2+\underline{\bm{\phi}}[\phi_1]\underline{\bmzeta}[\zeta_1]+\bm\Psi[\Psi_0]
+\underline{\bmzeta}[\zeta_1]+\underline{\bm{\phi}}[\phi_1]+1, \\
\frac{a}{|u|}||(a^{\frac{1}{2}}\mathcal{D})^k\Timu||_{L^2_{sc}(\mathcal{S}_{u,v})}\lesssim&
\bm\Psi[\Psi_2]+1, \quad
\frac{a}{|u|^2}||(a^{\frac{1}{2}}\mathcal{D})^k\mu||_{L^2_{sc}(\mathcal{S}_{u,v})}\lesssim1.
\end{align*}

\begin{align*}
\frac{1}{a^{\frac{1}{2}}}||(a^{\frac{1}{2}}\mathcal{D})^k\phi_0||_{L^2_{sc}(\mathcal{S}_{u,v})}\lesssim&
\underline{\bm{\phi}}[\phi_1]+1, \quad
\frac{1}{a^{\frac{1}{2}}}||(a^{\frac{1}{2}}\mathcal{D})^k\phi_1||_{L^2_{sc}(\mathcal{S}_{u,v})}\lesssim\bmphi[\phi_0]+1.
\end{align*}

\begin{align*}
\frac{1}{a^{\frac{1}{2}}}||(a^{\frac{1}{2}}\mathcal{D})^k\zeta_0||_{L^2_{sc}(\mathcal{S}_{u,v})}\lesssim&
\underline{\bmzeta}[\zeta_1]+1, \quad
||(a^{\frac{1}{2}}\mathcal{D})^k\zeta_1||_{L^2_{sc}(\mathcal{S}_{u,v})}\lesssim
\underline{\bmzeta}[\zeta_2]+1, \\
||(a^{\frac{1}{2}}\mathcal{D})^k\Tizeta_4||_{L^2_{sc}(\mathcal{S}_{u,v})}\lesssim&
1+\bmzeta[\zeta_3], \quad
\frac{a^{\frac{1}{2}}}{|u|}||(a^{\frac{1}{2}}\mathcal{D})^k\zeta_4||_{L^2_{sc}(\mathcal{S}_{u,v})}\lesssim
1+\underline{\bm{\phi}}[\phi_1], \\
||(a^{\frac{1}{2}}\mathcal{D})^k\zeta_2||_{L^2_{sc}(\mathcal{S}_{u,v})}\lesssim&
\bm\Psi[\Psi_0]\underline{\bm{\phi}}[\phi_1]
+\bm\Psi[\Psi_0]+\underline{\bm{\phi}}[\phi_1]+\bmzeta[\zeta_1]+1, \\
\frac{a^{\frac{1}{2}}}{|u|}||(a^{\frac{1}{2}}\mathcal{D})^k\zeta_5||_{L^2_{sc}(\mathcal{S}_{u,v})}\lesssim&
1+\bmzeta[\zeta_4].
\end{align*}
For $0\leq k\leq9$, we have 
\begin{align*}
||(a^{\frac{1}{2}}\mathcal{D})^k\zeta_3||_{L^2_{sc}(\mathcal{S}_{u,v})}\lesssim&\bm\Psi[\Psi_0]\bmphi[\phi_0]+\bm\Psi[\Psi_0]+\bmphi[\phi_0]+1, \\
||(a^{\frac{1}{2}}\mathcal{D})^k\Tizeta_5||_{L^2_{sc}(\mathcal{S}_{u,v})}\lesssim&1, \quad
\frac{a^{\frac{1}{2}}}{|u|}||(a^{\frac{1}{2}}\mathcal{D})^k\zeta_5||_{L^2_{sc}(\mathcal{S}_{u,v})}\lesssim
1+\underline{\bm{\phi}}[\phi_0].
\end{align*}

\begin{align*}
\frac{1}{a^{\frac{1}{2}}}||(a^{\frac{1}{2}}\meth)^k\Psi_0||_{L^2_{sc}(\mathcal{S}_{u,v})}\lesssim&1, \\
||(a^{\frac{1}{2}}\meth)^k\Psi_{1,2,3,4}||_{L^2_{sc}(\mathcal{S}_{u,v})}\lesssim&
\bmPsi[\Psi_0]+\underline{\bmzeta}[\zeta_1]\bmphi[\phi_0]
+\underline{\bmzeta}[\zeta_1]\underline{\bm{\phi}}[\phi_1]
+\underline{\bmzeta}[\zeta_1]+\bmphi[\phi_0]+\underline{\bm{\phi}}[\phi_1]+1,
\end{align*}

\begin{align*}
\frac{a}{|u|}||(a^{\frac{1}{2}}\mathcal{D})^kK||_{L^2_{sc}(\mathcal{S}_{u,v})}\lesssim&1.
\end{align*}

\begin{remark}
\label{AdvancedGauss}
Actually, we also anticipate that the derivative of $\mathcal{K}$ has "better" property 
by eliminating $\frac{a}{|u|}$ which comes from the initial data. 
When $1\leq k\leq9$, we can make use of results for $\rho$, $\mu$, $\lambda$, $\sigma$, $\Psi_2$, 
$\phi_j$, $\zeta_1$ and $\Tizeta_4$ obtained. 
For those terms under next-to-leading order derivative, 
we have better control.

For the k derivative of $\zeta_1$ and $\phi_1$ where $0\leq k\leq10$, i.e. the next to leading order derivative, the norm on the light cone can be bounded by its norm on $\mathcal{S}$, we have 
\begin{align*}
\underline{\bmzeta}[\zeta_1]_k=&\frac{1}{a^{\frac{1}{2}}}\int_{u_{\infty}}^u\frac{a}{|u'|^2}
||a^{\frac{k-1}{2}}\mathcal{D}^{k}\zeta_1||_{L^2_{sc}(\mathcal{S}_{u',v})}
\leq\frac{\mathcal{O}}{a}\lesssim1 ,\\
\underline{\bm{\phi}}[\phi_1]_k=&\frac{1}{a^{\frac{1}{2}}}\int_{u_{\infty}}^u\frac{a}{|u'|^2}
||a^{\frac{k-1}{2}}\mathcal{D}^{k}\phi_1||_{L^2_{sc}(\mathcal{S}_{u',v})}
\leq\frac{\mathcal{O}}{a^{\frac{1}{2}}}\lesssim1.
\end{align*}
Same analysis leads to 
\begin{align*}
\underline{\bmzeta}[\zeta_2]_k=&\int_{u_{\infty}}^u\frac{a}{|u'|^2}
||a^{\frac{k-1}{2}}\mathcal{D}^{k}\zeta_2||_{L^2_{sc}(\mathcal{S}_{u',v})}
\leq\frac{\mathcal{O}}{a^{\frac{1}{2}}}\lesssim1 ,\\
\bmzeta[\zeta_1]_k=&\int_{0}^v
||a^{\frac{k-1}{2}}\mathcal{D}^{k}\zeta_1||_{L^2_{sc}(\mathcal{S}_{u',v})}
\leq\frac{\mathcal{O}}{a^{\frac{1}{2}}}\lesssim1 ,\\
\bm{\phi}[\phi_0]_k=&\frac{1}{a^{\frac{1}{2}}}\int_{0}^v
||a^{\frac{k-1}{2}}\mathcal{D}^{k}\phi_0||_{L^2_{sc}(\mathcal{S}_{u',v})}
\leq\frac{\mathcal{O}}{a^{\frac{1}{2}}}\lesssim1.
\end{align*}

such results lead to 
\begin{align*}
||(a^{\frac{1}{2}}\mathcal{D})^{k\leq9}\Psi_{1,2}||_{L^2_{sc}(\mathcal{S}_{u,v})}\lesssim1, \quad
||(a^{\frac{1}{2}}\mathcal{D})^{k\leq9}\zeta_{1}||_{L^2_{sc}(\mathcal{S}_{u,v})}\lesssim1, \quad
\frac{a^{\frac{1}{2}}}{|u|}||(a^{\frac{1}{2}}\mathcal{D})^{k\leq9}\zeta_5||_{L^2_{sc}(\mathcal{S}_{u,v})}\lesssim1.
\end{align*}
Moreover for the derivative of $\Psi_0$ up to 9 derivative we have
\begin{align*}
\bm\Psi[\Psi_0]_k\lesssim1
\end{align*}
which leads to 
\begin{align*}
\Gamma[\sigma]_{0\leq k\leq9}\lesssim&1, \quad
\Gamma[\rho]_0\lesssim1, \Gamma[\rho]_{1\leq k\leq9}\lesssim\frac{a}{|u|}, \\
||(a^{\frac{1}{2}}\mathcal{D})^{0\leq k\leq9}\zeta_{2,3}||_{L^2_{sc}(\mathcal{S}_{u,v})}\lesssim&1,
\end{align*}

Then from the definition of the Gaussian curvature we can calculate directly and obtain: 
\begin{align*}
\sum_{i=1}^7||(a^{\frac{1}{2}}\mathcal{D})^iK||_{L_{sc}^{\infty}(\mathcal{S}_{u,v})}
+\sum_{i=1}^9||(a^{\frac{1}{2}}\mathcal{D})^iK||_{L_{sc}^{2}(\mathcal{S}_{u,v})}\lesssim1.
\end{align*}

\end{remark}


\section{Elliptic}
\label{Elliptic}
In this section we estimate the top-derivative, i.e. the 11th derivative of connections 
which are needed in the energy estimate for curvature as well as the matter fields. 
For the purpose of closing the bootstrap argument, we sidestep top-order Weyl curvature estimates. 
The strategy is to construct new quantities made of the divergence or the curl of connections. 
Then the elliptic inequalities lead to that the estimate of higher derivative can be obtained 
from the lower derivative of such quantities. 
By doing so, one can avoid the estimate of 11-derivative of curvature. 
We start with introducing the basic elliptic inequalities in the scale-invariant norms:

\begin{proposition}
\label{EllipticT-weightSC}
Let $f$ denote a quantity with nonzero
T-weight. Then under the 
bootstrap assumption, for $0\leq k\leq 11$, one has that 
\begin{align*}
||(a^{\frac{1}{2}}\mathcal{D})^kf||_{L^2_{sc}(\mathcal{S}_{u,v})}\lesssim\sum_{i\leq k-1}
\left(||(a^{\frac{1}{2}})^{i+1}\mathcal{D}^i\mathscr{D}_f||_{L_{sc}^2(\mathcal{S}_{u,v})}
+||(a^{\frac{1}{2}}\mathcal{D})^if||_{L_{sc}^2(\mathcal{S}_{u,v})}
\right).
\end{align*}
\end{proposition}
The proof can be found in Sec 6.1 in Paper \cite{An202209}.

\medskip
When $f$ is a pure scalar, we have the following:
\begin{proposition}
\label{EllipticPureScalar}
Let $f$ denote a quantity with zero T-weight. Then, under the 
bootstrap assumption, for $0\leq k\leq 10$, one has that 
\begin{align*}
||(a^{\frac{1}{2}}\mathcal{D})^{k+1}f||_{L^2_{sc}(\mathcal{S})}\lesssim
||(a^{\frac{1}{2}})^{k+1}\mathcal{D}^{k-1}(\bmDelta f)||_{L^2_{sc}(\mathcal{S})}
+\sum_{i=0}^k||(a^{\frac{1}{2}}\mathcal{D})^if||_{L^2_{sc}(\mathcal{S})},
\end{align*}
where $\bmDelta f\equiv 2\meth\meth'f$.
\end{proposition}

With the Prop. \ref{EllipticT-weightSC} we have the following remark:
\begin{remark}
\label{Ellipticzeta3zeta1}
From the equation \eqref{zeta3zeta1} which gives the relation between $\meth'\zeta_3$ and $\meth\zeta_1$, 
we have 
\begin{align*}
||(a^{\frac{1}{2}}\mathcal{D})^k\zeta_3||_{L^2_{sc}(\mathcal{S}_{u,v})}\lesssim&\sum_{i\leq k-1}
\left(||(a^{\frac{1}{2}})^{i+1}\mathcal{D}^i\meth'\zeta_3||_{L_{sc}^2(\mathcal{S}_{u,v})}
+||(a^{\frac{1}{2}}\mathcal{D})^i\zeta_3||_{L_{sc}^2(\mathcal{S}_{u,v})} \right) \\
\lesssim&||(a^{\frac{1}{2}}\mathcal{D})^k\zeta_1||_{L^2_{sc}(\mathcal{S}_{u,v})}+
\phi[\phi_{0,1}]\sum_{i=0}^{k-1}||(a^{\frac{1}{2}}\mathcal{D})^i\TiPsi_{1,2}||_{L^2_{sc}(\mathcal{S}_{u,v})} \\
&+\zeta[\zeta_4]\Gamma[\rho]+\zeta[\zeta_2]\Gamma[\sigma]
+\sum_{i=0}^{k-1}||(a^{\frac{1}{2}}\mathcal{D})^i\zeta_{1,3}||_{L^2_{sc}(\mathcal{S}_{u,v})}.
\end{align*}
When $k=10$, apply the estimate in previous we have
\begin{align*}
||(a^{\frac{1}{2}}\mathcal{D})^{10}\zeta_3||_{L^2_{sc}(\mathcal{S}_{u,v})}\lesssim&
||(a^{\frac{1}{2}}\mathcal{D})^{10}\zeta_1||_{L^2_{sc}(\mathcal{S}_{u,v})}+1
\lesssim1+\underline{\bmzeta}[\zeta_2].
\end{align*}
For the norm on the light cone, we have 
\begin{align*}
||a^{\frac{k-1}{2}}\mathcal{D}^{k}
\zeta_3||_{L^2_{sc}(\mathcal{N}_{u}(0,v))}\lesssim&
||a^{\frac{k-1}{2}}\mathcal{D}^{k}
\zeta_1||_{L^2_{sc}(\mathcal{N}_{u}(0,v))}+1, \\
\frac{1}{a^{\frac{1}{2}}}||a^{\frac{k-1}{2}}\mathcal{D}^{k}
\zeta_3||_{L^2_{sc}(\mathcal{N}'_{v}(u_{\infty},u))}\lesssim&
\frac{1}{a^{\frac{1}{2}}}||a^{\frac{k-1}{2}}\mathcal{D}^{k}
\zeta_1||_{L^2_{sc}(\mathcal{N}'_{v}(u_{\infty},u))}+1.
\end{align*}
The above results show the equivalence of the norm between $\zeta_3$ and $\zeta_1$.
\end{remark}

\subsection{Elliptic estimates for 11 derivative}
\label{EllipticEstGamma}
In this section we estimate the 11-derivative of connections. 
Again we focus on terms contain matter field.
The norm of the eleventh derivative quantities of the connection coefficients is defined as 
\begin{align*}
\bm{\Gamma}_{11}(u,v)\equiv&
||a^5\mathcal{D}^{11}\sigma||_{L^2_{sc}(\mathcal{N}_u(0,v))}
+||a^5\mathcal{D}^{11}\rho||_{L^2_{sc}(\mathcal{N}_u(0,v))} \\
&+||a^5\mathcal{D}^{11}\tau||_{L^2_{sc}(\mathcal{N}'_v(u_{\infty},u))}
+||a^5\mathcal{D}^{11}\tau||_{L^2_{sc}(\mathcal{N}_u(0,v))}\\
&+\frac{a}{|u|}||a^5\mathcal{D}^{11}\pi||_{L^2_{sc}(\mathcal{N}_u(0,v))}
+||a^5\mathcal{D}^{11}\ulomega||_{L^2_{sc}(\mathcal{N}'_v(u_{\infty},u))}\\
&+\int_{u_{\infty}}^u\frac{a^{\frac{1}{2}}}{|u|^3}||a^5\mathcal{D}^{11}\lambda||_{L_{sc}^2(\mathcal{S}_{u,v})}
+\int_{u_{\infty}}^u\frac{a^2}{|u|^3}||a^5\mathcal{D}^{11}\mu||_{L_{sc}^2(\mathcal{S}_{u,v})}
\end{align*}

In the following estimates, we make use of the $L^2(\mathcal{S})$ estimates and
the bootstrap assumptions 
\begin{align*}
 \bm{\phi},\bm{\zeta}, \bm{\Psi}, \bm{\Gamma}_{11}\leq\mathcal{O},
\end{align*}
We will show that 
\[
  \bmGamma_{11}\lesssim\bm\Psi^2+\bm\Psi+\bm\phi^2+\bm\phi+\bmzeta^2+\bmzeta+1.
  \]

\begin{proposition}
  \label{11Derrhosigma}
  One has that 
\begin{align*}
||a^5\mathcal{D}^{11}\rho||_{L^2_{sc}(\mathcal{S}_{u,v})}\lesssim&
\bm{\phi}[\phi_0]+\bm{\zeta}[\zeta_0]
+\bm\Psi[\Psi_0]+\bm\Psi[\TiPsi_1]+1,\\
||a^5\mathcal{D}^{11}\rho||_{L^2_{sc}(\mathcal{N}_u(0,v))}\lesssim&
\bm{\phi}[\phi_0]+\bm{\zeta}[\zeta_0]
+\bm\Psi[\Psi_0]+\bm\Psi[\TiPsi_1]+1,
\end{align*}
and
\begin{align*}
||a^5\mathcal{D}^{11}\sigma||_{L^2_{sc}(\mathcal{N}_u(0,v))}
\lesssim\bm{\phi}[\phi_0]+\bm{\zeta}[\zeta_0]
+\bm\Psi[\Psi_0]+\bm\Psi[\TiPsi_1]+1.
\end{align*}
\end{proposition}

\begin{proof}
For $\rho$ we make use of \eqref{Thornrho}
\begin{align*}
\mthorn\,\rho &= 2\,\mathrm{i}\,\bar\zeta_0\,\phi_0
- 2\,\mathrm{i}\,\zeta_0\,\bar\phi_0
+ \rho^2
+ \sigma\,\bar\sigma,  
\end{align*}
and focus on the matter field terms $\zeta_0\phi_0$ and have
\begin{align*}
\int_0^va^5||\bar\phi_0\meth^{11}\zeta_0||_{L^2_{sc}(\mathcal{S}_{u,v'})}\leq&
\phi[\phi_0]\int_0^v\frac{a^{\frac{1}{2}}}{|u|}
||a^5\mathcal{D}^{11}\zeta_0||_{L^2_{sc}(\mathcal{S}_{u,v'})} \\
\leq&\frac{a}{|u|}\phi[\phi_0]\frac{1}{a^{\frac{1}{2}}}
\left(\int_0^v||a^5\mathcal{D}^{11}\zeta_0||^2_{L^2_{sc}(\mathcal{S}_{u,v'})} \right)^{\frac{1}{2}}\\
\leq&\frac{a}{|u|}\phi[\phi_0]\bmzeta[\zeta_0]
\lesssim\frac{a}{|u|}\bmzeta[\zeta_0],
\end{align*}
and similarly we have
\begin{align*}
\int_0^va^5||\bar\zeta_0\meth^{11}\phi_0||_{L^2_{sc}(\mathcal{S}_{u,v'})}\leq&
\zeta[\zeta_0]\int_0^v\frac{a^{\frac{1}{2}}}{|u|}
||a^5\mathcal{D}^{11}\phi_0||_{L^2_{sc}(\mathcal{S}_{u,v'})} \\
\leq&\frac{a}{|u|}\zeta[\zeta_0]\bm\phi[\phi_0]
\lesssim\frac{a}{|u|}\bm\phi[\phi_0].
\end{align*}
Then we have
\begin{align}
\label{Ellipticrhomid}
||a^5\mathcal{D}^{11}\rho||_{L^2_{sc}(\mathcal{S}_{u,v})}\lesssim
\frac{a^{\frac{1}{2}}}{|u|}\int_0^v
||a^5\mathcal{D}^{11}\sigma||_{L^2_{sc}(\mathcal{S}_{u,v})}+
\frac{a}{|u|}\bm{\phi}[\phi_0]+\frac{a}{|u|}\bmzeta[\zeta_0].
\end{align}
Then make use of $\meth\rho-\meth'\sigma=\tau\rho-\bar\tau\sigma-\TiPsi_1$, 
here $\meth'\sigma=\mathscr{D}_{\sigma}$ and the elliptic inequality we have
\begin{align*}
||a^5\mathcal{D}^{11}\sigma||_{L^2_{sc}(\mathcal{S}_{u,v})}\lesssim&\sum_{i\leq10}
(\frac{1}{a^{\frac{1}{2}}}||(a^{\frac{1}{2}}\mathcal{D})^{i}\sigma||_{L^2_{sc}(\mathcal{S}_{u,v})}
+||(a^{\frac{1}{2}}\mathcal{D})^{i}\mathscr{D}_{\sigma}||_{L^2_{sc}(\mathcal{S}_{u,v})}) \\
\lesssim&||a^5\mathcal{D}^{11}\rho||_{L^2_{sc}(\mathcal{S}_{u,v})}+
\sum_{i=0}^{10}||(a^{\frac{1}{2}}\mathcal{D})^{i}\TiPsi_1||_{L^2_{sc}(\mathcal{S}_{u,v})}
+\bm\Psi[\Psi_0]+1 \\
\lesssim&\frac{a^{\frac{1}{2}}}{|u|}\int_0^v
||a^5\mathcal{D}^{11}\sigma||_{L^2_{sc}(\mathcal{S}_{u,v})}
+\sum_{i=0}^{10}||(a^{\frac{1}{2}}\mathcal{D})^{i}\TiPsi_1||_{L^2_{sc}(\mathcal{S}_{u,v})} \\
&+\bm\Psi[\Psi_0]+\bm{\phi}[\phi_0]+\bmzeta[\zeta_0]+1 \\
\lesssim&\sum_{i=0}^{10}||(a^{\frac{1}{2}}\mathcal{D})^{i}\TiPsi_1||_{L^2_{sc}(\mathcal{S}_{u,v})} 
+\bm\Psi[\Psi_0]+\bm{\phi}[\phi_0]+\bmzeta[\zeta_0]+1.
\end{align*}
Integrating along the light cone one obtains the results.

\end{proof}

\begin{proposition}
\label{11Dertau}
\begin{align*}
||a^5\mathcal{D}^{11}\tau||_{L^2_{sc}(\mathcal{N}_u(0,v))}
\lesssim\bm\Psi[\Psi_0]+\bm\Psi[\Psi_2]
+\underline{\bmzeta}[\zeta_1]+\underline{\bmphi}[\phi_1]+1,
\end{align*}
\begin{align*}
||a^5\mathcal{D}^{11}\tau||_{L^2_{sc}(\mathcal{N}'_v(u_{\infty},u))}
\lesssim\bm\Psi[\Psi_0]+\underline{\bm\Psi}[\Psi_2]
+\underline{\bmzeta}[\zeta_1]+\underline{\bmphi}[\phi_1]+1.
\end{align*}
\end{proposition}

\begin{proof}
Define
\begin{align*}
\tilde{\tau}\equiv\meth'\tau-\Psi_2,
\end{align*}
we can obtain the equation for $\tilde\tau$ \eqref{thornTitau}:
\begin{align*}
\mthorn\tilde{\tau}=2\mathrm{i}(\bar\zeta_0\phi_0\mu-\zeta_0\bar\phi_0\mu
+\phi_0\meth'\bar\zeta_1-\bar\phi_0\meth\zeta_1)+\zeta_i\phi_j^3+\zeta_i^2
+\Psi\phi_j^2+\Gamma\zeta_i\phi_j+\Gamma\meth\Gamma+\Gamma\Psi+\Gamma^3
\end{align*}
where we list the leading terms and top-derivative terms and then obtain
\begin{align*}
||(a^{\frac{1}{2}}\meth)^{10}\tilde{\tau}||_{L^2_{sc}(\mathcal{S}_{u,v})}\lesssim&
\bmPsi[\Psi_0]+\zeta[\zeta_0]\phi[\phi_0]
+\int_0^va^5||\phi_0\meth^{11}\bar\zeta_1||_{L^2_{sc}(\mathcal{S}_{u,v'})}+1+\frac{\mathcal{O}}{a^{\frac{1}{2}}}  \\
\lesssim&\bmPsi[\Psi_0]
+\underline{\bmzeta}[\zeta_1]+\underline{\bmphi}[\phi_1]+\frac{a^{\frac{1}{2}}}{|u|}\Gamma[\phi_0]
\int_0^v||a^5\mathcal{D}^{11}\zeta_1||_{L^2_{sc}(\mathcal{S}_{u,v'})}+1+\frac{\mathcal{O}}{a^{\frac{1}{2}}} \\
\lesssim&\bmPsi[\Psi_0]
+\underline{\bmzeta}[\zeta_1]+\underline{\bmphi}[\phi_1]+1.
\end{align*}

Next we make use of the definition of $\tilde\tau$ and $\mathscr{D}_{\tau}=\meth'\tau$ and the 
elliptic inequality we obtain
\begin{align*}
||a^5\mathcal{D}^{11}\tau||_{L^2_{sc}(\mathcal{N}_u(0,v))}
\lesssim\bm\Psi[\Psi_0]+\bm\Psi[\Psi_2]
+\underline{\bmzeta}[\zeta_1]+\underline{\bmphi}[\phi_1]+1,
\end{align*}
and
\begin{align*}
||a^5\mathcal{D}^{11}\tau||_{L^2_{sc}(\mathcal{N}'_v(u_{\infty},u))}
\lesssim\bm\Psi[\Psi_0]+\underline{\bm\Psi}[\Psi_2]
+\underline{\bmzeta}[\zeta_1]+\underline{\bmphi}[\phi_1]+1.
\end{align*}

\end{proof}

\begin{proposition}
\label{11Derpi}
\begin{align*}
\frac{a}{|u|}||a^5\mathcal{D}^{11}\pi||_{L^2_{sc}(\mathcal{N}_{u}(0,v))}\lesssim&
\bm\Psi[\Psi_0]+\underline{\bm\Psi}[\Psi_2]+\bm\Psi[\Psi_2]
+\underline{\bmzeta}[\zeta_1]+\underline{\bmphi}[\phi_1]+1.
\end{align*}
\end{proposition}

\begin{proof}
Define
\begin{align*}
\tilde\pi=\meth\pi+\Psi_2
\end{align*}
we can obtain its equation \eqref{thornprimeTipi}:
\begin{align*}
\mthorn'\tilde\pi=2\mathrm{i}(\bar\phi_1\meth\zeta_2-\phi_1\meth'\bar\zeta_2)
+\zeta_i\phi_j^3+\zeta_i^2
+\Psi\phi_j^2+\Gamma\zeta_i\phi_j+\Gamma\meth\Gamma+\Gamma\Psi+\Gamma^3
\end{align*}
and obtain
\begin{align*}
\frac{a}{|u|}||(a^{\frac{1}{2}}\meth)^{10}\tilde{\pi}||_{L^2_{sc}(\mathcal{S}_{u,v})}\lesssim&
\int_{u_{\infty}}^u\frac{a^2}{|u'|^3}||\phi_1a^5\mathcal{D}^{11}\zeta_2||_{L^2_{sc}(\mathcal{S}_{u',v})}+
\frac{\mathcal{O}^2}{a^{\frac{1}{2}}}+Vac.
\end{align*}

For the integral we have
\begin{align*}
\int_{u_{\infty}}^u\frac{a^2}{|u'|^3}||a^5\bar\phi_1\meth^{11}\zeta_2||_{L^2_{sc}(\mathcal{S}_{u',v})}\lesssim&\phi[\phi_1]
\int_{u_{\infty}}^u\frac{a^{\frac{5}{2}}}{|u'|^4}
||a^5\meth^{11}\zeta_2||_{L^2_{sc}(\mathcal{S}_{u',v})} \\
\leq&\phi[\phi_1]\left(\int_{u_{\infty}}^u\frac{a^{4}}{|u'|^6}\right)^{\frac{1}{2}}
\left(\int_{u_{\infty}}^u\frac{a}{|u'|^2}||a^{5}\mathcal{D}^{11}\zeta_2||^2_{L^2_{sc}(\mathcal{S}_{u',v})}\right)^{\frac{1}{2}}\\
\leq&\frac{a^{2}}{|u|^{\frac{5}{2}}}\phi[\phi_1]\underline{\bmzeta}[\zeta_2]\lesssim1
\end{align*}
and the Vacuum terms give control
$\bmGamma_{11}[\tau]+\underline{\bmPsi}[\Psi_2]+1$. 
Make use of the control for $\tau$ we obtain
\begin{align*}
\frac{a}{|u|}||(a^{\frac{1}{2}}\meth)^{10}\tilde{\pi}||_{L^2_{sc}(\mathcal{S}_{u,v})}\lesssim&
\bm\Psi[\Psi_0]+\underline{\bm\Psi}[\Psi_2]
+\underline{\bmzeta}[\zeta_1]+\underline{\bmphi}[\phi_1]+1,
\end{align*}
which leads to 
\begin{align*}
\frac{a}{|u|}||a^5\mathcal{D}^{11}\pi||_{L^2_{sc}(\mathcal{N}_{u}(0,v))}\lesssim&
\bm\Psi[\Psi_0]+\underline{\bm\Psi}[\Psi_2]+\bm\Psi[\Psi_2]
+\underline{\bmzeta}[\zeta_1]+\underline{\bmphi}[\phi_1]+1.
\end{align*}

\end{proof}

\begin{proposition}
\label{11Derulomega}
\begin{align*}
||a^5\mathcal{D}^{11}\ulomega||_{L^2_{sc}(\mathcal{N}'_v(u_{\infty},u))}\lesssim&
\bmphi[\phi_0]+\underline{\bmzeta}[\zeta_1]+\bmzeta[\zeta_4]
+\underline{\bm\Psi}[\TiPsi_3]+1, \\
||a^5\mathcal{D}^{11}\ulomega||_{L^2_{sc}(\mathcal{N}_u(0,v))}\lesssim&
\bmphi[\phi_0]+\underline{\bmzeta}[\zeta_1]+\bmzeta[\zeta_4]
+\bm\Psi[\TiPsi_3]+1. 
\end{align*}
\end{proposition}

\begin{proof}
First we construct an auxiliary function $\ulomega^{\dag}$ with zero
T-weight through the relation 
\begin{align}
\label{EQulomegastar}
\mthorn\ulomega^{\dag}=-i(\bar\Psi_2-\Psi_2),
\end{align}
with trivial initial data on $\mathcal{N}'_{\star}$. Here $\ulomega^{\dag}$ is real.
Then define 
\begin{align}
\tilde\ulomega\equiv\meth(\ulomega+i\ulomega^{\dag})-2\bar\TiPsi_3,
\label{DefinitionTiulomega}
\end{align}
and we can obtain its outgoing equation \eqref{thornTiulomega}:
\begin{align*}
\mthorn\tilde\ulomega=&2\mathrm{i}(\phi_1\meth\bar\zeta_1-3\bar\phi_1\meth\zeta_1
+3\phi_0\meth\bar\zeta_4-\bar\phi_0\meth\zeta_4)+
4\mathrm{i}\zeta_0\bar\phi_1\mu \\
&+\zeta_i\phi_j^3+\zeta_i^2+\Gamma\zeta_i\phi_j+\Gamma\meth\Gamma
+\Gamma\Psi+\Gamma^3+\Gamma\tilde\ulomega
\end{align*}
where we have listed the large terms which can be controlled in the following
\begin{align*}
\int_0^va^5||\meth^{10}(\bar\phi_0\meth\zeta_4)||_{L^2_{sc}(\mathcal{S}_{u,v'})}\leq&
\phi[\phi_0]\int_0^v\frac{a^{\frac{1}{2}}}{|u|}
||a^5\mathcal{D}^{11}\zeta_4||_{L^2_{sc}(\mathcal{S}_{u,v'})} 
+\frac{1}{a^{\frac{1}{2}}}\phi[\phi_0]\zeta[\zeta_4]\\
\lesssim&\bmzeta[\zeta_4]+1,
\end{align*}

\begin{align*}
\int_0^va^5||\meth^{10}(\bar\phi_1\meth\zeta_1)||_{L^2_{sc}(\mathcal{S}_{u,v'})}\leq&
\phi[\phi_1]\int_0^v\frac{a^{\frac{1}{2}}}{|u|}
||a^5\mathcal{D}^{11}\zeta_1||_{L^2_{sc}(\mathcal{S}_{u,v'})}
+\frac{1}{a}\phi[\phi_1]\zeta[\zeta_1] \\
\leq&\frac{a^{\frac{1}{2}}}{|u|}\Gamma[\phi_1]\bmzeta[\zeta_1]\lesssim1,
\end{align*}

\begin{align*}
\int_0^va^5||\meth^{10}(\zeta_0\phi_1\mu)||_{L^2_{sc}(\mathcal{S}_{u,v'})}\leq&
\zeta[\zeta_0]\phi[\phi_1]\lesssim
\bmphi[\phi_0]+\underline{\bmzeta}[\zeta_1]+1.
\end{align*}

Then with the vacuum results one can obtain
\begin{align*}
||(a^{\frac{1}{2}}\meth)^{10}\tilde\ulomega||_{L^2_{sc}(\mathcal{S}_{u,v})}\lesssim&
\bmphi[\phi_0]+\underline{\bmzeta}[\zeta_1]+\bmzeta[\zeta_4]+1.
\end{align*}
Follow the standard procedure one has
\begin{align*}
||a^5\mathcal{D}^{11}\ulomega||_{L^2_{sc}(\mathcal{N}'_v(u_{\infty},u))}\lesssim&
\bmphi[\phi_0]+\underline{\bmzeta}[\zeta_1]+\bmzeta[\zeta_4]
+\underline{\bm\Psi}[\TiPsi_3]+1, \\
||a^5\mathcal{D}^{11}\ulomega||_{L^2_{sc}(\mathcal{N}_u(0,v))}\lesssim&
\bmphi[\phi_0]+\underline{\bmzeta}[\zeta_1]+\bmzeta[\zeta_4]
+\bm\Psi[\TiPsi_3]+1. 
\end{align*}

\end{proof}

\begin{proposition}
\label{11Dermulambda}
We have that 
\begin{align*}
\int_{u_{\infty}}^u\frac{a^2}{|u'|^3}||a^5\mathcal{D}^{11}(\Timu,\mu)||_{L^2_{sc}(\mathcal{S}_{u,v})}
\lesssim\bmphi[\phi_0]+\underline{\bmphi}[\phi_1]
+\underline{\bmzeta}[\zeta_1]+\bmzeta[\zeta_4]
+\underline{\bmzeta}[\zeta_5]+\underline{\bm\Psi}[\TiPsi_3]+1,
\end{align*}
\begin{align*}
\int_{u_{\infty}}^u\frac{a^{\frac{3}{2}}}{|u'|^3}||a^5\mathcal{D}^{11}\lambda||_{L^2_{sc}(\mathcal{S}_{u,v})}
\lesssim1.
\end{align*}
Moreover we have
\begin{align*}
\int_{u_{\infty}}^u\frac{a^3}{|u'|^4}||a^5\mathcal{D}^{11}\Timu||^2_{L^2_{sc}(\mathcal{S}_{u,v})}
\lesssim(\bmphi[\phi_0]+\underline{\bmphi}[\phi_1]
+\underline{\bmzeta}[\zeta_1]+\bmzeta[\zeta_4]
+\underline{\bmzeta}[\zeta_5]+\underline{\bm\Psi}[\TiPsi_3]+1)^2,
\end{align*}
\begin{align*}
\int_{u_{\infty}}^u\frac{a^2}{|u'|^4}||a^5\mathcal{D}^{11}\lambda||^2_{L^2_{sc}(\mathcal{S}_{u,v})}\lesssim&1.
\end{align*}

\end{proposition}

\begin{proof}
We start with $\mu$ by 
\begin{align*}
\mthorn'\tilde\mu+2\mu\tilde\mu=\tilde\mu^2-\ulomega\mu-\lambda\bar\lambda
+2\mathrm{i}(\zeta_5\bar\phi_1-\bar\zeta_5\phi_1).
\end{align*}
One needs to estimate the following:
\begin{align*}
\int_{u_{\infty}}^u\frac{a^2}{|u'|^3}||a^5\bar\phi_1\meth^{11}\zeta_5||_{L^2_{sc}(\mathcal{S}_{u',v})}\lesssim&\phi[\phi_1]
\int_{u_{\infty}}^u\frac{a^{\frac{5}{2}}}{|u'|^4}||a^5\meth^{11}\zeta_5||_{L^2_{sc}(\mathcal{S}_{u',v})}\\
\leq&\phi[\phi_1]\left(\int_{u_{\infty}}^u\frac{a^{3}}{|u'|^4}\right)^{\frac{1}{2}}
\left(\int_{u_{\infty}}^u\frac{a}{|u'|^4}||(a^{\frac{1}{2}}\mathcal{D})^{11}\zeta_5||^2_{L^2_{sc}(\mathcal{S}_{u',v})}\right)^{\frac{1}{2}}\\
\leq&\frac{a^{\frac{3}{2}}}{|u|^{\frac{3}{2}}}\phi[\phi_1]\underline{\bmzeta}[\zeta_5]
\lesssim\underline{\bmzeta}[\zeta_5],
\end{align*}
\begin{align*}
\int_{u_{\infty}}^u\frac{a^2}{|u'|^3}||a^5\bar\zeta_5\meth^{11}\phi_1||_{L^2_{sc}(\mathcal{S}_{u',v})}\lesssim&\frac{1}{a^{\frac{1}{2}}}\zeta[\zeta_5]
\int_{u_{\infty}}^u\frac{a^2}{|u'|^3}||a^5\meth^{11}\phi_1||_{L^2_{sc}(\mathcal{S}_{u',v})}\\
\leq&\frac{1}{a^{\frac{1}{2}}}\zeta[\zeta_5]\left(\int_{u_{\infty}}^u\frac{a^{3}}{|u'|^4}\right)^{\frac{1}{2}}
\left(\int_{u_{\infty}}^u\frac{a}{|u'|^2}||a^{5}\mathcal{D}^{11}\phi_1||^2_{L^2_{sc}(\mathcal{S}_{u',v})}\right)^{\frac{1}{2}}\\
\leq&\frac{a^{{\frac{3}{2}}}}{|u|^{\frac{3}{2}}}\zeta[\zeta_5]\underline{\bmphi}[\phi_1]
\lesssim\underline{\bmphi}[\phi_1].
\end{align*}
Then we have
\begin{align}
\label{11DerivmuMid01}
\frac{a}{|u|}||a^5\mathcal{D}^{11}\Timu||_{L^2_{sc}(\mathcal{S}_{u,v})}\lesssim&
||a^5\mathcal{D}^{11}\ulomega||_{L^2_{sc}(\mathcal{N}'_v(u_{\infty},u))}+\underline{\bmzeta}[\zeta_5]+\underline{\bmphi}[\phi_1]+1.
\end{align}
Applying the estimate for $\ulomega$ and integrating along the ingoing light cone, we have
\begin{align*}
\int_{u_{\infty}}^u\frac{a^2}{|u'|^3}||a^5\mathcal{D}^{11}(\Timu,\mu)||_{L^2_{sc}(\mathcal{S}_{u,v})}
\lesssim\bmphi[\phi_0]+\underline{\bmphi}[\phi_1]
+\underline{\bmzeta}[\zeta_1]+\bmzeta[\zeta_4]
+\underline{\bmzeta}[\zeta_5]+\underline{\bm\Psi}[\TiPsi_3]+1.
\end{align*}
We also obtain
\begin{align*}
\int_{u_{\infty}}^u\frac{a^{\frac{3}{2}}}{|u'|^3}||a^5\mathcal{D}^{11}\lambda||_{L^2_{sc}(\mathcal{S}_{u,v})}
\lesssim1.
\end{align*}
Moreover we have
\begin{align*}
\int_{u_{\infty}}^u\frac{a^3}{|u'|^4}||a^5\mathcal{D}^{11}\Timu||^2_{L^2_{sc}(\mathcal{S}_{u,v})}
\lesssim(\bmphi[\phi_0]+\underline{\bmphi}[\phi_1]
+\underline{\bmzeta}[\zeta_1]+\bmzeta[\zeta_4]
+\underline{\bmzeta}[\zeta_5]+\underline{\bm\Psi}[\TiPsi_3]+1)^2,
\end{align*}
and
\begin{align*}
\int_{u_{\infty}}^u\frac{a^2}{|u'|^4}||a^5\mathcal{D}^{11}\lambda||^2_{L^2_{sc}(\mathcal{S}_{u,v})}\lesssim&1.
\end{align*}
\end{proof}

\subsection{Summary of elliptic}

Make use of the bootstrap assumption
\begin{align*}
\bm\Gamma,\bmphi,\bmzeta,\bm\Psi\leq\mathcal{O}, \quad
\bm\Gamma_{11}\leq\mathcal{O},
\end{align*}
we obtain
\begin{align*}
||a^5\mathcal{D}^{11}\rho||_{L^2_{sc}(\mathcal{S}_{u,v})}\lesssim&
\bm{\phi}[\phi_0]+\bm{\phi}[\zeta_0]
+\bm\Psi[\Psi_0]+\bm\Psi[\TiPsi_1]+1,\\
||a^5\mathcal{D}^{11}\rho||_{L^2_{sc}(\mathcal{N}_u(0,v))}\lesssim&
\bm{\phi}[\phi_0]+\bm{\phi}[\zeta_0]
+\bm\Psi[\Psi_0]+\bm\Psi[\TiPsi_1]+1,\\
||a^5\mathcal{D}^{11}\sigma||_{L^2_{sc}(\mathcal{N}_u(0,v))}
\lesssim&\bm{\phi}[\phi_0]+\bm{\phi}[\zeta_0]
+\bm\Psi[\Psi_0]+\bm\Psi[\TiPsi_1]+1, \\
||a^5\mathcal{D}^{11}\tau||_{L^2_{sc}(\mathcal{N}_u(0,v))}
\lesssim&\bm\Psi[\Psi_0]+\bm\Psi[\Psi_2]
+\underline{\bmzeta}[\zeta_1]+\underline{\bmphi}[\phi_1]+1,\\
||a^5\mathcal{D}^{11}\tau||_{L^2_{sc}(\mathcal{N}'_v(u_{\infty},u))}
\lesssim&\bm\Psi[\Psi_0]+\underline{\bm\Psi}[\Psi_2]
+\underline{\bmzeta}[\zeta_1]+\underline{\bmphi}[\phi_1]+1, \\
\frac{a}{|u|}||a^5\mathcal{D}^{11}\pi||_{L^2_{sc}(\mathcal{N}_{u}(0,v))}\lesssim&
\bm\Psi[\Psi_0]+\underline{\bm\Psi}[\Psi_2]+\bm\Psi[\Psi_2]
+\underline{\bmzeta}[\zeta_1]+\underline{\bmphi}[\phi_1]+1, \\
||a^5\mathcal{D}^{11}\ulomega||_{L^2_{sc}(\mathcal{N}'_v(u_{\infty},u))}\lesssim&
\bmphi[\phi_0]+\underline{\bmzeta}[\zeta_1]+\bmzeta[\zeta_4]
+\underline{\bm\Psi}[\TiPsi_3]+1, \\
||a^5\mathcal{D}^{11}\ulomega||_{L^2_{sc}(\mathcal{N}_u(0,v))}\lesssim&
\bmphi[\phi_0]+\underline{\bmzeta}[\zeta_1]+\bmzeta[\zeta_4]
+\bm\Psi[\TiPsi_3]+1, \\
\int_{u_{\infty}}^u\frac{a^2}{|u'|^3}||a^5\mathcal{D}^{11}(\Timu,\mu)||_{L^2_{sc}(\mathcal{S}_{u,v})}
\lesssim&\bmphi[\phi_0]+\underline{\bmphi}[\phi_1]
+\underline{\bmzeta}[\zeta_1]+\bmzeta[\zeta_4]
+\underline{\bmzeta}[\zeta_5]+\underline{\bm\Psi}[\TiPsi_3]+1,\\
\int_{u_{\infty}}^u\frac{a^3}{|u'|^4}||a^5\mathcal{D}^{11}\Timu||^2_{L^2_{sc}(\mathcal{S}_{u,v})}
\lesssim&(\bmphi[\phi_0]+\underline{\bmphi}[\phi_1]
+\underline{\bmzeta}[\zeta_1]+\bmzeta[\zeta_4]
+\underline{\bmzeta}[\zeta_5]+\underline{\bm\Psi}[\TiPsi_3]+1)^2, \\
\int_{u_{\infty}}^u\frac{a^{\frac{3}{2}}}{|u'|^3}||a^5\mathcal{D}^{11}\lambda||_{L^2_{sc}(\mathcal{S}_{u,v})}
\lesssim&1,\quad
\int_{u_{\infty}}^u\frac{a^2}{|u'|^4}||a^5\mathcal{D}^{11}\lambda||^2_{L^2_{sc}(\mathcal{S}_{u,v})}\lesssim1
\end{align*}

Hence we have the following improvements:
\[
  \bmGamma_{11}\lesssim\bm\Psi^2+\bm\Psi+\bm\phi^2+\bm\phi+\bmzeta^2+\bmzeta+1.
  \]

\section{Energy Estimate}
\label{EnergyEstimate}
In this section we show the energy estimate 
\begin{align*}
\bm\phi+\bmzeta+\Psi\lesssim\mathcal{I}_0+1
\end{align*}
to close the bootstrap argument. 
To achieve this goal, we make use of the energy inequality in scale-invariant norm. 
We introduce in the following: let
\[
(\Psi_I,\Psi_{II})\in
\{(\Psi_0,\tilde\Psi_1),(\tilde\Psi_1,\Psi_2),(\Psi_2,\tilde\Psi_3),(\tilde\Psi_3,\Psi_4),(\phi_{0},\phi_{1}),(\zeta_{0},\zeta_{1}),(\zeta_{1},\zeta_{2}),(\Tizeta_{4},\Tizeta_{5}),(\zeta_4,\zeta_5)
\}.
\]
the Bianchi identities and the equations for the
scalar field take the schematic form
\begin{align*}
&\mthorn'\Psi_I+\lambda_0\mu\Psi_I-\meth\Psi_{II}=P_0, \\
&\mthorn\Psi_{II}-\meth'\Psi_I=Q_0.
\end{align*}
Applying the derivative $\meth$ and the commutator relations it
follows that
\begin{subequations}
\begin{align}
&\mthorn'\meth^k\Psi_{I}+(\lambda_0+k)\mu\meth^k\Psi_{I}-\meth^{k+1}\Psi_{II}=P_k, \label{kthBianchi1} \\
&\mthorn\meth^k\Psi_{II}-\meth'\meth^k\Psi_{I}=Q_k, \label{kthBianchi2} 
\end{align}
\end{subequations}
where
\begin{align*}
P_k=\meth^{k}P_0+
\sum_{i=0}^k\meth^i\Gamma(\Timu,\lambda)\meth^{k-i}\Psi_{I}
\end{align*}
and 
\begin{align*}
Q_k=&\sum_{i_1+i_2+i_3=k}\meth^{i_1}\Gamma(\pi,\tau)^{i_2}\meth^{i_3}\Psi_{I}
+\sum_{i_1+i_2+i_3=k}\meth^{i_1}\Gamma(\pi,\tau)^{i_2}\meth^{i_3}Q_0 \\
&+\sum_{i_1+i_2+i_3+i_4=k}\meth^{i_1}\Gamma(\tau,\pi)^{i_2}\meth^{i_3}\Gamma(\tau,\pi,\rho,\sigma)\meth^{i_4}\Psi_{II}
+\sum_{i_1+i_2=k-1}\meth^{i_1}K\meth^{i_2}\Psi_{I}.
\end{align*}
\smallskip

we have that
\begin{align*}
&\int_0^v|u|^{2\lambda_0-4s_2(\Psi_{I})-2}||\meth^k\Psi_{I}||^2_{L^2_{sc}(\mathcal{S}_{u,v'})}
+\int_{u_{\infty}}^ua|u'|^{2\lambda_0-4s_2(\Psi_{I})-4}||\meth^k\Psi_{II}||^2_{L^2_{sc}(\mathcal{S}_{u',v})} \\
\lesssim&\int_0^v|u_{\infty}|^{2\lambda_0-4s_2(\Psi_{I})-2}||\meth^k\Psi_{I}||^2_{L^2_{sc}(\mathcal{S}_{u_{\infty},v'})}
+\int_{u_{\infty}}^ua|u'|^{2\lambda_0-4s_2(\Psi_{I})-4}||\meth^k\Psi_{II}||^2_{L^2_{sc}(\mathcal{S}_{u',0})} \\
&+2\int_0^v\int_{u_{\infty}}^ua|u'|^{2\lambda_0-4s_2(\Psi_{I})-4}
||\meth^k\Psi_{I}||_{L^2_{sc}(\mathcal{S}_{u',v'})}
||P_k||_{L^2_{sc}(\mathcal{S}_{u',v'})} \\
&+2\int_0^v\int_{u_{\infty}}^ua|u'|^{2\lambda_0-4s_2(\Psi_{I})-4}
||\meth^k\Psi_{II}||_{L^2_{sc}(\mathcal{S}_{u',v'})}
||Q_k||_{L^2_{sc}(\mathcal{S}_{u',v'})}.
\end{align*}
The proof can be found in~\cite{An2022,HilValZha23}.
The signature of the various terms involved satisfy 
\begin{align*}
s_2(\meth^k\Psi_I)=\frac{k}{2}+s_2(\Psi_I), \quad 
&s_2(\meth^k\Psi_{II})=\frac{k+1}{2}+s_2(\Psi_I),\\
s_2(P_k)=\frac{k+2}{2}+s_2(\Psi_I), \quad 
&s_2(Q_k)=\frac{k+1}{2}+s_2(\Psi_I).
\end{align*}. 
Consequently, for pair except $(\zeta_4,\zeta_5)$, the  
the above inequality takes the form
\begin{align*}
&\int_0^v||\meth^k\Psi_{I}||^2_{L^2_{sc}(\mathcal{S}_{u,v'})}
+\int_{u_{\infty}}^u\frac{a}{|u'|^2}||\meth^k\Psi_{II}||^2_{L^2_{sc}(\mathcal{S}_{u',v})} \\
\lesssim&\int_0^v||\meth^k\Psi_{I}||^2_{L^2_{sc}(\mathcal{S}_{u_{\infty},v'})}
+\int_{u_{\infty}}^u\frac{a}{|u'|^2}||\meth^k\Psi_{II}||^2_{L^2_{sc}(\mathcal{S}_{u',0})} \\
&+2\int_0^v\int_{u_{\infty}}^u\frac{a}{|u'|^2}
||\meth^k\Psi_{I}||_{L^2_{sc}(\mathcal{S}_{u',v'})}
||P_k||_{L^2_{sc}(\mathcal{S}_{u',v'})} \\
&+2\int_0^v\int_{u_{\infty}}^u\frac{a}{|u'|^2}
||\meth^k\Psi_{II}||_{L^2_{sc}(\mathcal{S}_{u',v'})}
||Q_k||_{L^2_{sc}(\mathcal{S}_{u',v'})}.
\end{align*}
 the definitions of the norms on the lightcone ---namely,
\begin{align*}
||\phi||^2_{L^2_{sc}(\mathcal{N}_u(0,v))}&
\equiv\int_0^v||\phi||^2_{L^2_{sc}(\mathcal{S}_{u,v'})}\mathrm{d}v',\\
||\phi||^2_{L^2_{sc}(\mathcal{N}'_v(u_{\infty},u))}&
\equiv\int_{u_{\infty}}^u\frac{a}{|u'|^2}||\phi||^2_{L^2_{sc}(\mathcal{S}_{u',v})}\mathrm{d}u',
\end{align*}
we conclude that 
\begin{align}
&||\meth^k\Psi_{I}||^2_{L^2_{sc}(\mathcal{N}_u(0,v))}
+||\meth^k\Psi_{II}||^2_{L^2_{sc}(\mathcal{N}'_v(u_{\infty},u))} \nonumber\\
\lesssim&||\meth^k\Psi_{I}||^2_{L^2_{sc}(\mathcal{N}_{u_{\infty}}(0,v))}
+||\meth^k\Psi_{II}||^2_{L^2_{sc}(\mathcal{N}'_0(u_{\infty},u))} \nonumber\\
&+2\int_0^v\int_{u_{\infty}}^u\frac{a}{|u'|^2}
||\meth^k\Psi_{I}||_{L^2_{sc}(\mathcal{S}_{u',v'})}
||P_k||_{L^2_{sc}(\mathcal{S}_{u',v'})} \nonumber\\
&+2\int_0^v\int_{u_{\infty}}^u\frac{a}{|u'|^2}
||\meth^k\Psi_{II}||_{L^2_{sc}(\mathcal{S}_{u',v'})}
||Q_k||_{L^2_{sc}(\mathcal{S}_{u',v'})}. \label{EnergyinequalityNormal}
\end{align}

For pair $(\zeta_4,\zeta_5)$, $2\lambda_0-4s_2(\Psi_{I})-4=2\times2-4\times1-4=-4$ hence we have
\begin{align}
&\int_0^v\frac{1}{|u|^2}||\meth^k\zeta_4||^2_{L^2_{sc}(\mathcal{S}_{u,v'})}
+\int_{u_{\infty}}^u\frac{a}{|u'|^4}||\meth^k\zeta_5||^2_{L^2_{sc}(\mathcal{S}_{u',v})} \nonumber\\
\lesssim&\int_0^v\frac{1}{|u_{\infty}|^2}||\meth^k\zeta_4||^2_{L^2_{sc}(\mathcal{S}_{u_{\infty},v'})}
+\int_{u_{\infty}}^u\frac{a}{|u'|^4}||\meth^k\zeta_5||^2_{L^2_{sc}(\mathcal{S}_{u',0})} \nonumber\\
&+2\int_0^v\int_{u_{\infty}}^u\frac{a}{|u'|^4}
||\meth^k\zeta_4||_{L^2_{sc}(\mathcal{S}_{u',v'})}
||P_k||_{L^2_{sc}(\mathcal{S}_{u',v'})} \nonumber\\
&+2\int_0^v\int_{u_{\infty}}^u\frac{a}{|u'|^4}
||\meth^k\zeta_5||_{L^2_{sc}(\mathcal{S}_{u',v'})}
||Q_k||_{L^2_{sc}(\mathcal{S}_{u',v'})}. \label{EnergyinequalityAbNormal}
\end{align}

\subsection{Energy estimate for $\phi_A$ }
We begin with pair $(\phi_0,\phi_1)$:

\begin{proposition}
\label{EnergyEstimatephiA}
For $0\leq k\leq11$,  one has that 
\begin{align*}
\frac{1}{a}||(a^{\frac{1}{2}})^{k-1}\mathcal{D}^k\phi_0||^2_{L^2_{sc}(\mathcal{N}_u(0,v))}
+\frac{1}{a}||(a^{\frac{1}{2}})^{k-1}\mathcal{D}^k\phi_1||^2_{L^2_{sc}(\mathcal{N}'_v(u_{\infty},u))} 
\leq\mathcal{I}^2_0+\frac{1}{a^{\frac{1}{4}}}.
\end{align*}
\end{proposition}

\begin{proof}
For $(\phi_0,\phi_1)$ we make use of the equations
\begin{align*}
\mthorn'\phi_{0}-\meth\phi_{1}&=(\frac{\ulomega}{2}-\mu)\phi_{0}
-\frac{\bar\tau\phi_{1}}{2}, \\
\mthorn\phi_1-\meth'\phi_0&=\phi_0(\pi-\frac{\bar\pi}{2})+\phi_1\rho,
\end{align*}
and the energy inequality \eqref{EnergyinequalityNormal}, then we obtain 
\begin{align}
\label{EstimatephiA}
&\frac{1}{a}||(a^{\frac{1}{2}})^{k-1}\meth^k\phi_0||^2_{L^2_{sc}(\mathcal{N}_u(0,v))}
+\frac{1}{a}||(a^{\frac{1}{2}})^{k-1}\meth^k\phi_1||^2_{L^2_{sc}(\mathcal{N}'_v(u_{\infty},u))} \nonumber \\
\lesssim&\frac{1}{a}||(a^{\frac{1}{2}})^{k-1}\mathcal{D}^k\phi_0||^2_{L^2_{sc}(\mathcal{N}_{u_{\infty}}(0,v))}
+\frac{1}{a}||(a^{\frac{1}{2}})^{k-1}\mathcal{D}^k\phi_1||^2_{L^2_{sc}(\mathcal{N}'_0(u_{\infty},u))} 
+\frac{1}{a}(M+N),
\end{align}
where
\begin{align*}
M\equiv&2\int_0^v\int_{u_{\infty}}^u\frac{a}{|u'|^2}
||(a^{\frac{1}{2}})^{k-1}\meth^k\phi_0||_{L^2_{sc}(\mathcal{S}_{u',v'})}
||(a^{\frac{1}{2}})^{k-1}P_k||_{L^2_{sc}(\mathcal{S}_{u',v'})}, \\
N\equiv&2\int_0^v\int_{u_{\infty}}^u\frac{a}{|u'|^2}
||(a^{\frac{1}{2}})^{k-1}\meth^k\phi_1||_{L^2_{sc}(\mathcal{S}_{u',v'})}
||(a^{\frac{1}{2}})^{k-1}Q_k||_{L^2_{sc}(\mathcal{S}_{u',v'})}. \\
\end{align*}
and
\begin{align*}
P_k=&\Gamma(\ulomega,\tau,\lambda)\meth^k\phi_{0,1}
+\phi_{0,1}\meth^{k}\Gamma(\ulomega,\tau,\Timu)
+\sum_{i=1}^{k-1}\meth^i\Gamma(\Timu,\ulomega,\tau,\lambda)\meth^{k-i}(\phi_{0},\phi_1), \\
Q_k=&\Gamma(\rho,\sigma,\pi,\tau)\meth^k\phi_{0,1}
+\phi_{0,1}\meth^{k}\Gamma(\pi,\rho)\\
&+\sum_{i_1+i_2+i_3+i_4=k, i_3+i_4<k}\meth^{i_1}\Gamma(\tau,\pi)^{i_2}\meth^{i_3}
\Gamma(\tau,\pi,\rho,\sigma)\meth^{i_4}(\phi_{0},\phi_{1})
+\sum_{i_1+i_2=k-1}\meth^{i_1}K\meth^{i_2}\phi_0.
\end{align*}

For $M$ we have 
\begin{align*}
M\leq&2\left(\int_0^v\int_{u_{\infty}}^u\frac{a}{|u'|^2}
||(a^{\frac{1}{2}})^{k-1}\meth^k\phi_0||^2_{L^2_{sc}(\mathcal{S}_{u',v'})}\right)^{\frac{1}{2}}
\left(\int_0^v\int_{u_{\infty}}^u\frac{a}{|u'|^2}||(a^{\frac{1}{2}})^{k-1}P_k||^2_{L^2_{sc}(\mathcal{S}_{u',v'})}\right)^{\frac{1}{2}} \\
=&2\left(\int_{u_{\infty}}^u\frac{a}{|u'|^2}
||(a^{\frac{1}{2}})^{k-1}\meth^k\phi_0||^2_{L^2_{sc}(\mathcal{N}_{u'}(0,v))}\right)^{\frac{1}{2}}J^{\frac{1}{2}}
\leq\frac{a^{\frac{1}{2}}}{|u|^{\frac{1}{2}}}a^{\frac{1}{2}}\bm{\phi}[\phi_0]J^{\frac{1}{2}},
\end{align*}
where
\begin{align*}
J\equiv\int_0^v\int_{u_{\infty}}^u\frac{a}{|u'|^2}||(a^{\frac{1}{2}})^{k-1}P_k||^2_{L^2_{sc}(\mathcal{S}_{u',v'})}.
\end{align*}
Then substitute $P_k$, we make use of the estimate in the elliptic part to estimate 
the top derivative of connections  
\begin{align*}
J\leq&\int_0^v\int_{u_{\infty}}^u\frac{a}{|u'|^2}
||\phi_{0,1}(a^{\frac{1}{2}})^{k-1}
\mathcal{D}^k\Gamma(\ulomega,\Timu,\tau,\lambda)||^2_{L^2_{sc}(\mathcal{S}_{u',v'})} \\
&+\int_0^v\int_{u_{\infty}}^u\frac{a}{|u'|^2}
||\Gamma(\ulomega,\lambda,\tau)(a^{\frac{1}{2}})^{k-1}
\mathcal{D}^k\phi_{0,1}||^2_{L^2_{sc}(\mathcal{S}_{u',v'})} \\
&+\sum_{i=1}^{k-1}\int_0^v\int_{u_{\infty}}^u\frac{a}{|u'|^2}
||(a^{\frac{i}{2}}\mathcal{D}^i)\Gamma(\ulomega,\Timu,\tau,\lambda)(a^{\frac{k-1-i}{2}}
\mathcal{D}^{k-i})\phi_{0,1}||^2_{L^2_{sc}(\mathcal{S}_{u',v'})} \\
=&J_1+J_2+J_3.
\end{align*}

Then we have
\begin{align*}
J_1\lesssim&\int_0^v\int_{u_{\infty}}^u\frac{a}{|u'|^2}\frac{a}{|u'|^2}\phi[\phi_{0,1}]^2
||(a^{\frac{1}{2}})^{k-1}
\mathcal{D}^k\Gamma(\ulomega,\Timu,\tau)||^2_{L^2_{sc}(\mathcal{S}_{u',v'})}\\
\lesssim&\frac{\phi[\phi_{0,1}]^2}{a}\int_0^v||(a^{\frac{1}{2}})^{k-1}
\mathcal{D}^k\Gamma(\ulomega,\tau)||^2_{L^2_{sc}(\mathcal{S}_{u',v'})}\\
&+\phi[\phi_{0,1}]^2\int_0^v\int_{u_{\infty}}^u\frac{a^2}{|u'|^4}
||(a^{\frac{1}{2}})^{k-1}\mathcal{D}^k\Timu)||^2_{L^2_{sc}(\mathcal{S}_{u',v'})} \\
&+\phi[\phi_{0,1}]^2\int_0^v\int_{u_{\infty}}^u\frac{a^2}{|u'|^4}
||(a^{\frac{1}{2}})^{k-1}\mathcal{D}^k\lambda)||^2_{L^2_{sc}(\mathcal{S}_{u',v'})} \\
\lesssim&\frac{1}{a}(\bm\phi+\underline{\bm\phi}+\bm\zeta+\underline{\bm\zeta}
+\underline{\bm\Psi}[\TiPsi_3]+1)^2+1,
\end{align*}

\begin{align*}
J_2\lesssim&\int_0^v\int_{u_{\infty}}^u\frac{a}{|u'|^2}\frac{1}{|u'|^2}\left(\frac{|u'|^2}{a}\Gamma[\lambda]^2
||(a^{\frac{1}{2}})^{k-1}\mathcal{D}^k\phi_{0}||^2_{L^2_{sc}(\mathcal{S}_{u',v'})}
+\Gamma[\tau]^2
||(a^{\frac{1}{2}})^{k-1}\mathcal{D}^k\phi_{1}||^2_{L^2_{sc}(\mathcal{S}_{u',v'})}\right) \\
\lesssim&\frac{a}{|u|}\bm{\phi}[\phi_0]^2
+\frac{a}{|u|^2}\Gamma[\tau]^2\underline{\bm\phi}[\phi_1]^2,
\end{align*}

\begin{align*}
J_3\lesssim&\frac{1}{a}\int_0^v\int_{u_{\infty}}^u\frac{a}{|u'|^2}\frac{1}{|u'|^2}
\frac{|u'|^2}{a}\Gamma[\lambda]a\phi[\phi_{0,1}]
\lesssim\frac{1}{|u|}\mathcal{O}^4.
\end{align*}

Then collect the results one obtain
\begin{align*}
J\lesssim&\frac{a}{|u|}\bm{\phi}[\phi_0]^2+1+\frac{\mathcal{O}^4}{a},
\end{align*}
and then 
\begin{align*}
M\lesssim\frac{a^{\frac{1}{2}}}{|u|^{\frac{1}{2}}}a^{\frac{1}{2}}\bm{\phi}[\phi_0]J^{\frac{1}{2}}
\lesssim\frac{a^{\frac{3}{2}}}{|u|}\bm{\phi}[\phi_0]^2
+\frac{a}{|u|^{\frac{1}{2}}}\bm{\phi}[\phi_0]
+\mathcal{O}^2\bm{\phi}[\phi_0].
\end{align*}

For N we have 
\begin{align*}
N\leq&2\left(\int_0^v\int_{u_{\infty}}^u\frac{a}{|u'|^2}
||(a^{\frac{1}{2}})^{k-1}\meth^k\phi_1||^2_{L^2_{sc}(\mathcal{S}_{u',v'})}\right)^{\frac{1}{2}}
\left(\int_0^v\int_{u_{\infty}}^u\frac{a}{|u'|^2}||(a^{\frac{1}{2}})^{k-1}Q_k||^2_{L^2_{sc}(\mathcal{S}_{u',v'})}\right)^{\frac{1}{2}} \\
=&2\left(\int_0^v
||(a^{\frac{1}{2}})^{k-1}\meth^k\phi_1||^2_{L^2_{sc}(\mathcal{N}'_{v'}(u_{\infty},u))}\right)^{\frac{1}{2}}H^{\frac{1}{2}}
\lesssim a^{\frac{1}{2}}\underline{\bm{\phi}}[\phi_1]H^{\frac{1}{2}},
\end{align*}
where 
\begin{align*}
H\equiv\int_0^v\int_{u_{\infty}}^u\frac{a}{|u'|^2}
||(a^{\frac{1}{2}})^{k-1}Q_k||^2_{L^2_{sc}(\mathcal{S}_{u',v'})}.
\end{align*}
Substituting the definition of $Q_k$ into $H$ we obtain
\begin{align*}
H\leq&\int_0^v\int_{u_{\infty}}^u\frac{a}{|u'|^2}
||\phi_{0,1}(a^{\frac{1}{2}})^{k-1}
\mathcal{D}^k\Gamma(\rho,\sigma,\pi)||^2_{L^2_{sc}(\mathcal{S}_{u',v'})} \\
&+\int_0^v\int_{u_{\infty}}^u\frac{a}{|u'|^2}
||\Gamma(\rho,\sigma,\pi,\tau)(a^{\frac{1}{2}})^{k-1}
\mathcal{D}^k\phi_{0,1}||^2_{L^2_{sc}(\mathcal{S}_{u',v'})} \\
&+\sum_{i_1+i_2+i_3+i_4=k, i_3+i_4<k}\int_0^v\int_{u_{\infty}}^u\frac{a}{|u'|^2}
||a^{\frac{k-1}{2}}\meth^{i_1}\Gamma^{i_2}\meth^{i_3}
\Gamma(\sigma,...)\meth^{i_4}\phi_{0,1}||^2_{L^2_{sc}(\mathcal{S}_{u',v'})} \\
&+\sum_{i_1+i_2=k-1}\int_0^v\int_{u_{\infty}}^u\frac{a}{|u'|^2}
||(a^{\frac{1}{2}})^{k-1}\mathcal{D}^{i_1}K\mathcal{D}^{i_2}\phi_0||^2_{L^2_{sc}(\mathcal{S}_{u',v'})} 
=H_1+...+H_4.
\end{align*}
We have

\begin{align*}
H_1\lesssim&\int_0^v\int_{u_{\infty}}^u\frac{a}{|u'|^2}\frac{a}{|u'|^2}\phi[\phi_{0,1}]^2
||(a^{\frac{1}{2}})^{k-1}
\mathcal{D}^k\Gamma(\rho,\sigma,\pi)||^2_{L^2_{sc}(\mathcal{S}_{u',v'})} \\
\lesssim&\frac{\phi[\phi_{0,1}]^2}{a}\bm\Gamma_{11}[\rho,\pi]^2
\lesssim\frac{1}{a}(\bm\phi+\underline{\bm\phi}+\bm\zeta+\underline{\bm\zeta}
+\bm\Psi+\underline{\bm\Psi}+1)^2,
\end{align*}

\begin{align*}
H_2\lesssim&\int_0^v\int_{u_{\infty}}^u\frac{a}{|u'|^2}
(\frac{a}{|u'|^2}\Gamma[\sigma]^2||(a^{\frac{1}{2}})^{k-1}
\mathcal{D}^k\phi_{1}||^2_{L^2_{sc}(\mathcal{S}_{u',v'})} \\
&+\frac{1}{|u'|^2}\Gamma[\pi]^2||(a^{\frac{1}{2}})^{k-1}
\mathcal{D}^k\phi_{0}||^2_{L^2_{sc}(\mathcal{S}_{u',v'})}) \\
\lesssim&\frac{a^2}{|u|^2}(\bm\Psi+1)^2\underline{\bm\phi}[\phi_1]^2
+\frac{a^2}{|u|^3}(\bm\Psi+\underline{\bm\Psi}+1)^2\bm\phi[\phi_0]^2,
\end{align*}

\begin{align*}
H_3\lesssim&\frac{1}{a}\int_0^v\int_{u_{\infty}}^u\frac{a}{|u'|^2}\frac{1}{|u'|^2} a\Gamma[\sigma]^2a\phi[\phi_{0,1}]^2
\lesssim\frac{1}{a}\mathcal{O}^4,
\end{align*}

\begin{align*}
H_4\lesssim&
\sum_{i=0}^{k-1}\int_0^v\int_{u_{\infty}}^u\frac{a}{|u'|^2}\frac{1}{|u'|^2}||(a^{\frac{1}{2}}\mathcal{D})^{i}K||^2
||(a^{\frac{1}{2}}\mathcal{D})^{k-1-i}\phi_0||^2 \\
\leq&\int_{u_{\infty}}^u\frac{a^2}{|u'|^4}\phi[\phi_{0}]^2 
\int_0^v||(a^{\frac{1}{2}}\mathcal{D})^{k-1}K||^2_{L^2_{sc}(\mathcal{S}_{u',v'})}\\
\lesssim&\frac{a^2}{|u|^3}\phi[\phi_{0}]^2
||(a^{\frac{1}{2}}\mathcal{D})^{k-1}K||^2_{L^2_{sc}(\mathcal{N}_{u}(0,v))}.
\end{align*}

Then collect the results we have
\begin{align*}
H\leq&\frac{a^2}{|u|^2}(\bm\Psi+1)^2\underline{\bm\phi}[\phi_1]^2+\frac{\mathcal{O}^4}{a},
\end{align*}
and hence we have
\begin{align*}
N\lesssim a^{\frac{1}{2}}\underline{\bm{\phi}}[\phi_1]H^{\frac{1}{2}}
\lesssim\frac{a^{\frac{3}{2}}}{|u|}(\bm\Psi+1)\bm{\phi}\underline{\bm\phi}+\mathcal{O}^2\bm{\phi}.
\end{align*}

Then we have

\begin{align*}
&\frac{1}{a}||(a^{\frac{1}{2}})^{k-1}\meth^k\phi_0||^2_{L^2_{sc}(\mathcal{N}_u(0,v))}
+\frac{1}{a}||(a^{\frac{1}{2}})^{k-1}\meth^k\phi_1||^2_{L^2_{sc}(\mathcal{N}'_v(u_{\infty},u))} \\
\lesssim&\mathcal{I}_0^2+\frac{1}{a}(\frac{a^{\frac{3}{2}}}{|u|}\bm{\phi}[\phi_0]^2
+\frac{a}{|u|^{\frac{1}{2}}}\bm{\phi}[\phi_0]
+\mathcal{O}^2\bm{\phi}[\phi_0]+\frac{a^{\frac{3}{2}}}{|u|}(\bm\Psi+1)\bm{\phi}\underline{\bm\phi}
+\mathcal{O}^2\bm{\phi}) \\
\lesssim&\mathcal{I}_0^2+\frac{1}{a^{\frac{1}{4}}}.
\end{align*}

\end{proof}

\subsection{Energy estimate for $\zeta_{ABA'}$ }

\begin{proposition}
\label{EnergyEstimatezeta01}
For $0\leq k\leq11$,  one has that 
\begin{align*}
\frac{1}{a}||(a^{\frac{1}{2}})^{k-1}\mathcal{D}^k\zeta_0||^2_{L^2_{sc}(\mathcal{N}_u(0,v))}
+\frac{1}{a}||(a^{\frac{1}{2}})^{k-1}\mathcal{D}^k\zeta_1||^2_{L^2_{sc}(\mathcal{N}'_v(u_{\infty},u))} 
\leq\mathcal{I}^2_0+\frac{1}{a^{\frac{1}{4}}}.
\end{align*}
\end{proposition}

\begin{proof}
For pair $(\zeta_0,\zeta_1)$, we make use of 
\begin{align*}
\mthorn' \zeta_0- \meth\,\zeta_1 &= 
 \mathrm{i}\,\bar{\zeta}_4\,\phi_0^2
- \mathrm{i}\,\zeta_4\,\phi_0\,\bar{\phi}_0
- \mathrm{i}\,\bar{\zeta}_1\,\phi_0\,\phi_1
+ \mathrm{i}\,\zeta_3\,\bar{\phi}_0\,\phi_1
- \mathrm{i}\,\zeta_1\,\phi_0\,\bar{\phi}_1 \\
&\quad
+ \mathrm{i}\,\zeta_0\,\phi_1\,\bar{\phi}_1
- \zeta_0\,\mu+\frac{3\,\zeta_0\,\ulomega}{2}
+ \zeta_4\,\rho
+ \zeta_2\,\sigma
- \frac{5\,\zeta_1\,\tau}{2}
- \zeta_3\,\bar{\tau}, \\
\mthorn \zeta_1- \meth'\,\zeta_0 &= -\mathrm{i}\,\bar{\zeta}_3\,\phi_0^2
+ 2\,\mathrm{i}\,\zeta_1\,\phi_0\,\bar{\phi}_0
+ \mathrm{i}\,\bar\zeta_0\,\phi_0\,\phi_1
- 2\,\mathrm{i}\,\zeta_0\,\bar{\phi}_0\,\phi_1 
+ \zeta_0\,\pi
+ 2\,\zeta_1\,\rho \\
&\quad
+ \zeta_3\,\bar\sigma
- \frac{3\,\zeta_0\,\bar{\tau}}{2}.
\end{align*}

Applying $\meth^k$ and commuting with $\mthorn$ and $\mthorn'$ we have that
\begin{align*}
&\mthorn'\meth^k\zeta_0+(k+1)\mu\meth^k\zeta_0-\meth^{k+1}\zeta_1=P_k, \\
&\mthorn\meth^k\zeta_1-\meth'\meth^k\zeta_0=Q_k,
\end{align*}
where
\begin{align*}
P_k=&\sum_{i_1+i_2+i_3=k}\meth^{i_1}\phi_j\meth^{i_1}\phi_l\meth^{i_1}\zeta_m
+\sum_{i_1+i_2=k}\meth^{i_1}\Gamma\meth^{i_1}\zeta_m,
\end{align*}
\begin{align*}
Q_k=&\sum_{i_1+...+i_5=k}\meth^{i_1}\Gamma^{i_2}\meth^{i_3}\phi_j\meth^{i_4}\phi_l\meth^{i_5}\zeta_m
+\sum_{i_1+...+i_4=k}\meth^{i_1}\Gamma^{i_2}\meth^{i_3}\Gamma\meth^{i_4}\zeta_m \\
&+\sum_{i_1+i_2=k-1}\meth^{i_1}K\meth^{i_2}\zeta_{0}.
\end{align*}
Then apply the energy inequality \eqref{EnergyinequalityNormal} one has
\begin{align*}
&\frac{1}{a}||(a^{\frac{1}{2}})^{k-1}\meth^k\zeta_0||^2_{L^2_{sc}(\mathcal{N}_u(0,v))}
+\frac{1}{a}||(a^{\frac{1}{2}})^{k-1}\meth^k\zeta_1||^2_{L^2_{sc}(\mathcal{N}'_v(u_{\infty},u))} 
\lesssim\mathcal{I}^2_0+\frac{1}{a}(M+N).
\end{align*}
where
\begin{align*}
M\equiv&2\int_0^v\int_{u_{\infty}}^u\frac{a}{|u'|^2}
||(a^{\frac{1}{2}})^{k-1}\meth^k\zeta_0||_{L^2_{sc}(\mathcal{S}_{u',v'})}
||(a^{\frac{1}{2}})^{k-1}P_k||_{L^2_{sc}(\mathcal{S}_{u',v'})}, \\
N\equiv&2\int_0^v\int_{u_{\infty}}^u\frac{a}{|u'|^2}
||(a^{\frac{1}{2}})^{k-1}\meth^k\zeta_1||_{L^2_{sc}(\mathcal{S}_{u',v'})}
||(a^{\frac{1}{2}})^{k-1}Q_k||_{L^2_{sc}(\mathcal{S}_{u',v'})}. \\
\end{align*}

For $M$ we have 
\begin{align*}
M\leq2\left(\int_{u_{\infty}}^u\frac{a}{|u'|^2}
||(a^{\frac{1}{2}})^{k-1}\meth^k\zeta_0||^2_{L^2_{sc}(\mathcal{N}_{u'}(0,v))}\right)^{\frac{1}{2}}J^{\frac{1}{2}}
\leq\frac{a^{\frac{1}{2}}}{|u|^{\frac{1}{2}}}a^{\frac{1}{2}}\bm{\zeta}[\zeta_0]J^{\frac{1}{2}},
\end{align*}
where 
\begin{align*}
J\equiv\int_0^v\int_{u_{\infty}}^u\frac{a}{|u'|^2}||(a^{\frac{1}{2}})^{k-1}P_k||^2_{L^2_{sc}(\mathcal{S}_{u',v'})}.
\end{align*}

Substitute $P_k$ we have
\begin{align*}
J\lesssim&\int_0^v\int_{u_{\infty}}^u\frac{a}{|u'|^2}
||\phi_0^2(a^{\frac{1}{2}})^{k-1}\meth^k\zeta_4||^2_{L^2_{sc}(\mathcal{S}_{u',v'})}
+\int_0^v\int_{u_{\infty}}^u\frac{a}{|u'|^2}
||\phi_{0,1}^2(a^{\frac{1}{2}})^{k-1}\meth^k\zeta_{0,1,3}||^2_{L^2_{sc}(\mathcal{S}_{u',v'})}\\
&+\int_0^v\int_{u_{\infty}}^u\frac{a}{|u'|^2}
||\phi_0\zeta_4(a^{\frac{1}{2}})^{k-1}\meth^k\phi_0||^2_{L^2_{sc}(\mathcal{S}_{u',v'})}
+\int_0^v\int_{u_{\infty}}^u\frac{a}{|u'|^2}
||\phi_{0,1}\zeta_{0,1,3}(a^{\frac{1}{2}})^{k-1}\meth^k\phi_{0,1}||^2_{L^2_{sc}(\mathcal{S}_{u',v'})} \\
&+\int_0^v\int_{u_{\infty}}^u\frac{a}{|u'|^2}
||\Gamma(\sigma,\lambda,...)(a^{\frac{1}{2}})^{k-1}\meth^k\zeta_l||^2_{L^2_{sc}(\mathcal{S}_{u',v'})}
+\int_0^v\int_{u_{\infty}}^u\frac{a}{|u'|^2}
||\zeta_m(a^{\frac{1}{2}})^{k-1}\meth^k\Gamma||^2_{L^2_{sc}(\mathcal{S}_{u',v'})} \\
&+\sum_{i_1+i_2+i_3=k, i_l<k}\int_0^v\int_{u_{\infty}}^u\frac{a}{|u'|^2}
||(a^{\frac{1}{2}})^{k-1}\meth^{i_1}\phi_j\meth^{i_1}\phi_l
\meth^{i_1}\zeta_m||^2_{L^2_{sc}(\mathcal{S}_{u',v'})} \\
&+\sum_{i_1+i_2=k, i_l<k}\int_0^v\int_{u_{\infty}}^u\frac{a}{|u'|^2}
||(a^{\frac{1}{2}})^{k-1}\meth^{i_1}\Gamma\meth^{i_1}\zeta_m||^2_{L^2_{sc}(\mathcal{S}_{u',v'})} 
=J_1+...+J_8.
\end{align*}

We then have
\begin{align*}
J_1\leq&\int_0^v\int_{u_{\infty}}^u\frac{a}{|u'|^2}\frac{1}{|u'|^4}(a\phi[\phi_{0}]^2)^2
||(a^{\frac{1}{2}})^{k-1}\mathcal{D}^k\zeta_4||^2_{L^2_{sc}(\mathcal{S}_{u',v'})} \\
\lesssim&\int_{u_{\infty}}^u\frac{a^2}{|u'|^4}\frac{a}{|u'|^2}\int_0^v
||(a^{\frac{1}{2}})^{k-1}\mathcal{D}^k\zeta_4||^2_{L^2_{sc}(\mathcal{S}_{u',v'})} 
\leq\frac{a^2}{|u|^3}\bm\zeta[\zeta_4]^2.
\end{align*}
Here the definition of $\bm\zeta[\zeta_4]$ is
\begin{align*}
\bm\zeta[\zeta_4]\equiv
||\frac{1}{|u|}(a^{\frac{1}{2}})^{k}\mathcal{D}^k\zeta_4||_{L^2_{sc}(\mathcal{N}_u(0,v))}
=\left(\frac{a}{|u|^2}\int_0^v||a^{\frac{k-1}{2}}
\mathcal{D}^k\zeta_4||^2_{L^2_{sc}(\mathcal{S}_{u,v'})} \right)^{\frac{1}{2}}.
\end{align*}

\begin{align*}
J_2\leq&\int_0^v\int_{u_{\infty}}^u\frac{a}{|u'|^2}\frac{1}{|u'|^4}(a\phi[\phi_{0,1}]^2)^2
||(a^{\frac{1}{2}})^{k-1}\mathcal{D}^k\zeta_{0,1,3}||^2_{L^2_{sc}(\mathcal{S}_{u',v'})} \\
\leq&\frac{a^3}{|u|^4}\frac{1}{a}\int_{u_{\infty}}^u\frac{a}{|u'|^2}
||(a^{\frac{1}{2}})^{k-1}\mathcal{D}^k\zeta_{1,3}||^2_{L^2_{sc}(\mathcal{S}_{u',v'})}\\
&+\frac{a^4}{|u|^5}\frac{1}{a}\int_0^v
||(a^{\frac{1}{2}})^{k-1}\mathcal{D}^k\zeta_{0}||^2_{L^2_{sc}(\mathcal{S}_{u',v'})}
\leq\frac{a^3}{|u|^4}\underline{\bm\zeta}[\zeta_{1,3}]^2
+\frac{a^4}{|u|^5}\bm\zeta[\zeta_{0}]^2.
\end{align*}

\begin{align*}
J_3\leq&\int_0^v\int_{u_{\infty}}^u\frac{a}{|u'|^2}\frac{1}{|u'|^4}
a\phi[\phi_0]^2\frac{|u'|^2}{a}\zeta[\zeta_4]^2
||(a^{\frac{1}{2}})^{k-1}\mathcal{D}^k\phi_0||^2_{L^2_{sc}(\mathcal{S}_{u',v'})} \\
\lesssim&\int_{u_{\infty}}^u\frac{a^2}{|u'|^4}\phi[\phi_0]^2\zeta[\zeta_4]^2
\frac{1}{a}\int_0^v||(a^{\frac{1}{2}})^{k-1}\mathcal{D}^k\phi_0||^2_{L^2_{sc}(\mathcal{S}_{u',v'})} \\
\leq&\frac{a^2}{|u|^3}\phi[\phi_0]^2\zeta[\zeta_4]^2\bm\phi[\phi_0]^2.
\end{align*}

\begin{align*}
J_4\leq&\int_0^v\int_{u_{\infty}}^u\frac{a}{|u'|^2}\frac{1}{|u'|^4}a\phi[\phi_{0,1}]^2
(\zeta[\zeta_{1,3}]^2,a\zeta[\zeta_{0}]^2)
||(a^{\frac{1}{2}})^{k-1}\mathcal{D}^k\phi_{0,1}||^2_{L^2_{sc}(\mathcal{S}_{u',v'})} \\
\lesssim&\frac{a^4}{|u|^5}\phi[\phi_{0,1}]^2\zeta[\zeta_{0}]^2\bm\phi[\phi_0]^2
+\frac{a^3}{|u|^4}\phi[\phi_{0,1}]^2\zeta[\zeta_{0}]^2\underline{\bm\phi}[\phi_1]^2.
\end{align*}

\begin{align*}
J_5\leq&\int_0^v\int_{u_{\infty}}^u\frac{a}{|u'|^2}(
||\Gamma(\lambda,\ulomega)(a^{\frac{1}{2}})^{k-1}\mathcal{D}^k\zeta_0||^2_{L^2_{sc}(\mathcal{S}_{u',v'})}
+||\sigma(a^{\frac{1}{2}})^{k-1}\mathcal{D}^k\zeta_2||^2_{L^2_{sc}(\mathcal{S}_{u',v'})} \\
&+||\tau(a^{\frac{1}{2}})^{k-1}\mathcal{D}^k\zeta_{1,3}||^2_{L^2_{sc}(\mathcal{S}_{u',v'})}
+||\rho(a^{\frac{1}{2}})^{k-1}\mathcal{D}^k\zeta_4||^2_{L^2_{sc}(\mathcal{S}_{u',v'})}) \\
\leq&\big(\frac{a}{|u|}\Gamma[\lambda]^2+\frac{a^2}{|u|^3}\Gamma[\ulomega]^2\big)\bm\zeta[\zeta_0]^2
+\frac{a}{|u|^2}\Gamma[\sigma]^2\underline{\bm\zeta}[\zeta_2]^2 \\
&+\frac{a}{|u|^2}\Gamma[\tau]^2\underline{\bm\zeta}[\zeta_{1,3}]^2
+\frac{1}{|u|}\Gamma[\rho]^2\bm\zeta[\zeta_4]^2
\lesssim\frac{a}{|u|}\bm\zeta[\zeta_0]^2+\frac{\mathcal{O}^4}{a}.
\end{align*}

\begin{align*}
J_6\leq&\int_0^v\int_{u_{\infty}}^u\frac{a}{|u'|^2}(
||\zeta_0(a^{\frac{1}{2}})^{k-1}\mathcal{D}^k\Gamma(\Timu,\ulomega,\lambda)||^2_{L^2_{sc}(\mathcal{S}_{u',v'})}
+||\zeta_2(a^{\frac{1}{2}})^{k-1}\mathcal{D}^k\sigma||^2_{L^2_{sc}(\mathcal{S}_{u',v'})} \\
&+||\zeta_{1,3}(a^{\frac{1}{2}})^{k-1}\mathcal{D}^k\tau||^2_{L^2_{sc}(\mathcal{S}_{u',v'})}
+||\zeta_4(a^{\frac{1}{2}})^{k-1}\mathcal{D}^k\rho||^2_{L^2_{sc}(\mathcal{S}_{u',v'})}) \\
\lesssim&\frac{a}{|u|^2}\zeta[\zeta_0]^2\underline{\bm\Gamma}_{11}[\ulomega]^2
+\frac{1}{a}\zeta[\zeta_0]^2\underline{\bm\Gamma}_{11}[\Timu]^2
+\zeta[\zeta_0]^2\underline{\bm\Gamma}_{11}[\lambda]^2
+\frac{a}{|u|^3}\zeta[\zeta_2]^2\bm\Gamma_{11}[\sigma]^2 \\
&+\frac{1}{|u|^2}\zeta[\zeta_{1,3}]^2\underline{\bm\Gamma}_{11}[\tau]^2
+\frac{1}{|u|}\zeta[\zeta_4]^2\bm\Gamma_{11}[\rho]^2
\lesssim1+\frac{\mathcal{O}^4}{a}.
\end{align*}

\begin{align*}
J_7\leq&\frac{1}{a}\int_0^v\int_{u_{\infty}}^u\frac{a}{|u'|^2}\frac{1}{|u|^4}
\frac{|u|^2}{a}\zeta[\zeta_4]^2a^2\phi[\phi_{0,1}]^4
\lesssim\frac{1}{|u|^3}\mathcal{O}^6.
\end{align*}

\begin{align*}
J_8\leq&\frac{1}{a}\int_0^v\int_{u_{\infty}}^u\frac{a}{|u'|^2}\frac{1}{|u|^2}
\frac{|u'|^2}{a}\mathcal{O}^4\lesssim\frac{1}{a|u|}\mathcal{O}^4.
\end{align*}

Collect the results one has
\begin{align*}
J\lesssim\frac{a}{|u|}\bm\zeta[\zeta_0]^2+1+\frac{\mathcal{O}^6}{a},
\end{align*}
and then 
\begin{align*}
M\leq\frac{a^{\frac{1}{2}}}{|u|^{\frac{1}{2}}}a^{\frac{1}{2}}\bm{\zeta}[\zeta_0]J^{\frac{1}{2}}
\lesssim\frac{a^{\frac{3}{2}}}{|u|}\bm{\zeta}[\zeta_0]^2
+\frac{a}{|u|^{\frac{1}{2}}}\bm{\zeta}[\zeta_0]
+\mathcal{O}^3\bm{\zeta}[\zeta_0].
\end{align*}

For $N$ We have
\begin{align*}
N\leq2\left(\int_0^v
||(a^{\frac{1}{2}})^{k-1}\meth^k\zeta_1||^2_{L^2_{sc}
(\mathcal{N}'_{v'}(u_{\infty},u))}\right)^{\frac{1}{2}}H^{\frac{1}{2}}
\lesssim a^{\frac{1}{2}}\underline{\bm\zeta}[\zeta_1]H^{\frac{1}{2}},
\end{align*}
where 
\begin{align*}
H\equiv\int_0^v\int_{u_{\infty}}^u\frac{a}{|u'|^2}||(a^{\frac{1}{2}})^{k-1}Q_k||^2_{L^2_{sc}(\mathcal{S}_{u',v'})}.
\end{align*}
Substitute $Q_k$ we have
\begin{align*}
H\lesssim&\int_0^v\int_{u_{\infty}}^u\frac{a}{|u'|^2}
||\phi_{0,1}^2(a^{\frac{1}{2}})^{k-1}\meth^k\zeta_{0,1,3}||^2_{L^2_{sc}(\mathcal{S}_{u',v'})}
+\int_0^v\int_{u_{\infty}}^u\frac{a}{|u'|^2}
||\phi_{0,1}\zeta_{0,1,3}(a^{\frac{1}{2}})^{k-1}\meth^k\phi_{0,1}||^2_{L^2_{sc}(\mathcal{S}_{u',v'})} \\
&+\int_0^v\int_{u_{\infty}}^u\frac{a}{|u'|^2}
||\Gamma(\sigma,...)(a^{\frac{1}{2}})^{k-1}\meth^k\zeta_l||^2_{L^2_{sc}(\mathcal{S}_{u',v'})}
+\int_0^v\int_{u_{\infty}}^u\frac{a}{|u'|^2}
||\zeta_m(a^{\frac{1}{2}})^{k-1}\meth^k\Gamma||^2_{L^2_{sc}(\mathcal{S}_{u',v'})} \\
&+\sum_{i_1+...+i_5=k, i_l<k}\int_0^v\int_{u_{\infty}}^u\frac{a}{|u'|^2}
||(a^{\frac{1}{2}})^{k-1}\meth^{i_1}\Gamma^{i_2}\meth^{i_3}\phi_j\meth^{i_4}\phi_l
\meth^{i_5}\zeta_m||^2_{L^2_{sc}(\mathcal{S}_{u',v'})} \\
&+\sum_{i_1+...+i_4=k, i_l<k}\int_0^v\int_{u_{\infty}}^u\frac{a}{|u'|^2}
||(a^{\frac{1}{2}})^{k-1}\meth^{i_1}\Gamma^{i_2}
\meth^{i_3}\Gamma\meth^{i_4}\zeta_m||^2_{L^2_{sc}(\mathcal{S}_{u',v'})} \\
&+\sum_{i_1+i_2=k-1}\int_0^v\int_{u_{\infty}}^u\frac{a}{|u'|^2}
||(a^{\frac{1}{2}})^{k-1}\meth^{i_1}K\meth^{i_2}\zeta_{0}||^2_{L^2_{sc}(\mathcal{S}_{u',v'})}
=H_1+...+H_7.
\end{align*}

From the analysis in $J$ we have
\begin{align*}
H_1\lesssim&\frac{a^3}{|u|^4}\underline{\bm\zeta}[\zeta_{1,3}]^2
+\frac{a^4}{|u|^5}\bm\zeta[\zeta_{0}]^2,
\end{align*}

\begin{align*}
H_2\lesssim&\frac{a^4}{|u|^5}\phi[\phi_{0,1}]^2\zeta[\zeta_{0}]^2\bm\phi[\phi_0]^2
+\frac{a^3}{|u|^4}\phi[\phi_{0,1}]^2\zeta[\zeta_{0}]^2\underline{\bm\phi}[\phi_0]^2,
\end{align*}

\begin{align*}
H_3\leq&\int_0^v\int_{u_{\infty}}^u\frac{a}{|u'|^2}(
||\Gamma(\tau,\pi)(a^{\frac{1}{2}})^{k-1}\mathcal{D}^k\zeta_0||^2_{L^2_{sc}(\mathcal{S}_{u',v'})}
+||\Gamma(\rho,\sigma)(a^{\frac{1}{2}})^{k-1}\mathcal{D}^k\zeta_{1,3}||^2_{L^2_{sc}(\mathcal{S}_{u',v'})} ) \\
\lesssim&\frac{a^2}{|u|^3}\Gamma[\tau,\pi]^2\bm\zeta[\zeta_0]^2
+\frac{a^2}{|u|^2}\Gamma[\sigma]^2\underline{\bm\zeta}[\zeta_{1,3}]^2,
\end{align*}

\begin{align*}
H_4\leq&\int_0^v\int_{u_{\infty}}^u\frac{a}{|u'|^2}(
||\zeta_0(a^{\frac{1}{2}})^{k-1}\meth^k\Gamma(\tau,\pi)||^2_{L^2_{sc}(\mathcal{S}_{u',v'})}
+||\zeta_{1,3}(a^{\frac{1}{2}})^{k-1}\meth^k\Gamma(\rho,\sigma)||^2_{L^2_{sc}(\mathcal{S}_{u',v'})}) \\
\lesssim&\frac{a}{|u|^2}\zeta[\zeta_0]^2\underline{\bm\Gamma}_{11}[\tau,\pi]^2
+\frac{a}{|u|^3}\zeta[\zeta_{1,3}]^2\bm\Gamma_{11}[\sigma]^2,
\end{align*}

\begin{align*}
H_5\leq&\frac{1}{a}\int_0^v\int_{u_{\infty}}^u\frac{a}{|u'|^2}
\frac{1}{|u'|^4}a\zeta[\zeta_0]^2a^2\phi[\phi_{0,1}]^4
\lesssim\frac{a^3}{|u|^5}\zeta[\zeta_0]^2\phi[\phi_{0,1}]^4,
\end{align*}

\begin{align*}
H_6\leq&\frac{1}{a}\int_0^v\int_{u_{\infty}}^u\frac{a}{|u'|^2}
\frac{1}{|u'|^2}a\mathcal{O}^4\lesssim\frac{a}{|u|^3}\mathcal{O}^4,
\end{align*}

\begin{align*}
H_7\lesssim&
\sum_{i=0}^{k-1}\int_0^v\int_{u_{\infty}}^u\frac{a}{|u'|^2}\frac{1}{|u'|^2}||(a^{\frac{1}{2}}\mathcal{D})^{i}K||^2
||(a^{\frac{1}{2}}\mathcal{D})^{k-1-i}\zeta_0||^2 \\
\lesssim&\int_0^v\frac{a^2}{|u|^4}\zeta[\zeta_{0}]^2
\int_0^v||(a^{\frac{1}{2}}\mathcal{D})^{k-1}K||^2_{L^2_{sc}(\mathcal{S}_{u',v'})}
\lesssim\frac{a^2}{|u|^3}\zeta[\zeta_{0}]^2
||(a^{\frac{1}{2}}\mathcal{D})^{k-1}K||^2_{L^2_{sc}(\mathcal{N}_{u}(0,v))}.
\end{align*}

Collect the results one has
\begin{align*}
H\lesssim\frac{a^2}{|u|^2}\Gamma[\sigma]^2\underline{\bm\zeta}[\zeta_{1,3}]^2+\frac{\mathcal{O}^6}{a},
\end{align*}
and then
\begin{align*}
N\lesssim a^{\frac{1}{2}}\underline{\bm\zeta}[\zeta_1](\frac{a}{|u|}\underline{\bm\zeta}[\zeta_{1}]
+\frac{\mathcal{O}^3}{a^{\frac{1}{2}}}).
\end{align*}
Combine with $M$ one finally obtain

\begin{align*}
&\frac{1}{a}||(a^{\frac{1}{2}})^{k-1}\meth^k\zeta_0||^2_{L^2_{sc}(\mathcal{N}_u(0,v))}
+\frac{1}{a}||(a^{\frac{1}{2}})^{k-1}\meth^k\zeta_1||^2_{L^2_{sc}(\mathcal{N}'_v(u_{\infty},u))} 
\lesssim\mathcal{I}^2_0+\frac{1}{a^{\frac{1}{4}}}.
\end{align*}

\end{proof}

\begin{proposition}
\label{EnergyEstimatezeta12}
For $0\leq k\leq11$,  one has that 
\begin{align*}
||(a^{\frac{1}{2}})^{k-1}\mathcal{D}^k\zeta_1||^2_{L^2_{sc}(\mathcal{N}_u(0,v))}
+||(a^{\frac{1}{2}})^{k-1}\mathcal{D}^k\zeta_2||^2_{L^2_{sc}(\mathcal{N}'_v(u_{\infty},u))} 
\leq\mathcal{I}^2_0+1.
\end{align*}
\end{proposition}

\begin{proof}
For pair $(\zeta_1,\zeta_2)$ we make use of 
\begin{align*}
\mthorn'\zeta_1- \meth\,\zeta_2 &= 
- \mathrm{i}\,\zeta_5\,\phi_0\,\bar{\phi}_0
+ \mathrm{i}\,\bar{\zeta}_4\,\phi_0\,\phi_1
+ \mathrm{i}\,\zeta_4\,\bar{\phi}_0\,\phi_1
- \mathrm{i}\,\bar{\zeta}_1\,\phi_1^2 
- \mathrm{i}\,\zeta_2\,\phi_0\,\bar{\phi}_1\\
&\quad
+ \mathrm{i}\,\zeta_1\,\phi_1\,\bar{\phi}_1
- 2\,\zeta_1\,\mu+\frac{\zeta_1\,\ulomega}{2}
+ \zeta_5\,\rho
- \frac{\zeta_2\,\tau}{2}
- \zeta_4\,\bar{\tau}, \\
\mthorn\,\zeta_2- \meth'\,\zeta_1 &= 2\,\mathrm{i}\,\zeta_2\,\phi_0\,\bar{\phi}_0
- \mathrm{i}\,\bar{\zeta}_3\,\phi_0\,\phi_1
- 2\,\mathrm{i}\,\zeta_1\,\bar{\phi}_0\,\phi_1
+ \mathrm{i}\,\bar{\zeta}_0\,\phi_1^2
- \zeta_0\,\lambda \\
&\quad
+ 2\,\zeta_1\,\pi
+ \zeta_2\,\rho
+ \zeta_4\,\bar{\sigma}
- \frac{\zeta_1\,\tau}{2} \,.
\end{align*}

Applying $\meth^k$ and commuting with $\mthorn$ and $\mthorn'$ we have that
\begin{align*}
&\mthorn'\meth^k\zeta_1+(k+2)\mu\meth^k\zeta_1-\meth^{k+1}\zeta_2=P_k, \\
&\mthorn\meth^k\zeta_2-\meth'\meth^k\zeta_1=Q_k,
\end{align*}
where
\begin{align*}
P_k=&\sum_{i_1+i_2+i_3=k}\meth^{i_1}\phi_j\meth^{i_1}\phi_l\meth^{i_1}\zeta_m
+\sum_{i_1+i_2=k}\meth^{i_1}\Gamma\meth^{i_1}\zeta_m,
\end{align*}
\begin{align*}
Q_k=&\sum_{i_1+...+i_5=k}\meth^{i_1}\Gamma^{i_2}\meth^{i_3}\phi_j\meth^{i_4}\phi_l\meth^{i_5}\zeta_m
+\sum_{i_1+...+i_4=k}\meth^{i_1}\Gamma^{i_2}\meth^{i_3}\Gamma\meth^{i_4}\zeta_m \\
&+\sum_{i_1+i_2=k-1}\meth^{i_1}K\meth^{i_2}\zeta_{0}.
\end{align*}

Then apply the energy inequality one has
\begin{align*}
&||(a^{\frac{1}{2}})^{k-1}\meth^k\zeta_1||^2_{L^2_{sc}(\mathcal{N}_u(0,v))}
+||(a^{\frac{1}{2}})^{k-1}\meth^k\zeta_2||^2_{L^2_{sc}(\mathcal{N}'_v(u_{\infty},u))} 
\lesssim\mathcal{I}^2_0+(M+N),
\end{align*}
where
\begin{align*}
M\equiv&2\int_0^v\int_{u_{\infty}}^u\frac{a}{|u'|^2}
||(a^{\frac{1}{2}})^{k-1}\meth^k\zeta_1||_{L^2_{sc}(\mathcal{S}_{u',v'})}
||(a^{\frac{1}{2}})^{k-1}P_k||_{L^2_{sc}(\mathcal{S}_{u',v'})}, \\
N\equiv&2\int_0^v\int_{u_{\infty}}^u\frac{a}{|u'|^2}
||(a^{\frac{1}{2}})^{k-1}\meth^k\zeta_2||_{L^2_{sc}(\mathcal{S}_{u',v'})}
||(a^{\frac{1}{2}})^{k-1}Q_k||_{L^2_{sc}(\mathcal{S}_{u',v'})}. \\
\end{align*}

For $M$ one need estimate

\begin{align*}
J\equiv\int_0^v\int_{u_{\infty}}^u\frac{a}{|u'|^2}||(a^{\frac{1}{2}})^{k-1}P_k||^2_{L^2_{sc}(\mathcal{S}_{u',v'})}.
\end{align*}
Compared with the analysis in pair $(\zeta_0,\zeta_1)$, we focus on the different terms. 
For terms $\zeta_{5}\phi_0\bar\phi_0$ and $\zeta_2\phi_0\bar\phi_1$ which lead to the following:
\begin{align*}
J_1\equiv&\int_0^v\int_{u_{\infty}}^u\frac{a}{|u|^2}
||\phi_{0}^2(a^{\frac{1}{2}})^{k-1}\meth^k\zeta_{5}||^2_{L^2_{sc}(\mathcal{S}_{u',v'})}
+\int_0^v\int_{u_{\infty}}^u\frac{a}{|u'|^2}
||\phi_{0}\zeta_{5}(a^{\frac{1}{2}})^{k-1}\meth^k\phi_{0}||^2_{L^2_{sc}(\mathcal{S}_{u',v'})} \\
&+\int_0^v\int_{u_{\infty}}^u\frac{a}{|u|^2}
||\phi_{0,1}^2(a^{\frac{1}{2}})^{k-1}\meth^k\zeta_{2}||^2_{L^2_{sc}(\mathcal{S}_{u',v'})}
+\int_0^v\int_{u_{\infty}}^u\frac{a}{|u'|^2}
||\phi_{0,1}\zeta_{2}(a^{\frac{1}{2}})^{k-1}\meth^k\phi_{0,1}||^2_{L^2_{sc}(\mathcal{S}_{u',v'})} \\
&+\sum_{i_1+i_2+i_3=k, i_l<k}\int_0^v\int_{u_{\infty}}^u\frac{a}{|u'|^2}
||(a^{\frac{1}{2}})^{k-1}\meth^{i_1}\phi_j\meth^{i_1}\phi_l
\meth^{i_1}\zeta_{2,5}||^2_{L^2_{sc}(\mathcal{S}_{u',v'})}
=J_{11}+...+J_{15}.
\end{align*}

One has the control
\begin{align*}
J_{11}\lesssim&\int_0^v\int_{u_{\infty}}^u\frac{a}{|u|^2}\frac{a^2}{|u|^4}\phi[\phi_{0}]^4
||(a^{\frac{1}{2}})^{k-1}\meth^k\zeta_{5}||^2_{L^2_{sc}(\mathcal{S}_{u',v'})}\\
\lesssim&\frac{a}{|u|^2}\phi[\phi_{0}]^4\underline{\bm\zeta}[\zeta_5]^2.
\end{align*}
Here we make use of the definition
\begin{align*}
\underline{\bm\zeta}[\zeta_5]\equiv
||\frac{1}{|u'|}(a^{\frac{1}{2}})^{k}\mathcal{D}^k\zeta_5||_{L^2_{sc}(\mathcal{N}'_v(u_{\infty},u))}\equiv\left(\int_{u_{\infty}}^u\frac{a^2}{|u'|^4}
||a^{\frac{k-1}{2}}\mathcal{D}^k\zeta_5||^2_{L^2_{sc}(\mathcal{S}_{u,v'})} \right)^{\frac{1}{2}}.
\end{align*}

$J_{12}$ is similar to $J_3$ in the analysis of pair $(\zeta_0,\zeta_1)$ and hence one has
\begin{align*}
J_{12}\lesssim\frac{a^2}{|u|^3}\phi[\phi_0]^2\zeta[\zeta_5]^2\bm\phi[\phi_0]^2.
\end{align*}

\begin{align*}
J_{13}\lesssim&\int_0^v\int_{u_{\infty}}^u\frac{a}{|u|^2}\frac{a^2}{|u|^4}\phi[\phi_{0,1}]^4
||(a^{\frac{1}{2}})^{k-1}\meth^k\zeta_{2}||^2_{L^2_{sc}(\mathcal{S}_{u',v'})}\\
\lesssim&\frac{a^2}{|u|^4}\phi[\phi_{0,1}]^4\underline{\bm\zeta}[\zeta_2]^2.
\end{align*}

$J_{14}$ is similar to $J_4$ in the analysis of pair $(\zeta_0,\zeta_1)$ and hence one has
\begin{align*}
J_{14}\lesssim\frac{a^3}{|u|^5}\phi[\phi_{0,1}]^2\zeta[\zeta_{2}]^2\bm\phi[\phi_0]^2
+\frac{a^2}{|u|^4}\phi[\phi_{0,1}]^2\zeta[\zeta_{2}]^2\underline{\bm\phi}[\phi_1]^2.
\end{align*}

$J_{15}$ is similar to $J_7$ in the analysis of pair $(\zeta_0,\zeta_1)$ and hence one has
\begin{align*}
J_{15}\leq&\frac{1}{a}\int_0^v\int_{u_{\infty}}^u\frac{a}{|u'|^2}\frac{1}{|u|^4}
\frac{|u|^2}{a}\zeta[\zeta_4]^2a^2\phi[\phi_{0,1}]^4
\lesssim\frac{1}{|u|^3}\mathcal{O}^6.
\end{align*}

For terms $\Gamma\zeta_l$ one need estimate
\begin{align*}
J_2\equiv&\int_0^v\int_{u_{\infty}}^u\frac{a}{|u'|^2}
||\Gamma(a^{\frac{1}{2}})^{k-1}\meth^k\zeta_l||^2_{L^2_{sc}(\mathcal{S}_{u',v'})}
+\int_0^v\int_{u_{\infty}}^u\frac{a}{|u'|^2}
||\zeta_m(a^{\frac{1}{2}})^{k-1}\meth^k\Gamma||^2_{L^2_{sc}(\mathcal{S}_{u',v'})} \\
&+\sum_{i_1+i_2=k, i_l<k}\int_0^v\int_{u_{\infty}}^u\frac{a}{|u'|^2}
||(a^{\frac{1}{2}})^{k-1}\meth^{i_1}\Gamma\meth^{i_1}\zeta_m||^2_{L^2_{sc}(\mathcal{S}_{u',v'})} 
=J_{21}+J_{22}+J_{23}.
\end{align*}
We have
\begin{align*}
J_{21}\lesssim&\int_0^v\int_{u_{\infty}}^u\frac{a}{|u'|^2}
(||\Gamma(\lambda,\ulomega)(a^{\frac{1}{2}})^{k-1}\meth^k\zeta_1||^2_{L^2_{sc}(\mathcal{S}_{u',v'})}
+||\tau(a^{\frac{1}{2}})^{k-1}\meth^k\zeta_2||^2_{L^2_{sc}(\mathcal{S}_{u',v'})} \\
&+||\tau(a^{\frac{1}{2}})^{k-1}\meth^k\zeta_4||^2_{L^2_{sc}(\mathcal{S}_{u',v'})}
+||\rho(a^{\frac{1}{2}})^{k-1}\meth^k\zeta_5||^2_{L^2_{sc}(\mathcal{S}_{u',v'})}) \\
\lesssim&\frac{1}{|u|}\Gamma[\lambda]^2\bm\zeta[\zeta_1]^2
+\frac{1}{|u|^2}\Gamma[\tau]^2\underline{\bm\zeta}[\zeta_2]^2
+\frac{1}{|u|}\Gamma[\tau]^2\bm\zeta[\zeta_4]^2
+\frac{1}{a}\Gamma[\rho]^2\underline{\bm\zeta}[\zeta_5]^2.
\end{align*}

\begin{align*}
J_{22}\leq&\int_0^v\int_{u_{\infty}}^u\frac{a}{|u'|^2}(
||\zeta_1(a^{\frac{1}{2}})^{k-1}\mathcal{D}^k\Gamma(\Timu,\ulomega,\lambda)||^2_{L^2_{sc}(\mathcal{S}_{u',v'})}
+||\zeta_2(a^{\frac{1}{2}})^{k-1}\mathcal{D}^k\tau||^2_{L^2_{sc}(\mathcal{S}_{u',v'})} \\
&+||\zeta_4(a^{\frac{1}{2}})^{k-1}\mathcal{D}^k\tau||^2_{L^2_{sc}(\mathcal{S}_{u',v'})}
+||\zeta_5(a^{\frac{1}{2}})^{k-1}\mathcal{D}^k\rho||^2_{L^2_{sc}(\mathcal{S}_{u',v'})}) \\
\lesssim&\frac{1}{a^2}\zeta[\zeta_1]^2\underline{\bm\Gamma}_{11}[\ulomega,\Timu]
+\frac{1}{a}\zeta[\zeta_1]^2\underline{\bm\Gamma}_{11}[\lambda]
+\frac{1}{|u|^2}\zeta[\zeta_2]^2\underline{\bm\Gamma}_{11}[\tau]
+\frac{1}{a}\zeta[\zeta_4]^2\underline{\bm\Gamma}_{11}[\tau]
+\frac{1}{|u|}\zeta[\zeta_5]^2\bm\Gamma_{11}[\rho].
\end{align*}

\begin{align*}
J_{23}\leq&\int_0^v\int_{u_{\infty}}^u\frac{a}{|u'|^2}\frac{1}{|u'|^2}\frac{|u'|^2}{a}\mathcal{O}^2
\mathcal{O}^2\lesssim\frac{1}{|u|}\mathcal{O}^4.
\end{align*}

Collect the results one has 
\begin{align*}
J_1\lesssim1.
\end{align*}

For the estimate of $N$ one need to analyse 
\begin{align*}
H\equiv\int_0^v\int_{u_{\infty}}^u\frac{a}{|u'|^2}||(a^{\frac{1}{2}})^{k-1}Q_k||^2_{L^2_{sc}(\mathcal{S}_{u',v'})}.
\end{align*}

For terms $\zeta_{1,2,3}\phi_j^2$, we already have control  
which can be bounded by $\frac{\mathcal{O}^6}{a}$.
Compared with pair $(\zeta_0,\zeta_1)$, one need estimate, 
$\zeta_0\lambda$, $\zeta_1\Gamma(\tau,\pi)$, $\zeta_2\rho$ and $\zeta_4\bar\sigma$. 
One then has
\begin{align*}
\tilde{H}\lesssim
&\int_0^v\int_{u_{\infty}}^u\frac{a}{|u'|^2}
||\Gamma(a^{\frac{1}{2}})^{k-1}\meth^k\zeta_l||^2_{L^2_{sc}(\mathcal{S}_{u',v'})}
+\int_0^v\int_{u_{\infty}}^u\frac{a}{|u'|^2}
||\zeta_m(a^{\frac{1}{2}})^{k-1}\meth^k\Gamma||^2_{L^2_{sc}(\mathcal{S}_{u',v'})} \\
&+\sum_{i_1+...+i_4=k, i_l<k}\int_0^v\int_{u_{\infty}}^u\frac{a}{|u'|^2}
||(a^{\frac{1}{2}})^{k-1}\meth^{i_1}\Gamma^{i_2}
\meth^{i_3}\Gamma\meth^{i_4}\zeta_m||^2_{L^2_{sc}(\mathcal{S}_{u',v'})} \\
&+\sum_{i_1+i_2=k-1}\int_0^v\int_{u_{\infty}}^u\frac{a}{|u'|^2}
||(a^{\frac{1}{2}})^{k-1}\meth^{i_1}K\meth^{i_2}\zeta_{1}||^2_{L^2_{sc}(\mathcal{S}_{u',v'})}
=\tilde{H}_1+...+\tilde{H}_4.
\end{align*}

\begin{align*}
\tilde{H}_1\lesssim&\int_0^v\int_{u_{\infty}}^u\frac{a}{|u'|^2}
(||\lambda(a^{\frac{1}{2}})^{k-1}\meth^k\zeta_0||^2_{L^2_{sc}(\mathcal{S}_{u',v'})}
+||\Gamma(\tau,\pi)(a^{\frac{1}{2}})^{k-1}\meth^k\zeta_1||^2_{L^2_{sc}(\mathcal{S}_{u',v'})} \\
&+||\Gamma(\rho,\sigma)(a^{\frac{1}{2}})^{k-1}\meth^k\zeta_2||^2_{L^2_{sc}(\mathcal{S}_{u',v'})}
+||\sigma(a^{\frac{1}{2}})^{k-1}\meth^k\zeta_4||^2_{L^2_{sc}(\mathcal{S}_{u',v'})}) \\
\lesssim&\frac{a}{|u|}\Gamma[\lambda]^2\bm\zeta[\zeta_0]^2
+\frac{a}{|u|^3}\Gamma[\tau,\pi]^2\bm\zeta[\zeta_1]^2
+\frac{1}{|u|^2}\Gamma[\rho]^2\underline{\bm\zeta}[\zeta_2]^2
+\frac{a}{|u|^2}\Gamma[\sigma]^2\underline{\bm\zeta}[\zeta_2]^2
+\frac{a}{|u|}\Gamma[\sigma]^2\bm\zeta[\zeta_4]^2\\
\lesssim&\bm\zeta[\zeta_4]^2+1.
\end{align*}

\begin{align*}
\tilde{H}_2\leq&\int_0^v\int_{u_{\infty}}^u\frac{a}{|u'|^2}(
||\zeta_0(a^{\frac{1}{2}})^{k-1}\mathcal{D}^k\lambda||^2_{L^2_{sc}(\mathcal{S}_{u',v'})}
+||\zeta_1(a^{\frac{1}{2}})^{k-1}\mathcal{D}^k\Gamma(\pi,\tau)||^2_{L^2_{sc}(\mathcal{S}_{u',v'})} \\
&+||\zeta_2(a^{\frac{1}{2}})^{k-1}\mathcal{D}^k\Gamma(\rho,\sigma)||^2_{L^2_{sc}(\mathcal{S}_{u',v'})}
+||\zeta_4(a^{\frac{1}{2}})^{k-1}\mathcal{D}^k\sigma||^2_{L^2_{sc}(\mathcal{S}_{u',v'})}) \\
\lesssim&\zeta[\zeta_0]^2\underline{\bm\Gamma}_{11}[\lambda]
+\frac{1}{a^2}\zeta[\zeta_1]^2\bm\Gamma_{11}[\tau,\pi]
+\frac{1}{a^2}\zeta[\zeta_2]^2\bm\Gamma_{11}[\rho,\sigma]
+\frac{1}{|u|}\zeta[\zeta_4]^2\bm\Gamma_{11}[\sigma] \\
\lesssim&\underline{\bm\zeta}[\zeta_1]^2+\underline{\bm\zeta}[\zeta_1]+1\lesssim1.
\end{align*}
Here we make use of the control for $\underline{\bm\zeta}[\zeta_1]$.

\begin{align*}
\tilde{H}_3\leq&\frac{1}{a}\int_0^v\int_{u_{\infty}}^u\frac{a}{|u'|^2}\frac{1}{|u'|^2}
(a\zeta[\zeta_0]^2\frac{|u'|^2}{a}\Gamma[\lambda]^2
+\frac{|u'|^2}{a}\zeta[\zeta_4]^2a\Gamma[\sigma]^2+\mathcal{O}^2) \\
\lesssim&\frac{1}{|u|}(\zeta[\zeta_0]^2\Gamma[\lambda]^2+\zeta[\zeta_4]^2\Gamma[\sigma]^2) 
\lesssim1.
\end{align*}

\begin{align*}
\tilde{H}_4\lesssim&
\sum_{i=0}^{k-1}\int_0^v\int_{u_{\infty}}^u\frac{a}{|u'|^2}\frac{1}{|u'|^2}||(a^{\frac{1}{2}}\mathcal{D})^{i}K||^2
||(a^{\frac{1}{2}}\mathcal{D})^{k-1-i}\zeta_1||^2 \\
\lesssim&\int_{u_{\infty}}^u\frac{a}{|u'|^4}\zeta[\zeta_{1}]^2
\int_0^v||(a^{\frac{1}{2}}\mathcal{D})^{k-1}K||^2_{L^2_{sc}(\mathcal{S}_{u',v'})}\\
\lesssim&\frac{a}{|u|^3}\zeta[\zeta_{1}]^2
||(a^{\frac{1}{2}}\mathcal{D})^{k-1}K||^2_{L^2_{sc}(\mathcal{N}_{u}(0,v))}.
\end{align*}

Then we have
\begin{align*}
H\lesssim\tilde{H}+\frac{\mathcal{O}^6}{a}
\lesssim\bm\zeta[\zeta_4]^2+1.
\end{align*}

Collect the results one has
\begin{align*}
&||(a^{\frac{1}{2}})^{k-1}\meth^k\zeta_1||^2_{L^2_{sc}(\mathcal{N}_u(0,v))}
+||(a^{\frac{1}{2}})^{k-1}\meth^k\zeta_2||^2_{L^2_{sc}(\mathcal{N}'_v(u_{\infty},u))} 
\lesssim\mathcal{I}^2_0+(M+N) \\
\lesssim&2\int_0^v\int_{u_{\infty}}^u\frac{a}{|u'|^2}
||(a^{\frac{1}{2}})^{k-1}\meth^k\zeta_1||^2_{L^2_{sc}(\mathcal{S}_{u',v'})}
+2\int_0^v\int_{u_{\infty}}^u\frac{a}{|u'|^2}
||(a^{\frac{1}{2}})^{k-1}\meth^k\zeta_2||^2_{L^2_{sc}(\mathcal{S}_{u',v'})} \\
&+J+H
\lesssim\mathcal{I}^2_0+\bm\zeta[\zeta_4]^2 
\lesssim\mathcal{I}^2_0+1.
\end{align*}
Here in the last step we make use of the results in \ref{EnergyEstimatezeta45}.

\end{proof}

\begin{remark}
\label{zetaphi2}
For terms $\zeta_l\phi_j^2$ we have control
\begin{align*}
\int_0^v\int_{u_{\infty}}^u\frac{a}{|u'|^2}
||(a^{\frac{1}{2}})^{k-1}\meth^k(\zeta_l\phi_j^2)||^2_{L^2_{sc}(\mathcal{S}_{u',v'})}
\lesssim\frac{\mathcal{O}^6}{a}.
\end{align*}
\end{remark}

\begin{proposition}
\label{EnergyEstimatezeta45}
For $0\leq k\leq11$,  one has that 
\begin{align*}
||\frac{a^{\frac{1}{2}}}{|u|}(a^{\frac{1}{2}})^{k-1}\mathcal{D}^k\zeta_4||^2_{L^2_{sc}(\mathcal{N}_u(0,v))}
+||\frac{a^{\frac{1}{2}}}{|u|}(a^{\frac{1}{2}})^{k-1}\mathcal{D}^k\zeta_5||^2_{L^2_{sc}(\mathcal{N}'_v(u_{\infty},u))} 
\leq\mathcal{I}^2_0+\frac{1}{a^{\frac{1}{4}}}.
\end{align*}
\end{proposition}

\begin{proof}
For pair $(\zeta_4,\zeta_5)$ we make use of 
\begin{align*}
\mthorn'\zeta_4 &:= -\frac{\zeta_4\,\ulomega}{2}
+ \mathrm{i}\,\bar{\zeta}_5\,\phi_0\,\phi_1
- \mathrm{i}\,\bar{\zeta}_2\,\phi_1^2
- 2\,\mathrm{i}\,\zeta_5\,\phi_0\,\bar{\phi}_1 
+ 2\,\mathrm{i}\,\zeta_4\,\phi_1\,\bar{\phi}_1
- \zeta_2\,\bar{\lambda}
- 2\,\zeta_4\,\mu
+ \frac{\zeta_5\,\tau}{2}
+ \meth\,\zeta_5, \\[6pt]
\mthorn\,\zeta_5 &:= \mathrm{i}\,\zeta_5\,\phi_0\,\bar{\phi}_0
- \mathrm{i}\,\bar{\zeta}_4\,\phi_0\,\phi_1
- \mathrm{i}\,\zeta_4\,\bar{\phi}_0\,\phi_1
+ \mathrm{i}\,\bar{\zeta}_1\,\phi_1^2
+ \mathrm{i}\,\zeta_2\,\phi_0\,\bar{\phi}_1
- \mathrm{i}\,\zeta_1\,\phi_1\,\bar{\phi}_1 \\ 
&\quad
- \zeta_3\,\lambda
- \zeta_1\,\mu
+ 2\,\zeta_4\,\pi
+ \zeta_2\,\bar{\pi}
+ \zeta_5\,\rho
+ \frac{\bar{\zeta}_4\,\bar{\tau}}{2}
+ \meth'\,\zeta_4.
\end{align*}

Applying $\meth^k$ and commuting with $\mthorn$ and $\mthorn'$ we have that
\begin{align*}
&\mthorn'\meth^k\zeta_4+(k+2)\mu\meth^k\zeta_1-\meth^{k+1}\zeta_5=P_k, \\
&\mthorn\meth^k\zeta_5-\meth'\meth^k\zeta_4=Q_k,
\end{align*}
where
\begin{align*}
P_k=&\sum_{i_1+i_2+i_3=k}\meth^{i_1}\phi_j\meth^{i_1}\phi_l\meth^{i_1}\zeta_m
+\sum_{i_1+i_2=k}\meth^{i_1}\Gamma\meth^{i_1}\zeta_m,
\end{align*}
\begin{align*}
Q_k=&\sum_{i_1+...+i_5=k}\meth^{i_1}\Gamma^{i_2}\meth^{i_3}\phi_j\meth^{i_4}\phi_l\meth^{i_5}\zeta_m
+\sum_{i_1+...+i_4=k}\meth^{i_1}\Gamma^{i_2}\meth^{i_3}\Gamma\meth^{i_4}\zeta_m \\
&+\sum_{i_1+i_2=k-1}\meth^{i_1}K\meth^{i_2}\zeta_{4}.
\end{align*}

Then apply the energy inequality one has
\begin{align*}
&||\frac{1}{|u|}(a^{\frac{1}{2}}\meth)^k\zeta_4||^2_{L^2_{sc}(\mathcal{N}_u(0,v))}
+||\frac{1}{|u|}(a^{\frac{1}{2}}\meth)^k\zeta_5||^2_{L^2_{sc}(\mathcal{N}'_v(u_{\infty},u))} \\
\lesssim&||\frac{1}{|u|}(a^{\frac{1}{2}}\meth)^k\zeta_4||^2_{L^2_{sc}(\mathcal{N}_{u_{\infty}}(0,v))}
+||\frac{1}{|u|}(a^{\frac{1}{2}}\meth)^k\zeta_5||^2_{L^2_{sc}(\mathcal{N}'_0(u_{\infty},u))} \\
&+2\int_0^v\int_{u_{\infty}}^u\frac{a}{|u'|^4}
||(a^{\frac{1}{2}}\meth)^k\zeta_4||_{L^2_{sc}(\mathcal{S}_{u',v'})}
||a^{\frac{k}{2}}P_k||_{L^2_{sc}(\mathcal{S}_{u',v'})} \\
&+2\int_0^v\int_{u_{\infty}}^u\frac{a}{|u'|^4}
||(a^{\frac{k}{2}}\meth)^k\zeta_5||_{L^2_{sc}(\mathcal{S}_{u',v'})}
||a^{\frac{k}{2}}Q_k||_{L^2_{sc}(\mathcal{S}_{u',v'})} \\
\lesssim&\left(\int_0^v\int_{u_{\infty}}^u\frac{a}{|u'|^4}
||(a^{\frac{1}{2}}\meth)^k\zeta_4||^2_{L^2_{sc}(\mathcal{S}_{u',v'})} \right)^{\frac{1}{2}}
\left(\int_0^v\int_{u_{\infty}}^u\frac{a}{|u'|^4}
||a^{\frac{k}{2}}P_k||^2_{L^2_{sc}(\mathcal{S}_{u',v'})} \right)^{\frac{1}{2}} \\
&+\left(\int_0^v\int_{u_{\infty}}^u\frac{a}{|u'|^4}
||(a^{\frac{1}{2}}\meth)^k\zeta_5||^2_{L^2_{sc}(\mathcal{S}_{u',v'})} \right)^{\frac{1}{2}}
\left(\int_0^v\int_{u_{\infty}}^u\frac{a}{|u'|^4}
||a^{\frac{k}{2}}Q_k||^2_{L^2_{sc}(\mathcal{S}_{u',v'})} \right)^{\frac{1}{2}}\\
\lesssim&\mathcal{I}_0^2+\bm\zeta[\zeta_4]J^{\frac{1}{2}}+\underline{\bm\zeta}[\zeta_5]H^{\frac{1}{2}},
\end{align*}
where
\begin{align*}
J\equiv\int_0^v\int_{u_{\infty}}^u\frac{a}{|u'|^4}||a^{\frac{k}{2}}P_k||^2_{L^2_{sc}(\mathcal{S}_{u',v'})}, \quad
H\equiv\int_0^v\int_{u_{\infty}}^u\frac{a}{|u'|^4}||a^{\frac{k}{2}}Q_k||^2_{L^2_{sc}(\mathcal{S}_{u',v'})}.
\end{align*}
For terms $\zeta_l\phi_j^2$ in $J$ and $H$, 
make use of the results in Remark \ref{zetaphi2} we have
\begin{align*}
&\int_0^v\int_{u_{\infty}}^u\frac{a}{|u'|^4}||a^{\frac{k}{2}}\meth^k
(\zeta_l\phi_j^2)||^2_{L^2_{sc}(\mathcal{S}_{u',v'})} 
\lesssim\frac{a}{|u|^2}\int_0^v\int_{u_{\infty}}^u\frac{a}{|u'|^2}||a^{\frac{k-1}{2}}\meth^k
(\zeta_l\phi_j^2)||^2_{L^2_{sc}(\mathcal{S}_{u',v'})}
\lesssim\frac{\mathcal{O}^6}{|u|^2}.
\end{align*}
For terms $\Gamma\zeta_l$ we have
\begin{align*}
I\lesssim
&\int_0^v\int_{u_{\infty}}^u\frac{a^2}{|u'|^4}
||\Gamma(a^{\frac{1}{2}})^{k-1}\meth^k\zeta_l||^2_{L^2_{sc}(\mathcal{S}_{u',v'})}
+\int_0^v\int_{u_{\infty}}^u\frac{a^2}{|u'|^4}
||\zeta_m(a^{\frac{1}{2}})^{k-1}\meth^k\Gamma||^2_{L^2_{sc}(\mathcal{S}_{u',v'})} \\
&+\sum_{i_1+...+i_4=k, i_l<k}\int_0^v\int_{u_{\infty}}^u\frac{a}{|u'|^4}
||(a^{\frac{1}{2}})^{k}\meth^{i_1}\Gamma^{i_2}
\meth^{i_3}\Gamma\meth^{i_4}\zeta_m||^2_{L^2_{sc}(\mathcal{S}_{u',v'})} 
=I_1+I_2+I_3.
\end{align*}

Then one has
\begin{align*}
I_1\leq\int_0^v\int_{u_{\infty}}^u\frac{a^2}{|u'|^4}(&
||\mu(a^{\frac{1}{2}})^{k-1}\meth^k\zeta_1||^2_{L^2_{sc}(\mathcal{S}_{u',v'})}+
||\Gamma(\lambda,\pi)(a^{\frac{1}{2}})^{k-1}\meth^k\zeta_2||^2_{L^2_{sc}(\mathcal{S}_{u',v'})}\\
&+||\lambda(a^{\frac{1}{2}})^{k-1}\meth^k\zeta_3||^2_{L^2_{sc}(\mathcal{S}_{u',v'})}
+||\Gamma(\tau,\pi,\ulomega,\lambda)(a^{\frac{1}{2}})^{k-1}\meth^k\zeta_4||^2_{L^2_{sc}(\mathcal{S}_{u',v'})}\\
&+||\Gamma(\tau,\rho,\sigma)(a^{\frac{1}{2}})^{k-1}\meth^k\zeta_5||^2_{L^2_{sc}(\mathcal{S}_{u',v'})}) \\
\lesssim&\frac{1}{|u|}\Gamma[\mu]^2\bm\zeta[\zeta_1]^2
+\frac{a}{|u|^4}\Gamma[\tau,\pi]^2\underline{\bm\zeta}[\zeta_2]^2
+\frac{a}{|u|^2}\Gamma[\lambda]^2\underline{\bm\zeta}[\zeta_3]^2 \\
&+\frac{a}{|u|^3}\Gamma[\tau,\pi,\ulomega]^2\bm\zeta[\zeta_4]^2
+\frac{1}{|u|}\Gamma[\lambda]^2\bm\zeta[\zeta_4]^2 \\
&+\frac{1}{|u|^2}\Gamma[\tau,\rho]^2\underline{\bm\zeta}[\zeta_5]^2
+\frac{a}{|u|^2}\Gamma[\sigma]^2\underline{\bm\zeta}[\zeta_5]^2.
\end{align*}

\begin{align*}
I_2\leq\int_0^v\int_{u_{\infty}}^u\frac{a^2}{|u'|^4}&(
||\zeta_1(a^{\frac{1}{2}})^{k-1}\meth^k\Timu||^2_{L^2_{sc}(\mathcal{S}_{u',v'})}
+||\zeta_2(a^{\frac{1}{2}})^{k-1}\meth^k\Gamma(\lambda,\pi)||^2_{L^2_{sc}(\mathcal{S}_{u',v'})}\\
&+||\zeta_3(a^{\frac{1}{2}})^{k-1}\meth^k\lambda||^2_{L^2_{sc}(\mathcal{S}_{u',v'})}
+||\zeta_4(a^{\frac{1}{2}})^{k-1}\meth^k\Gamma(\tau,\pi,\ulomega,
\Timu,\lambda)||^2_{L^2_{sc}(\mathcal{S}_{u',v'})}\\
&+||\zeta_5(a^{\frac{1}{2}})^{k-1}\meth^k\Gamma(\tau,\rho,\sigma)||^2_{L^2_{sc}(\mathcal{S}_{u',v'})}) \\
\lesssim&\frac{1}{a|u|^2}\zeta[\zeta_1]^2\underline{\bm\Gamma}_{11}[\Timu]^2
+(\frac{1}{|u|^2}\underline{\bm\Gamma}[\lambda]_{11}^2+\frac{1}{|u|^3}\bm\Gamma[\pi]_{11}^2)\zeta[\zeta_{2,3}]
+\frac{1}{a}\zeta[\zeta_{4,5}]^2\bm\Gamma_{11}.
\end{align*}

\begin{align*}
I_3\leq\int_0^v\int_{u_{\infty}}^u\frac{a^2}{|u'|^4}\frac{1}{|u'|^2}\frac{|u|^2}{a}\zeta^2\frac{|u|^2}{a}\Gamma^2
\lesssim\frac{\mathcal{O}^4}{a}.
\end{align*}

And for the last term contain Gaussian curvature we have
\begin{align*}
\sum_{i=0}^{k-1}\int_0^v\int_{u_{\infty}}^u\frac{a^2}{|u'|^4}
||\meth^{i_1}K\meth^{i_2}\zeta_{4}||^2_{L^2_{sc}(\mathcal{S}_{u',v'})}\leq&
\sum_{i=0}^{k-1}\int_0^v\int_{u_{\infty}}^u\frac{a^2}{|u'|^4}\frac{1}{|u'|^2}
||(a^{\frac{1}{2}}\mathcal{D})^{i}K||^2
||(a^{\frac{1}{2}}\mathcal{D})^{k-1-i}\zeta_4||^2 \\
\lesssim&\frac{a}{|u|^3}\zeta[\zeta_{4}]^2
||(a^{\frac{1}{2}}\mathcal{D})^{k-1}K||^2_{L^2_{sc}(\mathcal{N}_{u}(0,v))}.
\end{align*}

Collect the results above one concludes 
\begin{align*}
J+H\lesssim\frac{\mathcal{O}^6}{a},
\end{align*}
and that leads to
\begin{align*}
||\frac{1}{|u|}(a^{\frac{1}{2}}\meth)^k\zeta_4||^2_{L^2_{sc}(\mathcal{N}_u(0,v))}
+||\frac{1}{|u|}(a^{\frac{1}{2}}\meth)^k\zeta_5||^2_{L^2_{sc}(\mathcal{N}'_v(u_{\infty},u))} 
\lesssim\mathcal{I}^2_0+\frac{\mathcal{O}^4}{a^{\frac{1}{2}}}.
\end{align*}

\end{proof}

\begin{proposition}
\label{EnergyEstimateTizeta45}
For $0\leq k\leq10$,  one has that 
\begin{align*}
||(a^{\frac{1}{2}})^{k}\mathcal{D}^k\Tizeta_4||^2_{L^2_{sc}(\mathcal{N}_u(0,v))}
+||(a^{\frac{1}{2}})^{k}\mathcal{D}^k\Tizeta_5||^2_{L^2_{sc}(\mathcal{N}'_v(u_{\infty},u))} 
\leq\mathcal{I}^2_0+\frac{1}{a^{\frac{1}{4}}}.
\end{align*}
\end{proposition}

\begin{proof}
For pair $(\Tizeta_4,\Tizeta_5)$ we make use of the following:

\begin{align*}
\mthorn' \,\tilde\zeta_{4}-   \meth\,\tilde\zeta_{5}
&= -\,\mathrm{i}\,\bar\zeta_{2}\,\phi_{1}^{2}
   + 2\,\mathrm{i}\,\phi_{1}\,\bar\phi_{1}\,\tilde\zeta_{4}
   - \mathrm{i}\,\phi_{0}\,\phi_{1}\,\bar{\tilde\zeta}_{5}
   - \zeta_{2}\,\bar\lambda
   - \phi_{0}\,\lambda\,\bar\lambda
 - \phi_{1}\,\mu\,\tau \\
&\quad
- \frac{\ulomega\,\tilde\zeta_{4}}{2}
 - 3\,\tilde\zeta_{4}\,\mu
   + \frac{\tilde\zeta_{5}\,\tau}{2}
   - 2\,\phi_{1}\,\bigl(\meth\mu\bigr)\,, \\
\mthorn \,\tilde\zeta_{5}-  \meth'\,\tilde\zeta_{4}
&= 2\Psi_{2}\,\phi_{1}
   + \mathrm{i}\,\bar\zeta_{1}\,\phi_{1}^{2}
   + \mathrm{i}\,\zeta_{2}\,\phi_{0}\,\bar\phi_{1}
   - \mathrm{i}\,\zeta_{1}\,\phi_{1}\,\bar\phi_{1}
   - \mathrm{i}\,\bar\phi_{0}\,\phi_{1}\,\tilde\zeta_{4}
   - \mathrm{i}\,\phi_{0}\,\phi_{1}\,\bar{\tilde\zeta}_{4}
   + \mathrm{i}\,\phi_{0}\,\bar\phi_{0}\,\tilde\zeta_{5}
    \\
&\quad
- \zeta_{3}\,\lambda
   + 2\,\tilde\zeta_{4}\,\pi
   + \zeta_{2}\,\bar\pi
   + 2\,\phi_{1}\,\pi\,\bar\pi
   + \tilde\zeta_{5}\,\rho
   + \phi_{1}\,\mu\,\rho
   + 2\,\phi_{1}\,\lambda\,\sigma
   + \frac{\tilde\zeta_{4}\,\bar\tau}{2}
   - \phi_{0}\,\mu\,\bar\tau
    \\
&\quad
+ 2\,\phi_{1}\,\bigl(\meth\pi\bigr)
   - \phi_{0}\,\bigl(\meth'\mu\bigr)\,. \nonumber
\end{align*}

Applying $\meth^k$ and commuting with $\mthorn$ and $\mthorn'$ we have that
\begin{align*}
&\mthorn'\meth^k\Tizeta_4+(k+3)\mu\meth^k\zeta_1-\meth^{k+1}\Tizeta_5=P_k, \\
&\mthorn\meth^k\Tizeta_5-\meth'\meth^k\Tizeta_4=Q_k,
\end{align*}
where
\begin{align*}
P_k=&\sum_{i_1+i_2+i_3=k}\meth^{i_1}\phi_j\meth^{i_1}\phi_l\meth^{i_1}\zeta_m
+\sum_{i_1+i_2=k}\meth^{i_1}\Gamma\meth^{i_1}\zeta_m
+\sum_{i_1+i_2=k}\meth^{i_1}\phi_1\meth^{i_2+1}\Timu,
\end{align*}
\begin{align*}
Q_k=&\sum_{i_1+...+i_5=k}\meth^{i_1}\Gamma^{i_2}\meth^{i_3}\phi_j\meth^{i_4}\phi_l\meth^{i_5}\zeta_m
+\sum_{i_1+...+i_4=k}\meth^{i_1}\Gamma^{i_2}\meth^{i_3}\Gamma\meth^{i_4}\zeta_m \\
&+\sum_{i_1+...+i_4=k}\meth^{i_1}\Gamma^{i_2}\meth^{i_3}\phi_l\meth^{i_4+1}\Gamma(\pi,\Timu)
+\sum_{i_1+...+i_5=k}\meth^{i_1}\Gamma^{i_2}\meth^{i_3}\phi_l\meth^{i_4}\Gamma\meth^{i_5}\Gamma\\
&+\sum_{i_1+...+i_4=k}\meth^{i_1}\Gamma^{i_2}\meth^{i_3}\phi_1\meth^{i_4}\Psi_2
+\sum_{i_1+i_2=k-1}\meth^{i_1}K\meth^{i_2}\Tizeta_{4}.
\end{align*}

Then apply the energy inequality one has
\begin{align*}
&||(a^{\frac{1}{2}})^{k}\meth^k\Tizeta_4||^2_{L^2_{sc}(\mathcal{N}_u(0,v))}
+||(a^{\frac{1}{2}})^{k}\meth^k\Tizeta_5||^2_{L^2_{sc}(\mathcal{N}'_v(u_{\infty},u))} 
\lesssim\mathcal{I}^2_0+(M+N).
\end{align*}
where 
\begin{align*}
M\equiv&2\int_0^v\int_{u_{\infty}}^u\frac{a}{|u'|^2}
||(a^{\frac{1}{2}})^{k}\meth^k\Tizeta_4||_{L^2_{sc}(\mathcal{S}_{u',v'})}
||(a^{\frac{1}{2}})^{k}P_k||_{L^2_{sc}(\mathcal{S}_{u',v'})} \\
\lesssim&\left(\int_0^v\int_{u_{\infty}}^u\frac{a}{|u'|^2}
||(a^{\frac{1}{2}})^{k}\meth^k\Tizeta_4||^2_{L^2_{sc}(\mathcal{S}_{u',v'})}\right)^{\frac{1}{2}}
\left(\int_0^v\int_{u_{\infty}}^u\frac{a}{|u'|^2}
||(a^{\frac{1}{2}})^{k}P_k||^2_{L^2_{sc}(\mathcal{S}_{u',v'})}\right)^{\frac{1}{2}} \\
\lesssim&\bm\zeta[\Tizeta_4]J^{\frac{1}{2}} ,
\end{align*}
and
\begin{align*}
N\equiv&2\int_0^v\int_{u_{\infty}}^u\frac{a}{|u'|^2}
||(a^{\frac{1}{2}})^{k}\meth^k\Tizeta_5||_{L^2_{sc}(\mathcal{S}_{u',v'})}
||(a^{\frac{1}{2}})^{k}Q_k||_{L^2_{sc}(\mathcal{S}_{u',v'})} \\
\lesssim&\left(\int_0^v\int_{u_{\infty}}^u\frac{a}{|u'|^2}
||(a^{\frac{1}{2}})^{k}\meth^k\Tizeta_5||^2_{L^2_{sc}(\mathcal{S}_{u',v'})}\right)^{\frac{1}{2}}
\left(\int_0^v\int_{u_{\infty}}^u\frac{a}{|u'|^2}
||(a^{\frac{1}{2}})^{k}Q_k||^2_{L^2_{sc}(\mathcal{S}_{u',v'})}\right)^{\frac{1}{2}} \\
\lesssim&\underline{\bm\zeta}[\Tizeta_5]H^{\frac{1}{2}} .
\end{align*}
Here $0\leq k\leq10$ and we define
\begin{align*}
J\equiv\int_0^v\int_{u_{\infty}}^u\frac{a}{|u'|^2}||a^{\frac{k}{2}}P_k||^2_{L^2_{sc}(\mathcal{S}_{u',v'})}, \quad
H\equiv\int_0^v\int_{u_{\infty}}^u\frac{a}{|u'|^2}||a^{\frac{k}{2}}Q_k||^2_{L^2_{sc}(\mathcal{S}_{u',v'})}.
\end{align*}

We estimate $J$ and $H$ together and mainly focus on the top derivative terms.
For terms $\zeta_m\phi_l^2$ we have
\begin{align*}
I_1\equiv&\sum_{i_1+...+i_5=k}\int_0^v\int_{u_{\infty}}^u\frac{a}{|u'|^2}
||a^{\frac{k}{2}}\meth^{i_1}\Gamma^{i_2}\meth^{i_3}\phi_j\meth^{i_4}\phi_l
\meth^{i_5}\zeta_m||^2_{L^2_{sc}(\mathcal{S}_{u',v'})} \\
\lesssim&\int_0^v\int_{u_{\infty}}^u\frac{a}{|u'|^2}\frac{1}{|u'|^4}
a^2\phi[\phi_{0,1}]^4||a^{\frac{k}{2}}\mathcal{D}^{k}\Tizeta_{4,5}||^2_{L^2_{sc}(\mathcal{S}_{u',v'})} 
+\int_0^v\int_{u_{\infty}}^u\frac{a}{|u'|^2}\frac{1}{|u'|^4}a^2\mathcal{O}^6 \\
\lesssim&\frac{a^2}{|u|^4}\phi[\phi_{0,1}]^4\underline{\bm\zeta}[\zeta_5]^2
+\frac{a^3}{|u|^5}\phi[\phi_{0,1}]^4\bm\zeta[\zeta_4]^2
+\frac{a^3}{|u|^5}\mathcal{O}^6.
\end{align*}

For term $\Gamma\zeta_l$ we have
\begin{align*}
I_2\equiv&\sum_{i_1+...+i_4=k}\int_0^v\int_{u_{\infty}}^u\frac{a}{|u'|^2}
||a^{\frac{k}{2}}\meth^{i_1}\Gamma^{i_2}\meth^{i_3}\Gamma
\meth^{i_4}\zeta_m||^2_{L^2_{sc}(\mathcal{S}_{u',v'})} \\
\lesssim&\int_0^v\int_{u_{\infty}}^u\frac{a}{|u'|^2}
||\Gamma(\lambda,...)a^{\frac{k}{2}}\mathcal{D}^k\zeta_{4}||^2_{L^2_{sc}(\mathcal{S}_{u',v'})}
+\int_0^v\int_{u_{\infty}}^u\frac{a}{|u'|^2}
||\Gamma(\sigma,...)a^{\frac{k}{2}}\mathcal{D}^k\zeta_{5}||^2_{L^2_{sc}(\mathcal{S}_{u',v'})} \\
&+\sum_{i_1+...+i_4=k,rest}\int_0^v\int_{u_{\infty}}^u\frac{a}{|u'|^2}
||a^{\frac{k}{2}}\mathcal{D}^{i_1}\Gamma^{i_2}\mathcal{D}^{i_3}\Gamma(\Timu,\lambda,\sigma,...)
\mathcal{D}^{i_4}\zeta_m||^2_{L^2_{sc}(\mathcal{S}_{u',v'})} \\
\lesssim&\frac{1}{|u|}\Gamma[\lambda]^2\bm\zeta[\Tizeta_4]^2
+\frac{a}{|u|^2}\Gamma[\sigma]^2\underline{\bm\zeta}[\Tizeta_5]^2
+\frac{1}{|u|}\mathcal{O}^4.
\end{align*}

For terms $\Gamma^2\phi_j$, we have
\begin{align*}
I_3\equiv&\sum_{i_1+...+i_5=k}\int_0^v\int_{u_{\infty}}^u\frac{a}{|u'|^2}
||a^{\frac{k}{2}}\meth^{i_1}\Gamma^{i_2}\meth^{i_3}\phi_j\meth^{i_4}\Gamma
\meth^{i_5}\Gamma||^2_{L^2_{sc}(\mathcal{S}_{u',v'})} \\
\lesssim&\int_0^v\int_{u_{\infty}}^u\frac{a}{|u'|^2}
||a^{\frac{k}{2}}\mu\mathcal{D}^i\phi_{0,1}\mathcal{D}^{k-i}
\Gamma(\tau,\rho)||^2_{L^2_{sc}(\mathcal{S}_{u',v'})}\\
&+\sum_{i_1+...+i_3=k}\int_0^v\int_{u_{\infty}}^u\frac{a}{|u'|^2}
||a^{\frac{k}{2}}\mathcal{D}^{i_1}\lambda\mathcal{D}^{i_2}\lambda
\mathcal{D}^{i_3}\phi_0||^2_{L^2_{sc}(\mathcal{S}_{u',v'})} \\
&+\sum_{i_1+...+i_3=k}\int_0^v\int_{u_{\infty}}^u\frac{a}{|u'|^2}
||a^{\frac{k}{2}}\mathcal{D}^{i_1}\lambda\mathcal{D}^{i_2}\sigma
\mathcal{D}^{i_3}\phi_0||^2_{L^2_{sc}(\mathcal{S}_{u',v'})}\\
&+\sum_{i_1+...+i_5=k,rest}\int_0^v\int_{u_{\infty}}^u\frac{a}{|u'|^2}
||a^{\frac{k}{2}}\meth^{i_1}\Gamma^{i_2}\meth^{i_3}\phi_j\meth^{i_4}\Gamma
\meth^{i_5}\Gamma||^2_{L^2_{sc}(\mathcal{S}_{u',v'})} ,
\end{align*}
and we obtain
\begin{align*}
I_3\lesssim&\int_0^v\int_{u_{\infty}}^u\frac{a}{|u'|^2}\frac{1}{|u'|^4}(
\frac{|u'|^4}{a^2}\Gamma[\mu]^2\Gamma[\tau,\rho]^2a\phi[\phi_{0,1}]^2
+\frac{|u'|^2}{a}\Gamma[\lambda]^2\frac{|u'|^2}{a}\Gamma[\lambda]^2a\phi[\phi_{0}]^2 \\
&+\frac{|u'|^2}{a}\Gamma[\lambda]^2a\Gamma[\sigma]^2a\phi[\phi_{0}]^2 
+\Gamma[\pi]^2\Gamma[\pi]^2a\phi[\phi_{1}]) \\
\lesssim&\frac{1}{|u|}\mathcal{O}^6+\frac{a^2}{|u|^3}\mathcal{O}^6.
\end{align*}

For terms $\phi_{0,1}\meth\Timu$ and $\phi_1\meth\pi$ we have
\begin{align*}
I_4\equiv&\sum_{i_1+...+i_4=k}\int_0^v\int_{u_{\infty}}^u\frac{a}{|u'|^2}
||a^{\frac{k}{2}}\meth^{i_1}\Gamma^{i_2}\meth^{i_3}\phi_j
\meth^{i_4+1}\Gamma||^2_{L^2_{sc}(\mathcal{S}_{u',v'})} \\
\lesssim&\int_0^v\int_{u_{\infty}}^u\frac{a}{|u'|^2}\frac{1}{|u'|^2}
a\phi[\phi_{0,1}]^2||a^{\frac{k}{2}}
\mathcal{D}^{k+1}\Gamma(\Timu,\pi)||^2_{L^2_{sc}(\mathcal{S}_{u',v'})} \\
&+\frac{1}{a}\int_0^v\int_{u_{\infty}}^u\frac{a}{|u'|^2}\frac{1}{|u'|^2}
a\phi[\phi_{0,1}]^2\frac{|u'|^2}{a^2}\mathcal{O}^2 \\
\lesssim&\frac{1}{a}\phi[\phi_{0,1}]^2\underline{\bm\Gamma}_{11}[\Timu]
+\frac{1}{|u|}\phi[\phi_{0,1}]^2\bm\Gamma_{11}[\pi]
+\frac{1}{a^2}\mathcal{O}^4.
\end{align*}

For term $\phi_1\Psi_2$ we have
\begin{align*}
I_5\equiv&\sum_{i_1+...+i_4=k}\int_0^v\int_{u_{\infty}}^u\frac{a}{|u'|^2}
||a^{\frac{k}{2}}\meth^{i_1}\Gamma^{i_2}\meth^{i_3}\phi_1
\meth^{i_4}\Psi_2||^2_{L^2_{sc}(\mathcal{S}_{u',v'})} \\
\lesssim&\int_0^v\int_{u_{\infty}}^u\frac{a}{|u'|^2}\frac{1}{|u|^2}
(a\phi[\phi_{1}]^2||a^{\frac{k}{2}}\mathcal{D}^k\Psi_2||^2_{L^2_{sc}(\mathcal{S}_{u',v'})} 
+a\mathcal{O}^4) \\
\lesssim&\frac{a}{|u|^2}\phi[\phi_{1}]^2\underline{\bm\Psi}[\Psi_2]^2
+\frac{a^2}{|u|^3}\mathcal{O}^4.
\end{align*}

For the term contain the Gaussian curvature we have

\begin{align*}
\sum_{i=0}^{k-1}a\int_0^v\int_{u_{\infty}}^u\frac{a}{|u'|^2}
||a^{\frac{k-1}{2}}\meth^{i}K\meth^{k-1-i}\Tizeta_{4}||^2_{L^2_{sc}(\mathcal{S}_{u',v'})}\leq
\frac{1}{|u|}\zeta[\Tizeta_{4}]^2
||(a^{\frac{1}{2}}\mathcal{D})^{k-1}K||^2_{L^2_{sc}(\mathcal{S}_{u,v})}
\end{align*}

Collect the results we obtain
\begin{align*}
J+H\lesssim\frac{\mathcal{O}^6}{a},
\end{align*}
and conclude that
\begin{align*}
&||(a^{\frac{1}{2}})^{k}\meth^k\Tizeta_4||^2_{L^2_{sc}(\mathcal{N}_u(0,v))}
+||(a^{\frac{1}{2}})^{k}\meth^k\Tizeta_5||^2_{L^2_{sc}(\mathcal{N}'_v(u_{\infty},u))} 
\lesssim\mathcal{I}^2_0+\frac{\mathcal{O}^4}{a^{\frac{1}{2}}}.
\end{align*}

\end{proof}

\subsection{Energy estimate for curvature}

\begin{proposition}
\label{EnergyEstimateCurvature}
\begin{align*}
\bm\Psi^2+\underline{\bm\Psi}^2\lesssim\mathcal{I}^2_0+1.
\end{align*}
\end{proposition}

\begin{proof}
For the pair $(\Psi_{L},\Psi_{R})$ where $\Psi_L\in\{\Psi_0,\TiPsi_1,\Psi_2,\TiPsi_3\}$ 
and $\Psi_R\in\{\TiPsi_1,\Psi_2,\TiPsi_3,\Psi_4\}$ we make use of 

\begin{align*}
\mthorn'\Psi_L-\meth\Psi_R=&\phi_j\meth\zeta_l+\Gamma\Psi_l+\Psi_l\phi_j^2
+\zeta_j^2+\zeta_j\phi_l^3+\zeta_j\phi_l\Gamma+\phi_j^2\Gamma^2, \\
\mthorn\Psi_R-\meth'\Psi_L=&\phi_j\meth\zeta_l+\Gamma\Psi_l+\Psi_l\phi_j^2
+\zeta_j^2+\zeta_j\phi_l^3+\zeta_j\phi_l\Gamma+\phi_j^2\Gamma^2.
\end{align*}

Applying $\meth^k$ and commuting with $\mthorn$ and $\mthorn'$ we have that
\begin{align*}
&\mthorn'\meth^k\Psi_L+\left((2s_2(\Psi_L)+k+1\right)\mu\meth^k\Psi_R-\meth^{k+1}\TiPsi_1=P_k, \\
&\mthorn\meth^k\Psi_R-\meth'\meth^k\Psi_L=Q_k,
\end{align*}
where $P_k$ and $Q_k$ have the following form
\begin{align*}
&\sum_{i_1+i_2+i_3+i_4=k}\meth^{i_1}\Gamma(\tau,\pi)^{i_2}
\meth^{i_3}\phi_j\meth^{i_4+1}\zeta_{l}
+\sum_{i_1+i_2+i_3+i_4=k}\meth^{i_1}\Gamma^{i_2}\meth^{i_3}\Gamma\meth^{i_4}\Psi_l\\
&+\sum_{i_1+...+i_5=k}\meth^{i_1}\Gamma^{i_2}
\meth^{i_3}\phi_{j_1}\meth^{i_4}\phi_{j_2}\meth^{i_5}\Psi_l 
+\sum_{i_1+...+i_4=k}\meth^{i_1}\Gamma^{i_2}\meth^{i_3}\zeta_l\meth^{i_4}\zeta_j \\
&+\sum_{i_1+..+i_6=k}\meth^{i_1}\Gamma^{i_2}\meth^{i_3}\phi_{j_1}\meth^{i_4}\phi_{j_2}
\meth^{i_5}\phi_{j_3}\meth^{i_6}\zeta_l 
+\sum_{i_1+..+i_5=k}\meth^{i_1}\Gamma^{i_2}\meth^{i_3}\phi_{j}\meth^{i_4}\zeta_l
\meth^{i_5}\Gamma \\
&+\sum_{i_1+..+i_6=k}\meth^{i_1}\Gamma^{i_2}\meth^{i_3}\phi_{j_1}\meth^{i_4}\phi_{j_2}
\meth^{i_5}\Gamma\meth^{i_6}\Gamma
+\sum_{i_1+i_2=k-1}\meth^{i_1}K\meth^{i_2}\Psi_{L}.
\end{align*}

The make use of the energy inequality one has
\begin{align*}
&||(a^{\frac{1}{2}}\meth)^k\Psi_L||^2_{L^2_{sc}(\mathcal{N}_u(0,v))}
+||(a^{\frac{1}{2}}\meth)^k\Psi_R||^2_{L^2_{sc}(\mathcal{N}'_v(u_{\infty},u))} \\
\lesssim&||(a^{\frac{1}{2}}\meth)^k\Psi_L||^2_{L^2_{sc}(\mathcal{N}_{u_{\infty}}(0,v))}
+||(a^{\frac{1}{2}}\meth)^k\Psi_R||^2_{L^2_{sc}(\mathcal{N}'_0(u_{\infty},u))} 
+M+N,
\end{align*}
where
\begin{align*}
M\equiv&2\int_0^v\int_{u_{\infty}}^u\frac{a}{|u'|^2}
||(a^{\frac{1}{2}}\meth)^k\Psi_0||_{L^2_{sc}(\mathcal{S}_{u',v'})}
||(a^{\frac{1}{2}})^{k}P_k||_{L^2_{sc}(\mathcal{S}_{u',v'})},
\end{align*}
and
\begin{align*}
N\equiv&2\int_0^v\int_{u_{\infty}}^u\frac{a}{|u'|^2}
||(a^{\frac{1}{2}}\meth)^k\TiPsi_1||_{L^2_{sc}(\mathcal{S}_{u',v'})}
||(a^{\frac{1}{2}})^{k}Q_k||_{L^2_{sc}(\mathcal{S}_{u',v'})} .
\end{align*}

For pair $(\Psi_0,\TiPsi_1)$, one need multiply extra $\frac{1}{a}$
\begin{align*}
&\frac{1}{a}||(a^{\frac{1}{2}}\meth)^k\Psi_0||^2_{L^2_{sc}(\mathcal{N}_u(0,v))}
+\frac{1}{a}||(a^{\frac{1}{2}}\meth)^k\TiPsi_1||^2_{L^2_{sc}(\mathcal{N}'_v(u_{\infty},u))} \\
\lesssim&\frac{1}{a}||(a^{\frac{1}{2}}\meth)^k\Psi_0||^2_{L^2_{sc}(\mathcal{N}_{u_{\infty}}(0,v))}
+\frac{1}{a}||(a^{\frac{1}{2}}\meth)^k\TiPsi_1||^2_{L^2_{sc}(\mathcal{N}'_0(u_{\infty},u))} 
+\frac{1}{a}(M+N),
\end{align*}
and we have
\begin{align*}
M\leq&\frac{a}{|u|^{\frac{1}{2}}}\bm\Psi[\Psi_0]\left(\int_0^v\int_{u_{\infty}}^u\frac{a}{|u'|^2}
||(a^{\frac{1}{2}})^{k}P_k||^2_{L^2_{sc}(\mathcal{S}_{u',v'})}\right)^{\frac{1}{2}}
=\frac{a}{|u|^{\frac{1}{2}}}\bm\Psi[\Psi_0]J^{\frac{1}{2}} ,
\end{align*}
and
\begin{align*}
N\leq&a^{\frac{1}{2}}\underline{\bm\Psi}[\TiPsi_1]\left(\int_0^v\int_{u_{\infty}}^u\frac{a}{|u'|^2}
||(a^{\frac{1}{2}})^{k}Q_k||^2_{L^2_{sc}(\mathcal{S}_{u',v'})}\right)^{\frac{1}{2}}
=a^{\frac{1}{2}}\underline{\bm\Psi}[\TiPsi_1]H^{\frac{1}{2}}.
\end{align*}

For the rest pairs, i.e. $(\Psi'_L,\Psi'_R)$ which denote pair $(\TiPsi_1,\Psi_2)$, $(\Psi_2,\TiPsi_3)$ 
and $(\TiPsi_3,\Psi_4)$ one has
\begin{align*}
&||(a^{\frac{1}{2}}\meth)^k\Psi'_L||^2_{L^2_{sc}(\mathcal{N}_u(0,v))}
+||(a^{\frac{1}{2}}\meth)^k\Psi'_R||^2_{L^2_{sc}(\mathcal{N}'_v(u_{\infty},u))} \\
\lesssim&\int_{u_{\infty}}^u\frac{a}{|u'|^2}
||(a^{\frac{1}{2}}\meth)^k\Psi'_{L}||^2_{L^2_{sc}(\mathcal{N}_{u'}(0,v))}
+\int_0^v||(a^{\frac{1}{2}}\meth)^k\Psi'_{R}||^2_{L^2_{sc}(\mathcal{N}_{v'}(u_{\infty},u))}+J+H
+\mathcal{I}^2_0.
\end{align*}

We estimate $J$ and $H$ together. 

For terms $\Gamma\Psi$ , 
because the general properties of connections and curvatures are the same 
to that shown in, hence one can make use of the results in the paper directly and obtain the control by
$\frac{\mathcal{O}^4}{a}+1$.

For term $K\Psi_L$ one has
\begin{align*}
&\sum_{i=0}^{k-1}a\int_0^v\int_{u_{\infty}}^u\frac{a}{|u'|^2}
||a^{\frac{k-1}{2}}\meth^{i}K\meth^{k-1-i}\Psi_L||^2_{L^2_{sc}(\mathcal{S}_{u',v'})} \\
\lesssim&(\frac{a^2}{|u|^3}\Psi[\TiPsi_1,\Psi_2,\Psi_3]^2+\frac{a^3}{|u|^3}\Psi[\Psi_0]^2)
||(a^{\frac{1}{2}}\mathcal{D})^{k-1}K||^2_{L^2_{sc}(\mathcal{N}_{u}(0,v))}.
\end{align*}

For terms $\phi_j\meth\zeta_l$, one has
\begin{align*}
&\sum_{i_1+...+i_4=k}\int_0^v\int_{u_{\infty}}^u\frac{a}{|u'|^2}
||\meth^{i_1}\Gamma^{i_2}\meth^{i_3}\phi_j\meth^{i_4+1}\zeta_l||^2_{L^2_{sc}(\mathcal{S}_{u,v'})} \\
\leq&\frac{1}{a}\int_0^v\int_{u_{\infty}}^u\frac{a}{|u'|^2}\frac{1}{|u'|^2}a\phi[\phi_j]^2
\frac{|u'|^2}{a}\zeta[\zeta_{4,5}]^2
\lesssim1+\frac{1}{a}\phi[\phi_j]^2\zeta[\zeta_{4,5}]^2.
\end{align*}
Details can be found in \ref{phiethzeta}.

For terms $\phi_j^2\Psi_l$ one has
\begin{align*}
&\sum_{i_1+...+i_5=k}\int_0^v\int_{u_{\infty}}^u\frac{a}{|u'|^2}
||\meth^{i_1}\Gamma^{i_2}\meth^{i_3}\phi_{j_1}\meth^{i_4}\phi_{j_2}
\meth^{i_5}\Psi_l||^2_{L^2_{sc}(\mathcal{S}_{u,v'})} \\
\leq&\int_{u_{\infty}}^u\frac{a}{|u'|^2}\frac{a^2}{|u'|^4}\phi[\phi_j]^4a\Psi[\Psi_0]^2
\leq\frac{a^4}{|u|^5}\mathcal{O}^6.
\end{align*}
Details can be found in \ref{phi2Psi}. 

Next, appendix \ref{zeta2} shows that 
terms $\zeta_l^2$ can be bounded by $\frac{a^2}{|u|^3}\zeta[\zeta_j]^4$.
Terms $\zeta_j\phi_l^3$ can be bounded by $1/a$ which are shown in \ref{zetaphi3}.

For terms $\zeta_j\phi_l\Gamma$, one has
\begin{align*}
&\sum_{i_1+...+i_5=k}\int_0^v\int_{u_{\infty}}^u\frac{a}{|u'|^2}
||\meth^{i_1}\Gamma^{i_2}\meth^{i_3}\zeta_{l}
\meth^{i_4}\phi_{l}\meth^{i_5}\Gamma||^2_{L^2_{sc}(\mathcal{S}_{u,v'})} 
\leq\frac{a}{|u|}\zeta[\zeta_0]^2\phi[\phi_j]^2\Gamma[\mu]^2+\frac{1}{a}\mathcal{O}^6.
\end{align*}
See \ref{zetaphiGamma}.

For terms $\phi_j^2\mu\Gamma$, one has
\begin{align*}
&\sum_{i_1+...+i_6=k}\int_0^v\int_{u_{\infty}}^u\frac{a}{|u'|^2}
||\meth^{i_1}\Gamma^{i_2}\meth^{i_3}\phi_{l_1}\meth^{i_4}\phi_{l_2}
\meth^{i_5}\mu\meth^{i_6}\Gamma||^2_{L^2_{sc}(\mathcal{S}_{u,v'})}
\lesssim\frac{1}{|u|}\mathcal{O}^8.
\end{align*}
See \ref{phi2muGamma}.

Now collect the results above one can conclude that
\begin{align*}
\bm\Psi^2+\underline{\bm\Psi}^2\lesssim\mathcal{I}^2_0+1.
\end{align*}

With the results above one can estimate the top derivative of Gaussian curvature $K=\mathcal{K}+\bar{\mathcal{K}}$
where 
\begin{align*}
\mathcal{K}=\mathrm{i}\bar\zeta_1\phi_1-\mathrm{i}\zeta_1\bar\phi_1
+\mathrm{i}\bar\Tizeta_4\phi_0-\mathrm{i}\Tizeta_4\bar\phi_0
-\Psi_2+\mu\rho-\lambda\sigma.
\end{align*}
We have
\begin{align*}
||(a^{\frac{1}{2}}\mathcal{D})^k\mathcal{K}||_{L^2_{sc}(\mathcal{S}_{u,v'})}\lesssim&
\sum_{i=0}^k||a^{\frac{1}{2}}\mathcal{D}^i\zeta_1\mathcal{D}^{k-i}\phi_1||_{L^2_{sc}(\mathcal{S}_{u,v'})}
+\sum_{i=0}^k||a^{\frac{1}{2}}\mathcal{D}^i\Tizeta_4\mathcal{D}^{k-i}\phi_0||_{L^2_{sc}(\mathcal{S}_{u,v'})}\\
&+\sum_{i=0}^k||a^{\frac{1}{2}}\mathcal{D}^i\mu\mathcal{D}^{k-i}\rho||_{L^2_{sc}(\mathcal{S}_{u,v'})}
+\sum_{i=0}^k||a^{\frac{1}{2}}\mathcal{D}^i\lambda\mathcal{D}^{k-i}\sigma||_{L^2_{sc}(\mathcal{S}_{u,v'})} \\
&+||(a^{\frac{1}{2}}\mathcal{D})^k\Psi_2||_{L^2_{sc}(\mathcal{S}_{u,v'})} \\
\lesssim&\frac{1}{|u|}\zeta[\zeta_1]a^{\frac{1}{2}}\phi[\phi_1]
+\frac{1}{|u|}\zeta[\Tizeta_4]a^{\frac{1}{2}}\phi[\phi_0]
+\frac{1}{|u|}a^{\frac{1}{2}}\phi[\phi_1]||(a^{\frac{1}{2}}\mathcal{D})^k\Tizeta_4||_{L^2_{sc}(\mathcal{S}_{u,v'})}\\
&+\frac{1}{|u|}\frac{|u|^2}{a}\Gamma[\mu]\frac{a}{|u|}\Gamma[\rho]
+\frac{1}{|u|}\frac{|u|}{a}\Gamma[\Timu]\Gamma[\rho]
+\frac{1}{|u|}\frac{|u|}{a^{\frac{1}{2}}}\Gamma[\lambda]a^{\frac{1}{2}}\Gamma[\sigma]\\
&+||(a^{\frac{1}{2}}\mathcal{D})^k\Psi_2||_{L^2_{sc}(\mathcal{S}_{u,v'})} \\
\lesssim&||(a^{\frac{1}{2}}\mathcal{D})^k\Psi_2||_{L^2_{sc}(\mathcal{S}_{u,v'})}
+\frac{a^{\frac{1}{2}}}{|u|}||(a^{\frac{1}{2}}\mathcal{D})^k\Tizeta_4||_{L^2_{sc}(\mathcal{S}_{u,v'})}+1. 
\end{align*}
Then, integrating along the light cone one has
\begin{align*}
||(a^{\frac{1}{2}}\mathcal{D})^k\mathcal{K}||_{L^2_{sc}(\mathcal{N}_{u}(0,v))}\lesssim&
||(a^{\frac{1}{2}}\mathcal{D})^k\Psi_2||_{L^2_{sc}(\mathcal{N}_{u}(0,v))}+1
\lesssim1, \\
||(a^{\frac{1}{2}}\mathcal{D})^k\mathcal{K}||_{L^2_{sc}(\mathcal{N}_{v}(u_{\infty},u))}\lesssim&
||(a^{\frac{1}{2}}\mathcal{D})^k\Psi_2||_{L^2_{sc}(\mathcal{N}_{v}(u_{\infty},u))}+1
\lesssim1.
\end{align*}
Hence one obtains the control of the Gaussian curvature.

\end{proof}

Then we obtain the main existence result:

\begin{theorem}
\label{ExistenceResult}
\textbf{\em (Existence result)}
Given $\mathcal{I}$, there exists a sufficiently large
$a_0=a_0(\mathcal{I})$, such that for $0\leq a_0 \leq a$ and initial
data 
\begin{align*}
\mathcal{I}_0\equiv\sum_{j=0}^1\sum_{i=0}^{15}\frac{1}{a^{\frac{1}{2}}}
||\mthorn^j(|u_{\infty}|\mathcal{D})^i(\sigma,\phi_0)||_{L^{2}(\mathcal{S}_{u_{\infty},v})}\leq\mathcal{I},
\end{align*}
along the outgoing initial null hypersurface $u=u_{\infty}$ and 
Minkowskian initial data along the ingoing initial null hypersurface $v=0$, 
then the Einstein-Weyl spinor system has a solution in the causal diamond
\begin{align*}
\mathbb{D}=\big\{(u,v)|u_{\infty}\leq u\leq -a/4, \ \ 0\leq v\leq 1\big\}.
\end{align*}
\end{theorem}

\section{Trapped surface formation}
\label{TrappedSurface}
In this section, we prove the following 
\begin{theorem}
[\textbf{\em Trapped surface formation}]
Given $\mathcal{I}$, there exists a sufficiently large
$a=a(\mathcal{I})$ such that if $0\leq a_0 \leq a$, and the initial
data satisfies
\begin{align*}
\mathcal{I}_0\equiv\sum_{j=0}^1\sum_{i=0}^{15}\frac{1}{a^{\frac{1}{2}}}
||\mthorn^j(|u_{\infty}|\mathcal{D})^i(\sigma,\phi_0)||_{L^{2}(\mathcal{S}_{u_{\infty},v})}\leq\mathcal{I}
\end{align*}
along the outgoing initial null hypersurface $u=u_{\infty}$ and is 
Minkowskian initial data along ingoing initial null hypersurface $v=0$, 
with 
\begin{align*}
\int_0^1\left( |u_{\infty}|^2(\sigma\bar\sigma+2\mathrm{i}\phi_0\mthorn\bar\phi_0
-2\mathrm{i}\bar\phi_0\mthorn\phi_0)|\right)(u_{\infty},v') \mathrm{d}v'\geq a,
\end{align*}
uniformly for any point on the outgoing initial null hypersurface
$u=u_{\infty}$, then
we have that $\mathcal{S}_{-a/4,1}$ is a trapped surface.
\end{theorem}

\begin{proof}
We start with the frame coefficient. 
Since the properties are the same as in the Einstein–Scalar case, 
we state the results directly and refer to \cite{An2022,HilValZha23}.
\begin{align*}
|Q|-1\lesssim\frac{\mathcal{O}}{|u|}, \quad
|P^{\mathcal{A}}(u,v)|\lesssim|P^{\mathcal{A}}(u,0)|+\frac{a^{\frac{1}{2}}}{|u|}\mathcal{O}, \quad
|C^{\mathcal{A}}|\lesssim\frac{a^{\frac{1}{2}}}{|u|^2}\mathcal{O}.
\end{align*}
for a point $(x^2,x^3)\in\mathcal{S}$ in a coordinate area. 
From the equation \eqref{Thornrho}, 
\begin{align*}
\mthorn\,\rho &= 2\,\mathrm{i}\,\bar\zeta_0\,\phi_0
- 2\,\mathrm{i}\,\zeta_0\,\bar\phi_0
+ \rho^2
+ \sigma\,\bar\sigma,
\end{align*}
one defines 
\begin{align*}
F\equiv |u|^2Q^{-1}(\sigma\bar\sigma+2\mathrm{i}\bar\zeta_0\phi_0-2\mathrm{i}\zeta_0\bar\phi_0).
\end{align*}
Then from $\bmn=\partial_u+C^{\mathcal{A}}\partial_{\mathcal{A}}$, one has
\begin{align*}
\partial_uF=\mthorn'F-C^{\mathcal{A}}\partial_{\mathcal{A}}F\equiv I_1+I_2.
\end{align*}
From equation $\mthorn'\sigma$, $\mthorn'\phi_0$, $\mthorn'\zeta_0$ as well as $\mthorn'Q$, one can 
compute that 
\begin{align*}
\mthorn'F=|u|^2Q^{-1}M
\end{align*}
where the expression of M is given in terms of \ref{ExpressionofM}. 
Using the existence theorem \ref{ExistenceResult} we calculate and obtain
\begin{align*}
|M|\lesssim\frac{a}{|u|^4},
\end{align*}
and hence
\begin{align*}
|I_1|\lesssim\frac{a}{|u|^2}.
\end{align*}

For the analysis of $I_2$, because the property of $\phi_0$ and $\zeta_0$ is the same with that of $\varphi_0$, 
hence we obtain
\begin{align*}
|I_2|\leq\frac{a^{\frac{3}{2}}}{|u|^2}
\end{align*}
which leads to 
\[
  |\partial_uF|\leq\frac{a^{\frac{3}{2}}}{|u|^2}\ll\frac{a^{\frac{7}{4}}}{|u|^2}.
\]

Then assume that the initial data satisfy
\begin{align*}
\int_0^1\left( |u_{\infty}|^2(\sigma\bar\sigma+2\mathrm{i}\bar\zeta_0\phi_0
-2\mathrm{i}\zeta_0\bar\phi_0)|\right)(u_{\infty},v') \mathrm{d}v'\geq a.
\end{align*}
Using the above analysis of $F$, for sufficiently large a, we have 
\begin{align*}
  & \int_0^1\left( |u|^2Q^{-1}(\sigma\bar\sigma+2\mathrm{i}\bar\zeta_0\phi_0
-2\mathrm{i}\zeta_0\bar\phi_0)|\right)(u,v',x^2,x^3) \mathrm{d}v'\\
\geq&\int_0^1\left( |u_{\infty}|^2(\sigma\bar\sigma+2\mathrm{i}\bar\zeta_0\phi_0
-2\mathrm{i}\zeta_0\bar\phi_0)|\right)(u_{\infty},v',x^2,x^3) \mathrm{d}v'-\frac{a^{\frac{7}{4}}}{|u|} \\
\geq&a-\frac{a^{\frac{7}{4}}}{|u|}
\geq a-\frac{4a^{\frac{7}{4}}}{a}
\geq \frac{a}{2}.
\end{align*}
When $u=-\frac{a}{4}$,
\begin{align*}
\int_0^1\left(Q^{-1}(\sigma\bar\sigma+2\mathrm{i}\bar\zeta_0\phi_0
-2\mathrm{i}\zeta_0\bar\phi_0)\right)|_{(-\frac{a}{4},v',x^2,x^3)} \mathrm{d}v'\geq \frac{8}{a}.
\end{align*}

With the results above, now we move to the calculation of expansion.
For the outgoing expansion defined as 
$\bm\theta_l\equiv\boldsymbol{\sigmasl}^{ab}\nabla_al_b=-\rho-\bar\rho=-2\rho$, 
initially on $\mathcal{N}_{u_{\infty}}$, 
one makes use of $Q=1$, and the initial data for $\phi_0$, $\zeta_0$ 
and $\rho(u,0,x^2,x^3)=-\frac{1}{|u|}$ and has 
\begin{align*}
\rho(u_{\infty},v,x^2,x^3)
&=\rho(u_{\infty},0,x^2,x^3)+\int_0^1\frac{\partial\rho}{\partial v} \mathrm{d}v 
=-\frac{1}{|u_{\infty}|}+\int_0^v\mthorn\rho \mathrm{d}v' \\
&= -\frac{1}{|u_{\infty}|}+\int_0^v\left(\sigma\bar\sigma+2\,\mathrm{i}\,\bar\zeta_0\,\phi_0
- 2\,\mathrm{i}\,\zeta_0\,\bar\phi_0\right)|_{(-\frac{a}{4},v',x^2,x^3)} \mathrm{d}v'  \\
  &\lesssim-\frac{1}{|u_{\infty}|}+\frac{a}{|u_{\infty}|^2}<0,
\end{align*}
which means the outgoing expansion is positive along $\mathcal{N}_{u_{\infty}}$.
Then on $\mathcal{N}_{-a/4}$ we have that
\begin{align*}
\rho(-\frac{a}{4},1,x^2,x^3)=&\rho(-\frac{a}{4},0,x^2,x^3)+\int_0^1\frac{\partial\rho}{\partial v} \mathrm{d}v 
=-\frac{4}{a}+\int_0^1Q^{-1}\mthorn\rho \mathrm{d}v' \\
\geq& -\frac{4}{a}+\int_0^1Q^{-1}(\sigma\bar\sigma+2\,\mathrm{i}\,\bar\zeta_0\,\phi_0
- 2\,\mathrm{i}\,\zeta_0\,\bar\phi_0)|_{(-\frac{a}{4},v',x^2,x^3)} \mathrm{d}v'  \\
\geq&-\frac{4}{a}+\frac{8}{a}=\frac{4}{a}>0.
\end{align*}
That leads to $\bm\theta_l=-2\rho<0$, i.e. for any~$(x^2,x^3)$, 
the outgoing expansion on  $\mathcal{N}_{-a/4}$ is negative.
For the ingoing expansion~$\theta_{\bmn}$ which is defined by
$\theta_{\bmn}\equiv\boldsymbol{\sigmasl}^{ab}\nabla_an_b=\mu+\bar\mu=2\mu$,
we have the estimate 
\begin{align*}
||\Timu||_{L^{\infty}(\mathcal{S}_{u,v})}=
||\mu+\frac{1}{|u|}||_{L^{\infty}(\mathcal{S}_{u,v})}\leq\frac{1}{|u|^2},
\end{align*}
which leads to
\begin{align*}
\mu(-\frac{a}{4},v,x^2,x^3)<0, \quad
\mu(u_{\infty},v,x^2,x^3)<0
\end{align*}
for any~$(x^2,x^3)$. This means that the ingoing expansion is always negative.
Then we conclude that the 2-sphere $\mathcal{S}_{-\frac{a}{4},1}$ is a trapped surface. 
In summary, we have derived the dynamical process that governs the emergence of trapped surfaces 
in the Einstein–Weyl model, thereby providing a description of their formation from absence to existence.
\end{proof}

\begin{remark}[{Low-angular momentum argument}]
\label{low-angular momentum argument}
Make use of the existence result one has that 
\begin{align*}
|\phi_0|\lesssim\frac{a^{\frac{1}{2}}}{|u|}, \quad
|\phi_1|\lesssim\frac{a}{|u|^2}, \quad
|\zeta_0|\lesssim\frac{a^{\frac{1}{2}}}{|u|}, \quad
|\zeta_1,\zeta_3|\lesssim\frac{a^{\frac{1}{2}}}{|u|^2},
\end{align*}
then one can calculate that 
\begin{align*}
\Phi_{00} = 2\mathrm{i}\bigl(\bar\zeta_0\phi_0-\zeta_0\bar\phi_0\bigr)
\lesssim\frac{a}{|u|^2},\quad
\Phi_{01} = \mathrm{i}\bigl(2\bar\zeta_1\phi_0-\zeta_3\bar\phi_0-\zeta_0\bar\phi_1\bigr)
\lesssim\frac{a^{\frac{3}{2}}}{|u|^3},
\end{align*}
hence 
\begin{align*}
|\Phi_{01}|\sim\frac{1}{a^{\frac{1}{2}}}|\Phi_{00}|
\end{align*}
where $a$ is a universal large number. 
Together with the initial data assumption, we have shown that such relation keeps and 
physically, this means that compared with the energy density, 
angular momentum density is small physically. 
\end{remark}

\appendix

\section{T-weight equations}
\subsection{Definition of $\zeta_{ABA'}$}
\label{DefinitionEqzeta}
\begin{subequations}
\begin{align}
\mthorn\,\phi_{0}
&= \zeta_{0}, \\[6pt]
\mthorn'\,\phi_{1}
&= \zeta_{5}
  - \tfrac{\ulomega\,\phi_{1}}{2}, \\[6pt]
\meth\,\phi_{0}
&= \zeta_{3}
  - \phi_{1}\,\sigma
  + \tfrac{\phi_{0}\,\tau}{2}, \\[6pt]
\meth\,\phi_{1}
&= \zeta_{4}
  + \phi_{0}\,\mu
  - \tfrac{\phi_{1}\,\tau}{2}, \\[6pt]
\meth'\,\phi_{0}
&= \zeta_{1}
  - \phi_{1}\,\rho
  + \tfrac{\phi_{0}\,\bar\tau}{2}, \\[6pt]
\meth'\,\phi_{1}
&= \zeta_{2}
  + \phi_{0}\,\lambda
  - \tfrac{\phi_{1}\,\bar\tau}{2}, 
\end{align}
\end{subequations}

\subsection{Equations for $\zeta_{ABA'}$}
\label{Equationzeta}

\subsubsection{Transport equations of $\zeta_{ABA'}$ without curvature}
\label{EquationzetaNoCurv}

\begin{subequations}
\begin{align}
\mthorn' \zeta_0- \meth\,\zeta_1 &= 
 \mathrm{i}\,\bar{\zeta}_4\,\phi_0^2
- \mathrm{i}\,\zeta_4\,\phi_0\,\bar{\phi}_0
- \mathrm{i}\,\bar{\zeta}_1\,\phi_0\,\phi_1
+ \mathrm{i}\,\zeta_3\,\bar{\phi}_0\,\phi_1
- \mathrm{i}\,\zeta_1\,\phi_0\,\bar{\phi}_1 \nonumber\\
&\quad
+ \mathrm{i}\,\zeta_0\,\phi_1\,\bar{\phi}_1
- \zeta_0\,\mu+\frac{3\,\zeta_0\,\ulomega}{2}
+ \zeta_4\,\rho
+ \zeta_2\,\sigma
- \frac{5\,\zeta_1\,\tau}{2}
- \zeta_3\,\bar{\tau}, \label{thornprimezeta0}\\
\mthorn \zeta_1- \meth'\,\zeta_0 &= -\mathrm{i}\,\bar{\zeta}_3\,\phi_0^2
+ 2\,\mathrm{i}\,\zeta_1\,\phi_0\,\bar{\phi}_0
+ \mathrm{i}\,\bar\zeta_0\,\phi_0\,\phi_1
- 2\,\mathrm{i}\,\zeta_0\,\bar{\phi}_0\,\phi_1 
+ \zeta_0\,\pi
+ 2\,\zeta_1\,\rho \nonumber\\
&\quad
+ \zeta_3\,\bar\sigma
- \frac{3\,\zeta_0\,\bar{\tau}}{2} \label{thornzeta1} \\
\mthorn'\zeta_1- \meth\,\zeta_2 &= 
- \mathrm{i}\,\zeta_5\,\phi_0\,\bar{\phi}_0
+ \mathrm{i}\,\bar{\zeta}_4\,\phi_0\,\phi_1
+ \mathrm{i}\,\zeta_4\,\bar{\phi}_0\,\phi_1
- \mathrm{i}\,\bar{\zeta}_1\,\phi_1^2 
- \mathrm{i}\,\zeta_2\,\phi_0\,\bar{\phi}_1 \nonumber\\
&\quad
+ \mathrm{i}\,\zeta_1\,\phi_1\,\bar{\phi}_1
- 2\,\zeta_1\,\mu+\frac{\zeta_1\,\ulomega}{2}
+ \zeta_5\,\rho
- \frac{\zeta_2\,\tau}{2}
- \zeta_4\,\bar{\tau} \label{thornprimezeta1}\\
\mthorn\,\zeta_2- \meth'\,\zeta_1 &= 2\,\mathrm{i}\,\zeta_2\,\phi_0\,\bar{\phi}_0
- \mathrm{i}\,\bar{\zeta}_3\,\phi_0\,\phi_1
- 2\,\mathrm{i}\,\zeta_1\,\bar{\phi}_0\,\phi_1
+ \mathrm{i}\,\bar{\zeta}_0\,\phi_1^2
- \zeta_0\,\lambda \nonumber\\
&\quad
+ 2\,\zeta_1\,\pi
+ \zeta_2\,\rho
+ \zeta_4\,\bar{\sigma}
- \frac{\zeta_1\,\tau}{2} \label{thornzeta2} \\
\mthorn'\,\zeta_3 - \meth\,\zeta_4&= \frac{\zeta_3\,\ulomega}{2}
+ \mathrm{i}\,\bar{\zeta}_5\,\phi_0^2
- \mathrm{i}\,\bar{\zeta}_2\,\phi_0\,\phi_1
- 2\,\mathrm{i}\,\zeta_4\,\phi_0\,\bar{\phi}_1
+ 2\,\mathrm{i}\,\zeta_3\,\phi_1\,\bar{\phi}_1 \nonumber\\
&\quad
- \zeta_1\,\bar{\lambda}
- \zeta_3\,\mu
+ \zeta_5\,\sigma
- \frac{3\,\zeta_4\,\tau}{2} ,\label{thornprimezeta3}\\
\mthorn\,\zeta_4- \meth'\,\zeta_3 &= -\mathrm{i}\,\bar\zeta_4\,\phi_0^2
+ \mathrm{i}\,\zeta_4\,\phi_0\,\bar\phi_0
+ \mathrm{i}\,\bar\zeta_1\,\phi_0\,\phi_1
- \mathrm{i}\,\zeta_3\,\bar\phi_0\,\phi_1
+ \mathrm{i}\,\zeta_1\,\phi_0\,\bar\phi_1 \nonumber\\
&\quad
- \mathrm{i}\,\zeta_0\,\phi_1\,\bar\phi_1
- \zeta_0\,\mu
+ \zeta_3\,\pi
+ \zeta_1\,\bar\pi
+ 2\,\zeta_4\,\rho
- \frac{\zeta_3\,\bar\tau}{2} \label{thornzeta4}\\
\mthorn'\zeta_4- \meth\,\zeta_5 &=
 \mathrm{i}\,\bar{\zeta}_5\,\phi_0\,\phi_1
- \mathrm{i}\,\bar{\zeta}_2\,\phi_1^2
- 2\,\mathrm{i}\,\zeta_5\,\phi_0\,\bar{\phi}_1 
+ 2\,\mathrm{i}\,\zeta_4\,\phi_1\,\bar{\phi}_1
- \zeta_2\,\bar{\lambda}
- 2\,\zeta_4\,\mu \nonumber\\
&\quad-\frac{\zeta_4\,\ulomega}{2}
+ \frac{\zeta_5\,\tau}{2}, \label{thornprimezeta4}\\
\mthorn\,\zeta_5- \meth'\,\zeta_4 &= \mathrm{i}\,\zeta_5\,\phi_0\,\bar{\phi}_0
- \mathrm{i}\,\bar{\zeta}_4\,\phi_0\,\phi_1
- \mathrm{i}\,\zeta_4\,\bar{\phi}_0\,\phi_1
+ \mathrm{i}\,\bar{\zeta}_1\,\phi_1^2
+ \mathrm{i}\,\zeta_2\,\phi_0\,\bar{\phi}_1
- \mathrm{i}\,\zeta_1\,\phi_1\,\bar{\phi}_1 \nonumber\\ 
&\quad
- \zeta_3\,\lambda
- \zeta_1\,\mu
+ 2\,\zeta_4\,\pi
+ \zeta_2\,\bar{\pi}
+ \zeta_5\,\rho
+ \frac{\bar{\zeta}_4\,\bar{\tau}}{2}.\label{thornzeta5}
\end{align}
\end{subequations}

\subsubsection{Equations contain curvature}
\label{EquationzetaCurv}

\begin{subequations}
\begin{align}
\meth'\,\zeta_3 &= -\Psi_2\,\phi_0
+ \mathrm{i}\,\bar\zeta_4\,\phi_0^2
- \mathrm{i}\,\zeta_4\,\phi_0\,\bar\phi_0
+ \tilde\Psi_1\,\phi_1
+ \mathrm{i}\bar\zeta_1\,\phi_0\,\phi_1
- \mathrm{i}\,\zeta_1\,\phi_0\,\bar\phi_1 \nonumber\\
&\quad
- \zeta_4\,\rho
+ \zeta_2\,\sigma
- \frac{\zeta_1\,\tau}{2}
+ \frac{\zeta_3\,\bar\tau}{2}
+ \meth\,\zeta_1, \label{zeta3zeta1}\\[6pt]
\meth'\,\zeta_4 &= -\tilde\Psi_3\,\phi_0
+ \Psi_2\,\phi_1
- \mathrm{i}\,\bar\zeta_4\,\phi_0\,\phi_1
+ \mathrm{i}\,\zeta_4\,\bar\phi_0\,\phi_1
- \mathrm{i}\bar\zeta_1\,\phi_1^2
+ \mathrm{i}\,\zeta_1\,\phi_1\,\bar\phi_1\nonumber \\
&\quad
+ \zeta_3\,\lambda
- \zeta_1\,\mu
+ \frac{\zeta_2\,\tau}{2}
- \frac{\zeta_4\,\bar\tau}{2}
+ \meth\,\zeta_2 \,. \label{zeta4zeta2} \\
\mthorn'\,\zeta_2 &= -\frac{\zeta_2\,\omega}{2}
+ \Psi_4\,\phi_0
- \tilde\Psi_3\,\phi_1
+ \mathrm{i}\,\zeta_5\,\bar\phi_0\,\phi_1
- 2\,\mathrm{i}\,\bar\zeta_4\,\phi_1^2
+ \mathrm{i}\,\zeta_2\,\phi_1\,\bar\phi_1 \nonumber\\
&\quad
- 2\,\zeta_4\,\lambda
- \zeta_2\,\mu
+ \frac{\zeta_5\,\bar\tau}{2}
+ \meth'\,\zeta_5, \label{thornprimezeta2zeta5}\\
\mthorn\,\zeta_5 &= -\tilde\Psi_3\,\phi_0
+ \mathrm{i}\,\zeta_5\,\phi_0\,\bar\phi_0
+ \Psi_2\,\phi_1
- 2\,\mathrm{i}\,\bar\zeta_4\,\phi_0\,\phi_1
+ \mathrm{i}\,\zeta_2\,\phi_0\,\bar\phi_1 \nonumber\\
&\quad
- 2\zeta_1\,\mu
+ 2\,\zeta_4\,\pi
+ \zeta_2\,\bar\pi
+ \zeta_5\,\rho
+ \frac{\zeta_2\,\tau}{2}
+ \meth\,\zeta_2. \label{thornzeta5zeta2}
\end{align}
\end{subequations}

\subsubsection{Equations with renormalized $\Tizeta_4$, $\Tizeta_5$}

\begin{subequations}
\begin{align}
\meth'\zeta_{3}
=  &\meth\,\zeta_{1}-\Psi_{2}\,\phi_{0}
+ \TiPsi_{1}\,\phi_{1}+ \phi_{0}\,\mu\,\rho 
- \frac{\zeta_{1}\,\tau}{2}
+ \frac{\zeta_{3}\,\bar\tau}{2} \nonumber\\
&+ \mathrm{i}\,\bar\zeta_{1}\,\phi_{0}\,\phi_{1}
- \mathrm{i}\,\zeta_{1}\,\phi_{0}\,\bar\phi_{1}
- \mathrm{i}\,\phi_{0}\,\bar\phi_{0}\,\tilde\zeta_{4}
+ \mathrm{i}\,\phi_{0}^{2}\,\bar{\tilde\zeta}_{4}
- \tilde\zeta_{4}\,\rho
+ \zeta_{2}\,\sigma, \label{zeta3zeta1Renorm}\\
 \meth'\,\tilde\zeta_{4}
=& \meth\,\zeta_{2}-\,\tilde\Psi_{3}\,\phi_{0}
+ \Psi_{2}\,\phi_{1}+ \phi_{0}\,\bigl(\meth'\mu\bigr)
+ \frac{\phi_{0}\,\mu\,\tau}{2}
 -\phi_{1}\,\mu\,\rho
+ \frac{\phi_{0}\,\mu\,\tau}{2} \nonumber\\
&- \mathrm{i}\,\bar\zeta_{1}\,\phi_{1}^{2}
+ \mathrm{i}\,\zeta_{1}\,\phi_{1}\,\bar\phi_{1}
+ \mathrm{i}\,\bar\phi_{0}\,\phi_{1}\,\tilde\zeta_{4}
- \mathrm{i}\,\phi_{0}\,\phi_{1}\,\tilde\zeta_{4}
+ \zeta_{3}\,\lambda
+ \frac{\zeta_{2}\,\tau}{2}
- \frac{\tilde\zeta_{4}\,\bar\tau}{2}\label{zeta4zeta2Renorm}
\end{align}
\end{subequations}

\begin{subequations}
\begin{align}
\mthorn\,\tilde\zeta_{4}
&=\Psi_{2}\,\phi_{0}
+ \mathrm{i}\,\tilde\zeta_{1}\,\phi_{0}\,\phi_{1}
- \mathrm{i}\,\zeta_{3}\,\phi_{0}\,\bar\phi_{1}
+ \mathrm{i}\,\zeta_{1}\,\phi_{0}\,\bar\phi_{1}
- \mathrm{i}\,\zeta_{0}\,\phi_{1}\,\bar\phi_{1}
+ \mathrm{i}\,\phi_{0}\,\bar\phi_{0}\,\tilde\zeta_{4}
- \mathrm{i}\,\phi_{0}^{2}\,\bar{\tilde\zeta}_{4} \nonumber\\
&\quad
+ \zeta_{3}\,\pi
+ \zeta_{1}\,\bar\pi
+ \phi_{0}\,\pi\,\bar\pi
+ 2\,\tilde\zeta_{4}\,\rho
- \phi_{0}\,\mu\,\rho
+ \phi_{0}\,\lambda\,\sigma
- \frac{\zeta_{3}\,\bar\tau}{2}
+ \phi_{0}\,\bigl(\meth\,\pi\bigr)
+ \meth'\,\zeta_{3}, \label{thornTizeta4}\\
\mthorn'\,\zeta_{3}
&= \frac{\zeta_{3}\,\ulomega}{2}
- \mathrm{i}\,\bar\zeta_{2}\,\phi_{0}\,\phi_{1}
+ 2\,\mathrm{i}\,\zeta_{3}\,\phi_{1}\,\bar\phi_{1}
- 2\,\mathrm{i}\,\phi_{0}\,\bar\phi_{1}\,\tilde\zeta_{4}
+ \mathrm{i}\,\phi_{0}^{2}\,\bar{\tilde\zeta}_{5}
- \zeta_{1}\,\bar\lambda
- 2\,\zeta_{3}\,\mu
+ \Tizeta_{5}\,\sigma \nonumber\\
&\quad
- \phi_{1}\,\mu\,\sigma
- \frac{3\,\tilde\zeta_{4}\,\bar\tau}{2}
+ \phi_{0}\,\mu\,\tau
+ \meth\,\tilde\zeta_{4}
- \phi_{0}\,\bigl(\meth\,\mu\bigr)\,. \label{thornprimezeta3alt}
\end{align}
\end{subequations}

\subsection{Structure equation}
\label{StructureEq}

\begin{subequations}
\begin{align}
\mthorn\,\tau &= \tilde\Psi_1
+ 4\,\mathrm{i}\,\bar\zeta_1\,\phi_0
- 2\,\mathrm{i}\,\zeta_3\,\bar\phi_0
- 2\,\mathrm{i}\,\zeta_0\,\bar\phi_1
+ \bar\pi\,\rho
+ \pi\,\sigma
+ \rho\,\tau
+ \sigma\,\bar\tau, \label{Thorntau}\\[6pt]
\mthorn'\,\pi &= -\tilde\Psi_3
+ 2\,\mathrm{i}\,\zeta_5\,\bar\phi_0
- 4\,\mathrm{i}\,\bar\zeta_4\,\phi_1
+ 2\,\mathrm{i}\,\zeta_2\,\bar\phi_1
- \mu\,\pi
- \lambda\,\bar\pi
- \lambda\,\tau
- \mu\,\bar\tau, \label{Thornprimepi}\\[6pt]
\mthorn\,\ulomega &= \Psi_2
+ \bar\Psi_2
+ 2\,\mathrm{i}\,\bar\zeta_4\,\phi_0
- 2\,\mathrm{i}\,\zeta_4\,\bar\phi_0
+ 2\,\mathrm{i}\,\bar\zeta_1\,\phi_1
- 2\,\mathrm{i}\,\zeta_1\,\bar\phi_1
+ 2\,\pi\,\tau
+ 2\,\bar\pi\,\bar\tau
+ 2\,\tau\,\bar\tau, \label{Thornulomega}\\[6pt]
\mthorn'\,\tau &= -\bar{\tilde\Psi}_3
- 2\,\mathrm{i}\,\bar\zeta_5\,\phi_0
- 2\,\mathrm{i}\,\bar\zeta_2\,\phi_1
+ 4\,\mathrm{i}\,\zeta_4\,\bar\phi_1
- 2\,\mu\,\tau
- 2\,\bar\lambda\,\bar\tau
+ \meth\,\ulomega, \label{Thornprimetau}\\
\mthorn'\,\mu &= -2\,\mathrm{i}\,\bar\zeta_5\,\phi_1
+ 2\,\mathrm{i}\,\zeta_5\,\bar\phi_1
- \lambda\,\bar\lambda
- \ulomega\,\mu
- \mu^2,  \label{Thornprimemu}\\[6pt]
\mthorn\,\mu &= \Psi_2
+ \pi\,\bar\pi
+ \mu\,\rho
+ \lambda\,\sigma
+ \meth\,\pi,  \label{Thornmu}\\[6pt]
\mthorn'\,\rho &= -\Psi_2
+ \ulomega\,\rho
- \mu\,\rho
- \lambda\,\sigma
- \tau\,\bar\tau
+ \meth'\,\tau,  \label{Thornprimerho}\\[6pt]
\mthorn\,\rho &= 2\,\mathrm{i}\,\bar\zeta_0\,\phi_0
- 2\,\mathrm{i}\,\zeta_0\,\bar\phi_0
+ \rho^2
+ \sigma\,\bar\sigma,  \label{Thornrho}\\[6pt]
\mthorn'\,\sigma &= -2\,\mathrm{i}\,\bar\zeta_2\,\phi_0
+ 2\,\mathrm{i}\,\zeta_3\,\bar\phi_1
- \bar\lambda\,\rho
+ \ulomega\,\sigma
- \mu\,\sigma
- \tau^2
+ \meth\,\tau,  \label{Thornprimesigma}\\[6pt]
\mthorn\,\sigma &= \Psi_0
+ 2\,\rho\,\sigma, \label{Thornsigma}\\[6pt]
\mthorn'\,\lambda &= -\Psi_4
- \ulomega\,\lambda
- 2\,\lambda\,\mu,  \label{Thornprimelambda}\\[6pt]
\mthorn\,\lambda &= -2\,\mathrm{i}\,\zeta_2\,\bar\phi_0
+ 2\,\mathrm{i}\,\bar\zeta_3\,\phi_1
+ \pi^2
+ \lambda\,\rho
+ \mu\,\sigma
+ \meth'\,\pi,  \label{Thornlambda}\\[6pt]
\meth'\,\mu &= \tilde\Psi_3
+ \lambda\,\tau
- \mu\,\bar\tau
+ \meth\,\lambda,  \label{mulambda}\\[6pt]
\meth'\,\sigma &= \tilde\Psi_1
- \rho\,\tau
+ \sigma\,\bar\tau
+ \meth\,\rho. \label{rhosigma}
\end{align}
\end{subequations}

\subsection{Bianchi identity}
\label{BianchiIdentity}

\begin{align}
\mthorn' \Psi_{0}
&= -4\,\mathrm{i}\,\zeta_{0}\,\bar\zeta_{2}
+ 4\,\mathrm{i}\,\bar\zeta_{1}\,\zeta_{3}
+ 2\,\Psi_{0}\,\ulomega
- 4\,\bar\zeta_{2}\,\phi_{0}^{2}\,\bar\phi_{0}
- 2\,\mathrm{i}\,\tilde\Psi_{1}\,\phi_{0}\,\bar\phi_{1}
+ 8\,\bar\zeta_{1}\,\phi_{0}^{2}\,\bar\phi_{1}
\notag\\
&\quad
+ 2\,\mathrm{i}\,\Psi_{0}\,\phi_{1}\,\bar\phi_{1}
- 4\,\zeta_{0}\,\phi_{0}\,\bar\phi_{1}^{2}
- \Psi_{0}\,\mu
+ 2\,\mathrm{i}\,\zeta_{3}\,\bar\phi_{1}\,\rho
+ 3\,\Psi_{2}\,\sigma
- 2\,\mathrm{i}\,\bar\zeta_{1}\,\phi_{1}\,\sigma
\notag\\
&\quad
+ 2\,\mathrm{i}\,\zeta_{1}\,\bar\phi_{1}\,\sigma
- 2\,\mathrm{i}\,\bar\phi_{0}\,\tilde\zeta_{4}\,\sigma
+ 2\,\mathrm{i}\,\phi_{0}\,\bar\phi_{0}\,\mu\,\sigma
- 5\,\tilde\Psi_{1}\,\tau
- 9\,\mathrm{i}\,\bar\zeta_{1}\,\phi_{0}\,\tau
\notag\\
&\quad
+ 5\,\mathrm{i}\,\zeta_{3}\,\bar\phi_{0}\,\tau
+ 4\,\mathrm{i}\,\zeta_{0}\,\bar\phi_{1}\,\tau
+ 2\,\mathrm{i}\,\phi_{0}\,\bigl(\meth\,\bar\zeta_{1}\bigr)
- 2\,\mathrm{i}\,\bar\phi_{0}\,\bigl(\meth\,\zeta_{3}\bigr)
+ \meth\,\tilde\Psi_{1}\,. \label{thornprimePsi0}
\end{align}

\begin{align}
\mthorn'\,\tilde\Psi_{1}
&= -2\,\mathrm{i}\,\zeta_{1}\,\bar\zeta_{2}
+ \tilde\Psi_{1}\,\ulomega
- 2\,\mathrm{i}\,\Psi_{2}\,\phi_{0}\,\bar\phi_{1}
+ 2\,\mathrm{i}\,\tilde\Psi_{1}\,\phi_{1}\,\bar\phi_{1}
- 2\,\tilde\zeta_{1}\,\phi_{0}\,\phi_{1}\bar\phi_{1}
+ 2\,\zeta_{1}\,\phi_{0}\,\bar\phi_{1}^{2}
\notag\\
&\quad
+ 2\,\mathrm{i}\,\zeta_{3}\,\tilde\zeta_{4}
- 2\,\phi_{0}^{2}\,\bar\phi_{1}\,\bar{\tilde\zeta}_{4}
- 2\,\tilde\Psi_{1}\,\mu
+ 2\,\mathrm{i}\,\bar\phi_{1}\,\tilde\zeta_{4}\,\rho
- 2\,\mathrm{i}\,\phi_{0}\,\bar{\tilde\zeta}_{5}\,\rho
+ 2\,\mathrm{i}\,\phi_{0}\,\bar\phi_{1}\,\mu\,\rho
\notag\\
&\quad
+ 2\,\tilde\Psi_{3}\,\sigma
+ 4\,\mathrm{i}\,\phi_{1}\,\bar{\tilde\zeta}_{4}\,\sigma
- 2\,\mathrm{i}\,\bar\phi_{0}\,\tilde\zeta_{5}\,\sigma
- 3\,\Psi_{2}\,\tau
+ 2\,\mathrm{i}\,\bar\zeta_{1}\,\phi_{1}\,\tau
- 3\,\mathrm{i}\,\zeta_{1}\,\bar\phi_{1}\,\tau
\notag\\
&\quad
- 2\,\mathrm{i}\,\bar\phi_{0}\,\tilde\zeta_{4}\,\tau
+ 2\,\mathrm{i}\,\phi_{0}\,\bar{\tilde\zeta}_{4}\,\tau
+ \mathrm{i}\,\tilde\zeta_{2}\,\phi_{0}\,\bar\tau
- 2\,\mathrm{i}\,\zeta_{3}\,\bar\phi_{1}\,\tau
+ 2\,\mathrm{i}\,\bar\phi_{1}\,\bigl(\meth\,\zeta_{1}\bigr)
\notag\\
&\quad + 2\,\phi_{0}\,\bar\phi_{0}\,\bar\phi_{1}\,\tilde\zeta_{4}
+ \meth\,\Psi_{2}
- 2\,\mathrm{i}\,\phi_{0}\,\bigl(\meth'\,\bar\zeta_{2}\bigr). \label{thornprimePsi1}
\end{align}

\begin{align}
\mthorn'\,\Psi_{2}
&= -2\,\mathrm{i}\,\zeta_{2}\,\bar\zeta_{2}
  -2\,\bar\zeta_{2}\,\bar\phi_{0}\,\phi_{1}^{2}
  -2\,\mathrm{i}\,\tilde\Psi_{3}\,\phi_{0}\,\bar\phi_{1}
  +2\,\mathrm{i}\,\Psi_{2}\,\phi_{1}\,\bar\phi_{1}
  +2\,\bar\zeta_{1}\,\phi_{1}^{2}\,\bar\phi_{1}
  -2\,\zeta_{2}\,\phi_{0}\,\bar\phi_{1}^{2}
\nonumber\\
&\quad
  +2\,\bar\phi_{0}\,\phi_{1}\,\bar\phi_{1}\,\tilde\zeta_{4}
  +6\,\phi_{0}\,\phi_{1}\,\bar\phi_{1}\,\tilde\zeta_{4}
  +4\,\mathrm{i}\,\tilde\zeta_{4}\bar{\tilde\zeta}_{4}
  -2\,\phi_{0}\,\bar\phi_{0}\,\bar\phi_{1}\,\tilde\zeta_{5}
  -2\,\mathrm{i}\,\zeta_{1}\,\bar{\tilde\zeta}_{5}
  -2\,\phi_{0}\,\bar\phi_{0}\,\phi_{1}\,\bar{\tilde\zeta}_{5}
\nonumber\\
&\quad
  -3\,\Psi_{2}\,\mu
  +2\,\mathrm{i}\,\zeta_{1}\,\bar\phi_{1}\,\mu
  -2\,\mathrm{i}\,\bar\phi_{0}\,\tilde\zeta_{4}\,\mu
  -2\,\mathrm{i}\,\phi_{0}\,\bar{\tilde\zeta}_{4}\,\mu
  +2\,\mathrm{i}\,\bar\phi_{1}\,\tilde\zeta_{5}\,\rho
  -4\,\mathrm{i}\,\phi_{1}\,\bar\phi_{1}\,\mu\,\rho
\nonumber\\
&\quad
  + \Psi_{4}\,\sigma
  - \tilde\Psi_{3}\,\tau
  +2\,\mathrm{i}\,\zeta_{2}\,\bar\phi_{1}\,\tau
  -3\,\mathrm{i}\,\phi_{1}\,\bar{\tilde\zeta}_{4}\,\tau
  -\mathrm{i}\,\bar\phi_{0}\,\tilde\zeta_{5}\,\tau
  +5\,\mathrm{i}\,\bar\phi_{0}\,\phi_{1}\,\mu\,\tau
\nonumber\\
&\quad
  +2\,\mathrm{i}\,\phi_{1}\,\bigl(\meth\,\bar\zeta_{4}\bigr)
  -2\,\mathrm{i}\,\bar\phi_{0}\,\bigl(\meth\,\zeta_{5}\bigr)
  + \meth\,\tilde\Psi_{3}\,. \label{thornprimePsi2}
\end{align}

\begin{align}
\mthorn' \,\tilde\Psi_{3}
&= -\,\tilde\Psi_{3}\,\ulomega
+ 2\,\mathrm{i}\,\bar{\tilde\zeta}_{4}\,\tilde\zeta_{5}
- 2\,\mathrm{i}\,\zeta_{2}\,\bar{\tilde\zeta}_{5}
+ 2\,\mathrm{i}\,\bar\zeta_{2}\,\phi_{1}\,\lambda
- 4\,\mathrm{i}\,\bar\phi_{1}\,\tilde\zeta_{4}\,\lambda
- 4\,\tilde\Psi_{3}\,\mu
+ 6\,\mathrm{i}\,\zeta_{2}\,\bar\phi_{1}\,\mu
\nonumber\\
&\quad
- 8\,\mathrm{i}\,\phi_{1}\,\bar{\tilde\zeta}_{4}\,\mu
+ 2\,\mathrm{i}\,\bar\phi_{0}\,\tilde\zeta_{5}\,\mu
+ 4\,\mathrm{i}\,\phi_{0}\,\bar\phi_{1}\,\lambda\,\mu
+ \Psi_{4}\,\tau
+ \mathrm{i}\,\bar\phi_{1}\,\tilde\zeta_{5}\,\tau
- \mathrm{i}\,\phi_{1}\,\bar{\tilde\zeta}_{5}\,\tau
\nonumber\\
&\quad
+ \meth\,\Psi_{4}
+ 2\,\mathrm{i}\,\bar\phi_{1}\,\bigl(\meth'\,\zeta_{5}\bigr)
- 2\,\mathrm{i}\,\phi_{1}\,\bigl(\meth'\,\bar{\zeta}_{5}\bigr)\,. \label{thornprimePsi3}
\end{align}

\begin{align}
\mthorn\,\tilde\Psi_{1}
&= 2\,\mathrm{i}\,\zeta_{0}\,\bar\zeta_{1}
- 2\,\mathrm{i}\,\bar\zeta_{0}\,\zeta_{3}
+ \Psi_{0}\,\pi
- 2\,\mathrm{i}\,\bar\zeta_{0}\,\phi_{0}\,\pi
+ 2\,\mathrm{i}\,\zeta_{0}\,\bar\phi_{0}\,\pi
+ 4\,\tilde\Psi_{1}\,\rho
+ 4\,\mathrm{i}\,\bar\zeta_{1}\,\phi_{0}\,\rho
\nonumber\\
&\quad
- 2\,\mathrm{i}\,\zeta_{3}\,\bar\phi_{0}\,\rho
- 4\,\mathrm{i}\,\zeta_{0}\,\bar\phi_{1}\,\rho
- 2\,\mathrm{i}\,\bar\zeta_{3}\,\phi_{0}\,\sigma
+ 4\,\mathrm{i}\,\zeta_{1}\,\bar\phi_{0}\,\sigma
+ 3\,\mathrm{i}\,\bar\zeta_{0}\,\phi_{0}\,\tau
\nonumber\\
&\quad
- 3\,\mathrm{i}\,\zeta_{0}\,\bar\phi_{0}\,\tau
- 2\,\Psi_{0}\,\bar\tau
+ 2\,\mathrm{i}\,\bar\phi_{0}\,\bigl(\meth\,\zeta_{0}\bigr)
- 2\,\mathrm{i}\,\phi_{0}\,\bigl(\meth\,\bar{\zeta}_{0}\bigr)
+ \meth'\,\Psi_{0}\,. \label{thornPsi1}
\end{align}

\begin{align}
\mthorn\,\Psi_{2}
&= 4\,\mathrm{i}\,\zeta_{1}\,\bar\zeta_{1}
- 2\,\mathrm{i}\,\zeta_{3}\,\bar\zeta_{3}
+ 2\,\mathrm{i}\,\Psi_{2}\,\phi_{0}\,\bar\phi_{0}
- 2\,\mathrm{i}\,\tilde\Psi_{1}\,\bar\phi_{0}\,\phi_{1}
+ 6\,\bar\zeta_{1}\,\phi_{0}\,\bar\phi_{0}\,\phi_{1}
- 2\,\zeta_{3}\,\bar\phi_{0}^{2}\,\phi_{1}
- 2\,\bar\zeta_{3}\,\phi_{0}^{2}\,\bar\phi_{1}
\nonumber\\
&\quad
+ 2\,\zeta_{1}\,\phi_{0}\,\bar\phi_{0}\,\bar\phi_{1}
- 2\,\bar\zeta_{0}\,\phi_{0}\,\phi_{1}\,\bar\phi_{1}
+ 2\,\tilde\Psi_{1}\,\pi
+ 4\,\mathrm{i}\,\zeta_{1}\,\phi_{0}\,\pi
- 2\,\mathrm{i}\,\zeta_{3}\,\bar\phi_{0}\,\pi
\nonumber\\
&\quad
- 2\,\mathrm{i}\,\zeta_{0}\,\bar\phi_{1}\,\pi
+ 3\,\Psi_{2}\,\rho
- 2\,\mathrm{i}\,\bar\zeta_{1}\,\phi_{1}\,\rho
- 2\,\mathrm{i}\,\zeta_{1}\,\bar\phi_{1}\,\rho
+ 2\,\mathrm{i}\,\bar\phi_{0}\,\tilde\zeta_{4}\,\rho
-\Psi_0\lambda
\nonumber\\
&\quad
- 2\,\mathrm{i}\,\phi_{0}\,\bar\phi_{0}\,\mu\,\rho
- \tilde\Psi_{1}\,\tau
- \mathrm{i}\,\bar\zeta_{1}\,\phi_{0}\,\tau
+ 3\,\mathrm{i}\,\zeta_{0}\,\bar\phi_{1}\,\tau
+2\mathrm{i}\bar\zeta_0\phi_0\mu
-2\mathrm{i}\zeta_0\bar\phi_0\mu
- 2\,\mathrm{i}\,\bar\phi_{1}\,\bigl(\meth'\,\zeta_{0}\bigr)
\nonumber\\
&\quad
+ 2\,\mathrm{i}\,\phi_{0}\,\bigl(\meth'\,\bar\zeta_{1}\bigr)
+ \meth'\,\tilde\Psi_{1}\,-2\zeta_0\bar\phi_0\phi_1\bar\phi_1
-2\mathrm{i}\bar\zeta_0\Tizeta_4
-2\phi_0\bar\phi_0^2\Tizeta_4+2\phi_0^2\bar\phi_0\bar\Tizeta_4. \label{thornPsi2}
\end{align}

\begin{align}
\mthorn\,\tilde\Psi_{3}
&= 2\,\mathrm{i}\,\bar\zeta_{1}\,\zeta_{2}
- 2\,\mathrm{i}\,\bar\zeta_{3}\,\tilde\zeta_{4}
- 2\,\tilde\Psi_{1}\,\lambda
- 4\,\mathrm{i}\,\bar\zeta_{1}\,\phi_{0}\,\lambda
+ 2\,\mathrm{i}\,\zeta_{3}\,\bar\phi_{0}\,\lambda
+ 2\,\mathrm{i}\,\zeta_{0}\,\bar\phi_{1}\,\lambda
\nonumber\\
&\quad
+ 2\,\mathrm{i}\,\bar\zeta_{3}\,\phi_{0}\,\mu
- 4\,\mathrm{i}\,\zeta_{1}\,\bar\phi_{0}\,\mu
+ 2\,\mathrm{i}\,\bar\zeta_{0}\,\phi_{1}\,\mu
+ 3\,\Psi_{2}\,\pi
- 2\,\mathrm{i}\,\bar\zeta_{1}\,\phi_{1}\,\pi
+ 2\,\mathrm{i}\,\zeta_{1}\,\bar\phi_{1}\,\pi
\nonumber\\
&\quad
+ 2\,\mathrm{i}\,\bar\phi_{0}\,\tilde\zeta_{4}\,\pi
- 2\,\mathrm{i}\,\phi_{0}\,\bar{\tilde\zeta}_{4}\,\pi
+ 2\,\mathrm{i}\,\zeta_{2}\,\bar\phi_{0}\,\bar\pi
- 2\,\mathrm{i}\,\zeta_{3}\,\phi_{1}\,\bar\pi
+ 2\,\tilde\Psi_{3}\,\rho
- 2\,\mathrm{i}\,\zeta_{2}\,\bar\phi_{1}\,\rho
\nonumber\\
&\quad
+ \mathrm{i}\,\zeta_{2}\,\bar\phi_{0}\,\tau
+ \mathrm{i}\,\bar\zeta_{3}\,\phi_{1}\,\tau
+ 2\,\mathrm{i}\,\bar\phi_{0}\,\bigl(\meth\,\zeta_{2}\bigr)
- 2\,\mathrm{i}\,\phi_{1}\,\bigl(\meth\,\bar\zeta_{3}\bigr)
+ \meth'\,\Psi_{2}\,. \label{thornPsi3}
\end{align}

\begin{align*}
\mthorn\,\Psi_{4}
&= 2\,\mathrm{i}\,\Psi_{4}\,\phi_{0}\,\bar\phi_{0}
  - 2\,\mathrm{i}\,\tilde\Psi_{3}\,\bar\phi_{0}\,\phi_{1}
  - 4\,\bar\zeta_{3}\,\phi_{1}^{2}\,\bar\phi_{1}
  + 4\,\mathrm{i}\,\zeta_{2}\,\bar{\tilde\zeta}_{4}
  + 8\,\bar\phi_{0}\,\phi_{1}^{2}\,\bar{\tilde\zeta}_{4}
  - 4\,\mathrm{i}\,\bar\zeta_{3}\,\tilde\zeta_{5}
  - 4\,\bar\phi_{0}^{2}\,\phi_{1}\,\tilde\zeta_{5}
  - 3\,\Psi_{2}\,\lambda
\\
&\quad
  + 2\,\mathrm{i}\,\zeta_{1}\,\bar\phi_{1}\,\lambda
  - 2\,\mathrm{i}\,\bar\phi_{0}\,\tilde\zeta_{4}\,\lambda
  + 2\,\mathrm{i}\,\phi_{0}\,\bar{\tilde\zeta}_{4}\,\lambda
  - 6\,\mathrm{i}\,\zeta_{2}\,\bar\phi_{0}\,\mu
  + 8\,\mathrm{i}\,\bar\zeta_{3}\,\phi_{1}\,\mu
\\
&\quad
  + 4\,\tilde\Psi_{3}\,\pi
  - 4\,\mathrm{i}\,\zeta_{2}\,\bar\phi_{1}\,\pi
  + 8\,\mathrm{i}\,\phi_{1}\,\bar{\tilde\zeta}_{4}\,\pi
  - 4\,\mathrm{i}\,\bar\phi_{0}\,\tilde\zeta_{5}\,\pi
  + \Psi_{4}\,\rho
  + \tilde\Psi_{3}\,\tau
\\
&\quad
  - \mathrm{i}\,\zeta_{2}\,\bar\phi_{1}\,\tau
  + \mathrm{i}\,\phi_{1}\,\bar{\tilde\zeta}_{4}\,\tau
  - \mathrm{i}\,\bar\phi_{0}\,\phi_{1}\,\mu\,\tau
  - 2\,\mathrm{i}\,\bar\phi_{1}\,\bigl(\meth'\,\zeta_{2}\bigr)
  + 2\,\mathrm{i}\,\phi_{1}\,\bigl(\meth'\,\bar\zeta_{4}\bigr)
  + \meth'\,\tilde\Psi_{3}\,. \label{thornPsi4}
\end{align*}

\begin{align}
\mthorn\,\mathcal{K}
&= -4\,\mathrm{i}\,\zeta_{1}\,\bar\zeta_{1}
  +2\,\mathrm{i}\,\zeta_{3}\,\bar\zeta_{3}
  +\mathrm{i}\,\mathcal{K}\,\phi_{0}\,\bar\phi_{0}
  +\mathrm{i}\,\bar{\mathcal{K}}\,\phi_{0}\,\bar\phi_{0}
  +\mathrm{i}\,\tilde\Psi_{1}\,\bar\phi_{0}\,\phi_{1}
  +\mathrm{i}\,\bar\TiPsi_{1}\,\phi_{0}\,\bar\phi_{1}
  +\mathrm{i}\,\bar\zeta_{0}\,\tilde\zeta_{4}
  +\mathrm{i}\,\zeta_{0}\,\bar{\tilde\zeta}_{4}
\nonumber\\
&\quad
  -2\,\tilde\Psi_{1}\,\pi
  -2\,\mathrm{i}\,\bar\zeta_{1}\,\phi_{0}\,\pi
  +\mathrm{i}\,\zeta_{3}\,\bar\phi_{0}\,\pi
  +\mathrm{i}\,\zeta_{0}\,\bar\phi_{1}\,\pi
  +\mathrm{i}\,\bar\zeta_{3}\,\phi_{0}\,\bar\pi
  -2\,\mathrm{i}\,\zeta_{1}\,\bar\phi_{0}\,\bar\pi
  +\mathrm{i}\,\zeta_{0}\,\phi_{1}\,\bar\pi
  +2\,\mathcal{K}\,\rho
\nonumber\\
&\quad
  +2\,\mathrm{i}\,\bar\zeta_{1}\,\phi_{1}\,\rho
  +2\,\mathrm{i}\,\zeta_{1}\,\bar\phi_{1}\,\rho
  -\mathrm{i}\,\bar\phi_{0}\,\tilde\zeta_{4}\,\rho
  -\mathrm{i}\,\phi_{0}\,\bar{\tilde\zeta}_{4}\,\rho
  +\pi\bar\pi\,\rho
  +\mathrm{i}\,\zeta_{2}\,\bar\phi_{0}\,\sigma
  -\mathrm{i}\,\zeta_{3}\,\phi_{1}\,\sigma
\nonumber\\
&\quad
  +\mathrm{i}\,\phi_{0}\,\bar\phi_{0}\,\lambda\,\sigma
  -\pi^{2}\,\sigma
  +\mathrm{i}\,\zeta_{2}\,\phi_{0}\,\bar\sigma
  -\mathrm{i}\,\zeta_{3}\,\phi_{1}\,\bar\sigma
  +\mathrm{i}\,\phi_{0}\,\bar\phi_{0}\,\bar\lambda\,\bar\sigma
+\tfrac{1}{2}\,\mathrm{i}\,\zeta_{1}\,\bar\phi_{0}\,\tau
  -\tfrac{3}{2}\,\mathrm{i}\,\bar\zeta_{0}\,\phi_{1}\,\tau
\nonumber\\
&\quad
  +\tfrac{3}{2}\,\mathrm{i}\,\zeta_{0}\,\bar\phi_{1}\,\bar\tau
  +\mathrm{i}\,\phi_{1}\,\bigl(\meth\,\bar\zeta_{0}\bigr)
  +\TiPsi_{1}\,\tau
  +\tfrac{1}{2}\,\mathrm{i}\,\bar\zeta_{1}\,\phi_{0}\,\bar\tau
  -\tfrac{3}{2}\,\mathrm{i}\,\zeta_{0}\,\bar\phi_{1}\,\tau
\nonumber\\
&\quad
  -\mathrm{i}\,\bar\phi_{0}\,\bigl(\meth\,\zeta_{1}\bigr)
  +\rho\,\bigl(\meth\,\pi\bigr)
  +\mathrm{i}\,\bar\phi_{1}\,\bigl(\meth'\,\zeta_{0}\bigr)
  -\mathrm{i}\,\phi_{0}\,\bigl(\meth'\,\bar\zeta_{1}\bigr)
  -\meth'\,\Psi_{1}
  -\sigma\,\bigl(\meth'\,\pi\bigr).
\end{align}

\begin{align*}
\mthorn'\tilde\Psi_1 &:= -3\,\mathrm{i}\,\zeta_1\,\bar\zeta_2
+ \tilde\Psi_1\,\ulomega
+ 2\,\mathrm{i}\,\mathcal{K}\,\phi_0\,\bar\phi_1
+ 2\,\mathrm{i}\,\tilde\Psi_1\,\phi_1\,\bar\phi_1
+ 3\,\mathrm{i}\,\zeta_3\,\tilde\zeta_4
- \mathrm{i}\,\zeta_1\,\bar\phi_0\,\bar\lambda
- 2\,\tilde\Psi_1\,\mu \\
&\quad
+ \mathrm{i}\,\bar\zeta_1\,\phi_0\,\mu
- \mathrm{i}\,\zeta_3\,\bar\phi_0\,\mu
+ 3\,\mathrm{i}\,\bar\phi_1\,\tilde\zeta_4\,\rho
- 2\,\mathrm{i}\,\phi_0\,\bar{\tilde\zeta}_5\,\rho
- \mathrm{i}\,\phi_0\,\bar\phi_1\,\mu\,\rho
+ 2\,\tilde\Psi_3\,\sigma
+ 3\,\mathrm{i}\,\phi_1\,\bar\tilde\zeta_4\,\sigma \\
&\quad
- 2\,\mathrm{i}\,\bar\phi_0\,\tilde\zeta_5\,\sigma
+ 2\,\mathrm{i}\,\phi_0\,\bar\phi_1\,\lambda\,\sigma
+ \mathrm{i}\,\bar\phi_0\,\phi_1\,\mu\,\sigma
+ 3\,\mathcal{K}\,\tau
- \tfrac{3}{2}\,\mathrm{i}\,\bar\zeta_1\,\phi_1\,\tau
+ \tfrac{1}{2}\,\mathrm{i}\,\zeta_1\,\bar\phi_1\,\tau
+ \tfrac{1}{2}\,\mathrm{i}\,\bar\phi_0\,\tilde\zeta_4\,\tau \\
&\quad
- \tfrac{1}{2}\,\mathrm{i}\,\phi_0\,\bar{\tilde\zeta}_4\,\tau
- 3\,\mu\,\rho\,\tau
+ 3\,\lambda\,\sigma\,\tau 
+ \mathrm{i}\,\bar\zeta_2\,\phi_0\,\tau
- 2\,\mathrm{i}\,\zeta_3\,\bar\phi_1\,\bar\tau
- \meth\,\mathcal{K}
+ \mathrm{i}\,\bar\phi_1\,(\meth\,\zeta_1)
+ \mathrm{i}\,\phi_1\,(\meth\,\bar\zeta_1) \\[6pt]
&\quad
- \mathrm{i}\,\bar\phi_0\,(\meth\,\zeta_4)
+ \mathrm{i}\,\phi_0\,(\meth\,\bar\zeta_4)
- \sigma\,(\meth\,\lambda)
+ \rho\,(\meth\,\mu)
+ \mu\,(\meth\,\rho)
- \lambda\,(\meth\,\sigma) 
- 2\,\mathrm{i}\,\phi_0\,(\meth'\,\bar\zeta_2)\,.
\end{align*}

\begin{align*}
\mthorn'\mathcal{K} &= 2\,\mathrm{i}\,\zeta_2\,\bar\zeta_2
+ 2\,\mathrm{i}\,\tilde\Psi_3\,\phi_0\,\bar\phi_1
+ 2\,\mathrm{i}\,\mathcal{K}\,\phi_1\,\bar\phi_1
- 4\,\mathrm{i}\,\Tizeta_4\,\bar{\tilde\zeta}_4 
+ \mathrm{i}\,\bar\zeta_1\,\tilde\zeta_5
+ \mathrm{i}\,\zeta_1\,\bar{\tilde\zeta}_5
+ \mathrm{i}\,\bar\zeta_2\,\phi_0\,\lambda
- 2\,\mathrm{i}\,\zeta_3\,\bar\phi_1\,\lambda \\ 
&\quad
+ \mathrm{i}\,\zeta_2\,\bar\phi_0\,\bar\lambda
- 2\,\mathcal{K}\,\mu
- 2\,\mathrm{i}\,\bar\zeta_1\,\phi_1\,\mu
+ 2\,\mathrm{i}\,\bar\phi_0\,\tilde\zeta_4\,\mu
+ 2\,\mathrm{i}\,\phi_0\,\bar{\tilde\zeta}_4\,\mu
- \mathrm{i}\,\bar\phi_1\,\tilde\zeta_5\,\rho
- \mathrm{i}\,\phi_1\,\bar{\tilde\zeta}_5\,\rho \\ 
&\quad
+ 2\,\mathrm{i}\,\phi_1\,\bar\phi_1\,\mu\,\rho
+ 2\,\mathrm{i}\,\phi_1\,\bar\phi_1\,\lambda\,\sigma
+ \tilde\Psi_3\,\tau
- \tfrac{3}{2}\,\mathrm{i}\,\zeta_2\,\bar\phi_1\,\tau
+ \mathrm{i}\,\phi_1\,\bar{\tilde\zeta}_4\,\tau
+ \tfrac{1}{2}\,\mathrm{i}\,\bar\phi_0\,\tilde\zeta_5\,\tau
- 2\,\mathrm{i}\,\bar\phi_0\,\phi_1\,\mu\,\tau\\ 
&\quad
+ \lambda\,\tau^2
- \tfrac{1}{2}\,\mathrm{i}\,\bar\zeta_2\,\phi_1\,\bar\tau
+ 2\,\mathrm{i}\,\bar\phi_1\,\tilde\zeta_4\,\bar\tau
+ \tfrac{1}{2}\,\mathrm{i}\,\phi_0\,\bar{\tilde\zeta}_5\,\bar\tau
- 3\,\mathrm{i}\,\phi_0\,\bar\phi_1\,\mu\,\bar\tau 
- \mu\,\tau\,\bar\tau
- \mathrm{i}\,\bar\phi_1\,(\meth\,\zeta_2)\\ 
&\quad
- 2\,\mathrm{i}\,\phi_1\,(\meth\,\bar\zeta_4)
+ \mathrm{i}\,\bar\phi_0\,(\meth\,\zeta_5)
- \meth\,\tilde\Psi_3
- \lambda\,(\meth\,\tau)
+ \mathrm{i}\,\phi_1\,(\meth'\,\bar\zeta_2)
+ \mathrm{i}\,\phi_0\,(\meth'\,\bar\zeta_5)
+ \mu\,(\meth'\,\tau)\,.
\end{align*}

\begin{align*}
\mthorn\,\tilde\Psi_3 &:= 3\,\mathrm{i}\,\bar\zeta_1\,\zeta_2
- 3\,\mathrm{i}\,\bar\zeta_3\,\tilde\zeta_4
- 2\,\tilde\Psi_1\,\lambda
- 3\,\mathrm{i}\,\bar\zeta_1\,\phi_0\,\lambda
+ 2\,\mathrm{i}\,\zeta_3\,\bar\phi_0\,\lambda
+ 2\,\mathrm{i}\,\zeta_0\,\bar\phi_1\,\lambda
+ 3\,\mathrm{i}\,\bar\zeta_3\,\phi_0\,\mu \\ 
&\quad
- 5\,\mathrm{i}\,\zeta_1\,\bar\phi_0\,\mu
+ 2\,\mathrm{i}\,\bar\zeta_0\,\phi_1\,\mu
- 3\,\mathcal{K}\,\pi
+ \mathrm{i}\,\bar\zeta_1\,\phi_1\,\pi
- \mathrm{i}\,\zeta_1\,\bar\phi_1\,\pi
- \mathrm{i}\,\bar\phi_0\,\tilde\zeta_4\,\pi
+ \mathrm{i}\,\phi_0\,\bar{\tilde\zeta}_4\,\pi \\ 
&\quad
+ 2\,\mathrm{i}\,\zeta_2\,\bar\phi_0\,\bar\pi
- 2\,\mathrm{i}\,\bar\zeta_3\,\phi_1\,\bar\pi
+ 2\,\tilde\Psi_3\,\rho
- 2\,\mathrm{i}\,\zeta_2\,\bar\phi_1\,\rho
- \mathrm{i}\,\phi_1\,\bar{\tilde\zeta}_4\,\rho
+ \mathrm{i}\,\bar\phi_0\,\phi_1\,\mu\,\rho
+ 3\,\mu\,\pi\,\rho \\ 
&\quad
- 3\,\lambda\,\pi\,\sigma
+ \mathrm{i}\,\bar\phi_1\,\tilde\zeta_4\,\bar\sigma
- \mathrm{i}\,\phi_0\,\bar\phi_1\,\mu\,\bar\sigma
+ \mathrm{i}\,\zeta_2\,\bar\phi_0\,\tau
+ \mathrm{i}\,\bar\zeta_3\,\phi_1\,\tau
- \tfrac{1}{2}\,\mathrm{i}\,\bar\zeta_1\,\phi_1\,\bar\tau
+ \tfrac{1}{2}\,\mathrm{i}\,\zeta_1\,\bar\phi_1\,\bar\tau \\ 
&\quad
- \tfrac{1}{2}\,\mathrm{i}\,\bar\phi_0\,\tilde\zeta_4\,\bar\tau
+ \tfrac{1}{2}\,\mathrm{i}\,\phi_0\,\bar{\tilde\zeta}_4\,\bar\tau 
+ 2\,\mathrm{i}\,\bar\phi_0\,(\meth\,\zeta_2)
- 2\,\mathrm{i}\,\phi_1\,(\meth\,\tilde\zeta_3)
- \meth'\,\mathcal{K}
- \mathrm{i}\,\bar\phi_1\,(\meth'\,\zeta_1)
+ \mathrm{i}\,\phi_1\,(\meth'\,\bar\zeta_1) \\ 
&\quad
- \mathrm{i}\,\bar\phi_0\,(\meth'\,\zeta_4)
+ \mathrm{i}\,\phi_0\,(\meth'\,\bar\zeta_4)
- \sigma\,(\meth'\,\lambda)
+ \rho\,(\meth'\,\mu)
+ \mu\,(\meth'\,\rho)
- \lambda\,(\meth'\,\sigma).
\end{align*}

\subsection{Advanced Structure}

\begin{align}
\mthorn\,\tilde\tau
&= -2\,\mathrm{i}\,\Psi_{2}\,\phi_{0}\,\bar\phi_{0}
+2\,\mathrm{i}\,\tilde\Psi_{1}\,\bar\phi_{0}\,\phi_{1}
-6\,\bar\zeta_{1}\,\phi_{0}\,\bar\phi_{0}\,\phi_{1}
+2\,\zeta_{3}\,\bar\phi_{0}^{2}\,\phi_{1}
+2\,\bar\zeta_{3}\,\phi_{0}^{2}\,\bar\phi_{1}
-2\,\zeta_{1}\,\phi_{0}\,\bar\phi_{0}\,\bar\phi_{1}
\nonumber\\
&\quad
+2\,\bar\zeta_{0}\,\phi_{0}\,\phi_{1}\,\bar\phi_{1}
+2\,\zeta_{0}\,\bar\phi_{0}\,\phi_{1}\,\bar\phi_{1}
+2\,\mathrm{i}\,\bar\zeta_{0}\,\tilde\zeta_{4}
+2\,\phi_{0}\,\bar\phi_{0}^{2}\,\tilde\zeta_{4}
-2\,\mathrm{i}\,\zeta_{0}\,\bar{\tilde\zeta}_{4}
-2\,\phi_{0}^{2}\,\bar\phi_{0}\,\bar{\tilde\zeta}_{4}
+\Psi_{0}\,\lambda
\nonumber\\
&\quad
-2\,\mathrm{i}\,\bar\zeta_{0}\,\phi_{0}\,\mu
+2\,\mathrm{i}\,\zeta_{0}\,\bar\phi_{0}\,\mu
-\tilde\Psi_{1}\,\pi
-\Psi_{2}\,\rho
-2\,\mathrm{i}\,\bar\zeta_{1}\,\phi_{1}\,\rho
+2\,\mathrm{i}\,\zeta_{1}\,\bar\phi_{1}\,\rho
\nonumber\\
&\quad
-2\,\mathrm{i}\,\bar\phi_{0}\,\tilde\zeta_{4}\,\rho
+2\,\mathrm{i}\,\phi_{0}\,\bar\phi_{0}\,\mu\,\rho
+\pi\bar\pi\,\rho
+\pi^{2}\,\sigma
+2\,\mathrm{i}\,\zeta_{3}\,\bar\phi_{1}\,\sigma
+\pi\,\bar\sigma\,\tau
- \sigma\,\tau^{2}
\nonumber\\
&\quad
+2\,\rho\,\tilde\tau
- \mathrm{i}\,\bar\zeta_{1}\,\phi_{0}\,\bar\tau
+ \mathrm{i}\,\zeta_{3}\,\bar\phi_{0}\,\bar\tau
- \bar\pi\,\rho\,\bar\tau
- \sigma\,\bar\tau^{2}
+ \tau\,\bigl(\meth\,\bar\sigma\bigr)
+ \bar\sigma\,\bigl(\meth\,\bar\tau\bigr)
\nonumber\\
&\quad
+2\,\mathrm{i}\,\phi_{0}\,\bigl(\meth'\,\bar\zeta_{1}\bigr)
-2\,\mathrm{i}\,\bar\phi_{0}\,\bigl(\meth'\,\zeta_{3}\bigr)
+ \sigma\,\bigl(\meth'\,\pi\bigr)
+ \rho\,\bigl(\meth'\,\bar\pi\bigr)
+ \bar\pi\,\bigl(\meth'\,\rho\bigr)
+ \pi\,\bigl(\meth'\,\sigma\bigr)
\nonumber\\
&\quad
+ \bar\tau\,\bigl(\meth'\,\sigma\bigr)
+ \sigma\,\bigl(\meth'\,\bar\tau\bigr). \label{thornTitau}
\end{align}

\begin{align}
\mthorn'\,\tilde\pi &:= 2\,\mathrm{i}\,\bar{\tilde\Psi}_3\,\bar\phi_0\,\phi_1
- 2\,\bar\zeta_2\,\bar\phi_0\,\phi_1^2
- 2\,\mathrm{i}\,\tilde\Psi_3\,\phi_0\,\bar\phi_1
+ 2\,\mathrm{i}\,\Psi_2\,\phi_1\,\bar\phi_1
- 2\,\mathrm{i}\,\bar\Psi_2\,\phi_1\,\bar\phi_1   \nonumber\\
&\quad
- 2\,\zeta_2\,\phi_0\,\bar\phi_1^2
+ 4\,\bar\phi_0\,\phi_1\bar\phi_1\,\tilde\zeta_4
+ 4\,\phi_0\,\phi_1\bar\phi_1\,\bar{\tilde{\zeta}}_4
+ 2\,\mathrm{i}\,\bar\zeta_1\,\tilde\zeta_5
- 2\,\phi_0\,\bar\phi_0\,\bar\phi_1\,\tilde\zeta_5   \nonumber\\
&\quad
-2\mathrm{i}\zeta_1\bar{\tilde\zeta}_5
-2\phi_0\bar\phi_0\phi_1\bar{\tilde\zeta}_5
+ 2\,\mathrm{i}\,\zeta_2\,\bar\phi_0\,\bar\lambda
- 2\,\mathrm{i}\,\bar\zeta_3\,\phi_1\,\bar\lambda
- \Psi_2\,\mu
- 2\,\mathrm{i}\,\bar\zeta_1\,\phi_1\,\mu  \nonumber \\
&\quad
+2\,\mathrm{i}\,\zeta_1\,\bar\phi_1\,\mu 
+ 2\,\mathrm{i}\,\bar\phi_0\,\tilde\zeta_4\,\mu
- 2\,\mathrm{i}\,\phi_0\,\bar{\tilde\zeta}_4\,\mu
- \bar\lambda\,\bar\pi^2
- 2\,\mu\,\bar\pi
- \tilde\Psi_3\,\bar\pi
+ 2\,\mathrm{i}\,\zeta_2\,\bar\phi_1\,\bar\pi   \nonumber\\
&\quad
- 4\,\mathrm{i}\,\phi_1\,\bar{\tilde\zeta}_4\,\bar\pi
+ 2\,\mathrm{i}\,\bar\phi_0\,\tilde\zeta_5\,\bar\pi
- \lambda\,\bar\pi^2
+ \Psi_4\,\sigma
- \mathrm{i}\,\zeta_2\,\bar\phi_1\,\tau   \nonumber\\
&\quad
+ 4\,\mathrm{i}\,\phi_1\,\bar{\tilde\zeta}_4\,\tau
- 2\,\mathrm{i}\,\bar\phi_0\,\tilde\zeta_5\,\tau
+ \lambda\,\tau^2
- \mathrm{i}\,\bar\zeta_2\,\phi_1\,\bar\tau
+ \bar\lambda\,\pi\,\bar\tau
- \mu\,\bar\pi\,\bar\tau
+ \mu\,\tau\,\bar\tau   \nonumber\\
& \quad
+2\,\mathrm{i}\,\bar\phi_{1}\,\bigl(\meth\,\zeta_{2}\bigr)
   -2\,\mathrm{i}\,\phi_{1}\,\bigl(\meth'\,\bar\zeta_{2}\bigr)
   -\bar\pi\,\bigl(\meth\,\lambda\bigr)
   -\tau\,\bigl(\meth\,\lambda\bigr)
   \nonumber\\
&\quad
   -\bar\tau\,\bigl(\meth\,\mu\bigr)
   -\lambda\,\bigl(\meth\,\bar\pi\bigr)
   -\lambda\,\bigl(\meth\,\tau\bigr)
   -\mu\,\bigl(\meth\,\bar\tau\bigr)
   -\pi\,\bigl(\meth'\,\bar\lambda\bigr)
   -\bar\lambda\,\bigl(\meth'\,\pi\bigr). \label{thornprimeTipi}
\end{align}

\begin{align}
\mthorn \, \tilde{\ulomega}
&= 2\,\mathrm{i}\,\zeta_{1}\,\bar{\zeta}_{2}
 + 4\,\mathrm{i}\,\bar{\tilde\Psi}_{3}\,\phi_{0}\,\bar{\phi}_{0}
 + 4\,\mathrm{i}\,\Psi_{2}\,\phi_{0}\,\bar{\phi}_{1}
 - 4\,\mathrm{i}\,\bar{\Psi}_{2}\,\phi_{0}\,\bar{\phi}_{1} 
  - 4\,\mathrm{i}\,\tilde{\Psi}_{1}\,\phi_{1}\,\bar{\phi}_{1}
  -2\mathrm{i}\zeta_3\bar{\tilde\zeta}_4 \nonumber\\
&\quad
 + 4\,\bar{\tilde\Psi}_{1}\,\bar{\lambda}
 - 10\,\mathrm{i}\,\zeta_{1}\,\bar{\phi}_{0}\,\bar{\lambda}
 + 4\,\mathrm{i}\,\bar{\zeta}_{0}\,\phi_{1}\,\bar{\lambda}
  - 2\,\mathrm{i}\,\bar\zeta_{1}\,\phi_{0}\,\mu
 + 2\,\mathrm{i}\,\zeta_{3}\,\bar{\phi}_{0}\,\mu
 + 4\,\mathrm{i}\,\zeta_{0}\,\bar{\phi}_{1}\,\mu \nonumber\\
&\quad
 + 4\,\mathrm{i}\,\bar{\zeta}_{2}\,\phi_{0}\,\pi 
 - 4\,\mathrm{i}\,\zeta_{3}\,\bar{\phi}_{1}\,\pi
 - 4\,\bar{\Psi}_{2}\,\bar\pi
 + 6\,\mathrm{i}\,\bar\zeta_{1}\,\phi_{1}\,\bar\pi
 - 6\,\mathrm{i}\,\zeta_{1}\,\bar{\phi}_{1}\,\bar\pi
 - 6\,\mathrm{i}\,\bar{\phi}_{0}\,\tilde{\zeta}_{4}\,\pi \nonumber\\
&\quad
+ 6\,\mathrm{i}\,\phi_{0}\,\bar{\tilde{\zeta}}_{4}\,\pi
- 2\,\bar{\tilde{\Psi}}_{3}\,\rho
+ \tilde{\ulomega}\,\rho
- 4\,\mathrm{i}\,\bar{\zeta}_{2}\,\phi_{1}\,\rho
+ 6\,\mathrm{i}\,\bar{\phi}_{1}\,\tilde{\zeta}_{4}\,\rho 
- 6\,\mathrm{i}\,\phi_{0}\,\bar{\phi}_{1}\,\mu\,\rho
- 2\,\tilde{\Psi}_{3}\,\sigma \nonumber\\
&\quad
- \bar{\tilde\ulomega}\,\sigma
- 4\,\mathrm{i}\,\zeta_{2}\,\bar{\phi}_{1}\,\sigma
- 2\,\mathrm{i}\,\phi_{1}\,\bar{\tilde{\zeta}}_{4}\,\sigma
+ 2\,\mathrm{i}\,\bar{\phi}_{0}\,\phi_{1}\,\mu\,\sigma
- 2\,\bar{\Psi}_{2}\,\tau
- 3\,\mathrm{i}\,\bar{\zeta}_{1}\,\phi_{1}\,\tau
+ 5\,\mathrm{i}\,\zeta_{1}\,\bar{\phi}_{1}\,\tau \nonumber\\
&\quad
+ \mathrm{i}\,\bar{\phi}_{0}\,\tilde{\zeta}_{4}\,\tau
+ \mathrm{i}\,\phi_{0}\,\bar{\tilde{\zeta}}_{4}\,\tau
- 2\,\mathrm{i}\,\phi_{0}\,\bar{\phi}_{0}\,\mu\,\tau
+ 2\,\pi\,\bar{\pi}\,\tau
- 2\,\pi\,\tau^{2}
+ 2\,\bar{\pi}^{2}\,\bar{\tau}
- 2\,\tau^{2}\,\bar{\tau} \nonumber\\
&\quad
- 6\,\mathrm{i}\,\bar{\phi}_{1}\,(\meth \zeta_{1})
+ 2\,\mathrm{i}\,\phi_{1}\,(\meth \bar{\zeta}_{1})
- 2\,\mathrm{i}\,\bar{\phi}_{0}\,(\meth \zeta_{4})
+ 6\,\mathrm{i}\,\phi_{0}\,(\meth \bar{\zeta}_{4}) \nonumber\\
&\quad
+ 2\,\tau\,(\meth \pi)
+ 2\,\bar{\tau}\,(\meth \bar{\pi})
+ 2\,\pi\,(\meth \tau)
+ 2\,\bar{\pi}\,(\meth \bar{\tau})
+ 2\,\tau\,(\meth \bar{\tau})
+ 2\,\bar\tau\,(\meth\tau)
+ 2\,\sigma\,(\meth' \,\ulomega). \label{thornTiulomega}
\end{align}

\subsection{Equations of M in Sec. \ref{TrappedSurface}}
\label{ExpressionofM}

\begin{align*}
M&=
 4\,\mathrm{i}\,\overline{\zeta}_{0}\,\ulomega\,\phi_{0}
- 4\,\mathrm{i}\,\zeta_{0}\,\ulomega\,\overline{\phi}_{0}
- 4\,\overline{\zeta}_{1}\,\phi_{0}\,\overline{\phi}_{0}\,\phi_{1}
+ 2\,\zeta_{3}\,\overline{\phi}_{0}^{2}\,\phi_{1}
+ 2\,\overline{\zeta}_{3}\,\phi_{0}^{2}\,\overline{\phi}_{1}
- 4\,\zeta_{1}\,\phi_{0}\,\overline{\phi}_{0}\,\overline{\phi}_{1} \\[2pt]
&\quad
+ 2\,\overline{\zeta}_{0}\,\phi_{0}\,\phi_{1}\,\overline{\phi}_{1}
+ 2\,\zeta_{0}\,\overline{\phi}_{0}\,\phi_{1}\,\overline{\phi}_{1}
+ 2\,\mathrm{i}\,\overline{\zeta}_{0}\,\tilde{\zeta}_{4}
-2\,\mathrm{i}\,\zeta_{0}\,\bar{\tilde\zeta}_{4}
- 4\,\mathrm{i}\,\overline{\zeta}_{0}\,\phi_{0}\,\Timu
+ 4\,\mathrm{i}\,\zeta_{0}\,\overline{\phi}_{0}\,\Timu \\[2pt]
&\quad
- 2\,\mathrm{i}\,\overline{\phi}_{0}\,\tilde{\zeta}_{4}\,\rho
+ 2\,\mathrm{i}\,\phi_{0}\,\overline{\tilde{\zeta}}_{4}\,\rho
- 2\,\mathrm{i}\,\overline{\zeta}_{3}\,\phi_{1}\,\sigma
- \lambda\,\rho\,\sigma
+ 2\,\mathrm{i}\,\zeta_{3}\,\overline{\phi}_{1}\,\overline{\sigma}
- \overline{\lambda}\,\rho\,\overline{\sigma}
+ 2\,\ulomega\,\sigma\,\overline{\sigma} \\[2pt]
&\quad
- 2\,\Timu\,\sigma\,\overline{\sigma}
- 2\,\mathrm{i}\,\overline{\zeta}_{3}\,\phi_{0}\,\tau
+ 5\,\mathrm{i}\,\zeta_{1}\,\overline{\phi}_{0}\,\tau
- 2\,\mathrm{i}\,\overline{\zeta}_{0}\,\phi_{1}\,\tau
- \overline{\sigma}\,\tau^{2}
- 5\,\mathrm{i}\,\overline{\zeta}_{1}\,\phi_{0}\,\bar\tau
+ 2\,\mathrm{i}\,\zeta_{3}\,\overline{\phi}_{0}\,\tau \\[2pt]
&\quad
+ 2\,\mathrm{i}\,\zeta_{0}\,\overline{\phi}_{1}\,\bar\tau
-\sigma\bar\tau^2
- 2\,\mathrm{i}\,\overline{\phi}_{0}\,(\meth\,\zeta_{1})
+ \overline{\sigma}\,(\meth\,\overline{\tau})
+ 2\,\mathrm{i}\,\phi_{0}\,(\meth'\,\overline{\zeta}_{1})
+ \sigma\,(\meth'\,\overline{\tau}).
\end{align*}

\section{Energy estimate support for Bianchi identity}
In this section we show details of analysis of matter fields in Prop. \ref{EnergyEstimateCurvature}.
\subsection{Term $\phi_j\meth\zeta_l$}
\label{phiethzeta}
For the top derivative terms we have
\begin{align*}
\int_0^v\int_{u_{\infty}}^u\frac{a}{|u'|^2}||a^5\phi_j\meth^{11}\zeta_0||^2_{L^2_{sc}(\mathcal{S}_{u',v'})}\leq&
\phi[\phi_j]^2\int_0^v\int_{u_{\infty}}^u
\frac{a^2}{|u'|^4}||a^5\mathcal{D}^{11}\zeta_0||^2_{L^2_{sc}(\mathcal{S}_{u',v'})} \\
\leq&\phi[\phi_j]^2\bmzeta[\zeta_0]^2
\int_{u_{\infty}}^u\frac{a^3}{|u'|^4}
\leq\frac{a^3}{|u|^3}\phi[\phi_j]^2\bmzeta[\zeta_0]^2
\end{align*}
Here 
\begin{align*}
\bmzeta[\zeta_0]^2\equiv\frac{1}{a}\int_0^v||a^{\frac{k-1}{2}}\mathcal{D}^{k}\zeta_0||^2_{L^2_{sc}(\mathcal{S}_{u,v'})}
\end{align*}

\begin{align*}
\int_0^v\int_{u_{\infty}}^u\frac{a}{|u'|^2}||a^5\phi_j\meth^{11}\zeta_{1,3}||^2_{L^2_{sc}(\mathcal{S}_{u',v'})}\leq&
\phi[\phi_j]^2\int_0^v\int_{u_{\infty}}^u
\frac{a^2}{|u'|^4}||a^5\mathcal{D}^{11}\zeta_{1,3}||^2_{L^2_{sc}(\mathcal{S}_{u',v'})} \\
\leq&\frac{a^2}{|u|^2}\phi[\phi_j]^2
\frac{1}{a}\int_{u_{\infty}}^u\frac{a}{|u'|^2}||a^5\mathcal{D}^{11}\zeta_{1,3}||^2_{L^2_{sc}(\mathcal{S}_{u',v'})}\\
\leq&\frac{a^2}{|u|^2}\phi[\phi_j]^2\underline{\bmzeta}[\zeta_{1,3}]^2
\end{align*}

\begin{align*}
\int_0^v\int_{u_{\infty}}^u\frac{a}{|u'|^2}||a^5\phi_j\meth^{11}\zeta_{2}||^2_{L^2_{sc}(\mathcal{S}_{u',v'})}\leq&
\phi[\phi_j]^2\int_0^v\int_{u_{\infty}}^u
\frac{a^2}{|u'|^4}||a^5\mathcal{D}^{11}\zeta_{2}||^2_{L^2_{sc}(\mathcal{S}_{u',v'})} \\
\leq&\frac{a}{|u|^2}\phi[\phi_j]^2
\int_{u_{\infty}}^u\frac{a}{|u'|^2}||a^5\mathcal{D}^{11}\zeta_{2}||^2_{L^2_{sc}(\mathcal{S}_{u',v'})}\\
\leq&\frac{a}{|u|^2}\phi[\phi_j]^2\underline{\bmzeta}[\zeta_{2}]^2
\end{align*}

\begin{align*}
\int_0^v\int_{u_{\infty}}^u\frac{a}{|u'|^2}||a^5\phi_j\meth^{11}\zeta_4||^2_{L^2_{sc}(\mathcal{S}_{u',v'})}\leq&\phi[\phi_i]^2\int_0^v\int_{u_{\infty}}^u
\frac{a^2}{|u'|^4}||a^5\mathcal{D}^{11}\zeta_4||^2_{L^2_{sc}(\mathcal{S}_{u',v'})} \\
=&\phi[\phi_j]^2
\int_{u_{\infty}}^u\frac{a}{|u'|^2}\frac{1}{|u'|^2}
\int_0^v||a^{\frac{11}{2}}\mathcal{D}^{11}\zeta_4||^2_{L^2_{sc}(\mathcal{S}_{u',v'})} \\
\leq&\frac{a}{|u|}\phi[\phi_j]^2\bmzeta[\zeta_4]^2,
\end{align*}
Here
\begin{align*}
\bm\zeta[\zeta_4]^2\equiv
\frac{a}{|u|^2}\int_0^v||a^{\frac{k-1}{2}}
\mathcal{D}^k\zeta_4||^2_{L^2_{sc}(\mathcal{S}_{u,v'})}.
\end{align*}

\begin{align*}
\int_0^v\int_{u_{\infty}}^u\frac{a}{|u'|^2}||a^5\phi_j\meth^{11}\zeta_5||^2_{L^2_{sc}(\mathcal{S}_{u',v'})}\leq&
\phi[\phi_j]^2\int_0^v\int_{u_{\infty}}^u
\frac{a^2}{|u'|^4}||a^5\mathcal{D}^{11}\zeta_5||^2_{L^2_{sc}(\mathcal{S}_{u',v'})} \\
=&\phi[\phi_j]^2
\int_0^v\int_{u_{\infty}}^u
\frac{a}{|u'|^4}||a^{\frac{11}{2}}\mathcal{D}^{11}\zeta_5||^2_{L^2_{sc}(\mathcal{S}_{u',v'})} \\
\leq&\phi[\phi_j]^2\underline{\bmzeta}[\zeta_5]^2,
\end{align*}
Here
\begin{align*}
\underline{\bm\zeta}[\zeta_5]^2\equiv
\int_{u_{\infty}}^u\frac{a^2}{|u'|^4}
||a^{\frac{k-1}{2}}\mathcal{D}^k\zeta_5||^2_{L^2_{sc}(\mathcal{S}_{u,v'})} 
\end{align*}

For the non-top-derivative terms
\begin{align*}
&\sum_{i_1+...+i_4=k, i_4<k}\int_0^v\int_{u_{\infty}}^u\frac{a}{|u'|^2}
||\meth^{i_1}\Gamma^{i_2}\meth^{i_3}\phi_j\meth^{i_4+1}\zeta_l||^2_{L^2_{sc}(\mathcal{S}_{u,v'})} \\
\leq&\frac{1}{a}\int_0^v\int_{u_{\infty}}^u\frac{a}{|u'|^2}\frac{1}{|u'|^2}a\phi[\phi_j]^2
\frac{|u'|^2}{a}\zeta[\zeta_{4,5}]^2
\lesssim\frac{1}{a}\phi[\phi_j]^2\zeta[\zeta_{4,5}]^2.
\end{align*}

\subsection{Term $\phi_j^2\Psi_l$}
\label{phi2Psi}
The top derivative of curvature
\begin{align*}
&\int_0^v\int_{u_{\infty}}^u\frac{a}{|u'|^2}||\phi_j^2
(a^{\frac{1}{2}}\mathcal{D})^k\Psi_l||^2_{L^2_{sc}(\mathcal{S}_{u',v'})}\leq
\int_0^v\int_{u_{\infty}}^u\frac{a}{|u'|^2}\frac{1}{|u'|^4}a\phi[\phi_j]^2a\phi[\phi_l]^2
||(a^{\frac{1}{2}}\mathcal{D})^k\Psi_l||^2_{L^2_{sc}(\mathcal{S}_{u',v'})} \\
\leq&\int_{u_{\infty}}^u\frac{a^4}{|u'|^6}\phi[\phi_j]^4
\frac{1}{a}\int_0^v||(a^{\frac{1}{2}}\mathcal{D})^k\Psi_0||^2_{L^2_{sc}(\mathcal{S}_{u',v'})}
+\int_0^v\frac{a^2}{|u|^4}\phi[\phi_j]^4\int_{u_{\infty}}^u\frac{a}{|u'|^2}
||(a^{\frac{1}{2}}\mathcal{D})^k\Psi_{1,2,3,4}||^2_{L^2_{sc}(\mathcal{S}_{u',v'})} \\
\leq&\frac{a^4}{|u|^5}\phi[\phi_j]^4\bm\Psi[\Psi_0]^2
+\frac{a^2}{|u|^4}\phi[\phi_j]^4\underline{\bm\Psi}[\Psi_{1,2,3,4}]^2
\end{align*}

For the non-top-derivative terms
\begin{align*}
&\sum_{i_1+...+i_5=k, i_5<k}\int_0^v\int_{u_{\infty}}^u\frac{a}{|u'|^2}
||\meth^{i_1}\Gamma^{i_2}\meth^{i_3}\phi_{j_1}\meth^{i_4}\phi_{j_2}
\meth^{i_5}\Psi_l||^2_{L^2_{sc}(\mathcal{S}_{u,v'})} \\
\leq&\int_{u_{\infty}}^u\frac{a}{|u'|^2}\frac{a^2}{|u'|^4}\phi[\phi_j]^4a\Psi[\Psi_0]^2
\leq\frac{a^4}{|u|^5}\phi[\phi_j]^4\Psi[\Psi_0]^2.
\end{align*}

\subsection{Term $\zeta_l^2$}
\label{zeta2}
The large terms are $\zeta_0\zeta_l$ where $l\neq0$. We have
\begin{align*}
&\sum_{i_1+...+i_4=k}\int_0^v\int_{u_{\infty}}^u\frac{a}{|u'|^2}
||\meth^{i_1}\Gamma^{i_2}\meth^{i_3}\zeta_{j_1}
\meth^{i_4}\zeta_{j_2}||^2_{L^2_{sc}(\mathcal{S}_{u,v'})} \\
\leq&\int_{u_{\infty}}^u\frac{a}{|u'|^2}\frac{a}{|u'|^2}\zeta[\zeta_0]^2\zeta[\zeta_j]^2
\leq\frac{a^2}{|u|^3}\zeta[\zeta_j]^4.
\end{align*}

\subsection{Term $\zeta_j\phi_l^3$}
\label{zetaphi3}
The large terms are $\zeta_0\phi_l^3$. We have
\begin{align*}
&\sum_{i_1+...+i_6=k}\int_0^v\int_{u_{\infty}}^u\frac{a}{|u'|^2}
||\meth^{i_1}\Gamma^{i_2}\meth^{i_3}\zeta_{j}
\meth^{i_4}\phi_{l_1}\meth^{i_5}\phi_{l_2}\meth^{i_6}\phi_{l_3}||^2_{L^2_{sc}(\mathcal{S}_{u,v'})} \\
\leq&\int_{u_{\infty}}^u\frac{a}{|u'|^2}\frac{1}{|u'|^6}a^4\mathcal{O}^8
\leq\frac{a^5}{|u|^7}\mathcal{O}^8.
\end{align*}

\subsection{Term $\zeta_j\phi_l\Gamma$}
\label{zetaphiGamma}
The large terms are $\zeta_0\phi_1\mu$ in eq $\mthorn\TiPsi_3$, 
$\zeta_0\phi_0\mu$ in eq $\mthorn\Psi_2$, we have
\begin{align*}
&\sum_{i_1+...+i_5=k}\int_0^v\int_{u_{\infty}}^u\frac{a}{|u'|^2}
||\meth^{i_1}\Gamma^{i_2}\meth^{i_3}\zeta_{l}
\meth^{i_4}\phi_{l}\meth^{i_5}\Gamma||^2_{L^2_{sc}(\mathcal{S}_{u,v'})} \\
\leq&\int_{u_{\infty}}^u\frac{a}{|u'|^2}\frac{1}{|u'|^4}
a\zeta[\zeta_0]^2a\phi[\phi_j]^2\frac{|u'|^4}{a^2}\Gamma[\mu]^2 
\leq\frac{a}{|u|}\zeta[\zeta_0]^2\phi[\phi_j]^2\Gamma[\mu]^2.
\end{align*}
Moreover, note here $\zeta_j$ does contain $\Tizeta_{4,5}$. 
For terms contain $\Tizeta_{4,5}$, the large terms are $\Tizeta_{4,5}\phi_j\mu$ in 
$\mthorn'\TiPsi_3$ and $\mthorn'\Psi_2$. 
We have control
\begin{align*}
&\sum_{i_1+...+i_5=k}\int_0^v\int_{u_{\infty}}^u\frac{a}{|u'|^2}
||\meth^{i_1}\Gamma^{i_2}\meth^{i_3}\Tizeta_{4,5}
\meth^{i_4}\phi_{l}\meth^{i_5}\Gamma||^2_{L^2_{sc}(\mathcal{S}_{u,v'})} \\
\leq&\int_0^v\int_{u_{\infty}}^u\frac{a}{|u'|^2}\frac{1}{|u'|^4}a\phi[\phi_j]^2\frac{|u'|^4}{a^2}\Gamma[\mu]^2
||(a^{\frac{1}{2}}\mathcal{D})^k\Tizeta_{4,5}||^2_{L^2_{sc}(\mathcal{S}_{u,v'})}\\
&+\int_0^v\int_{u_{\infty}}^u\frac{a}{|u'|^2}\frac{1}{|u'|^4}a\phi[\phi_j]^2\frac{|u'|^4}{a^2}\Gamma[\mu]^2
\zeta[\Tizeta_{4,5}]^2\\
\leq&\frac{1}{|u|}\phi[\phi_j]^2\Gamma[\mu]^2\bm\zeta[\Tizeta_4]^2
+\frac{1}{a}\phi[\phi_j]^2\Gamma[\mu]^2\underline{\bm\zeta}[\Tizeta_5]^2
+\frac{1}{a}\phi[\phi_j]^2\Gamma[\mu]^2\zeta[\Tizeta_{4,5}]^2.
\end{align*}

\subsection{Term $\phi_j^2\mu\Gamma(\lambda,\sigma,\rho,\tau)$}
\label{phi2muGamma}
The large terms are $\phi_j^2\mu\lambda$ and we have control
\begin{align*}
&\sum_{i_1+...+i_6=k}\int_0^v\int_{u_{\infty}}^u\frac{a}{|u'|^2}
||\meth^{i_1}\Gamma^{i_2}\meth^{i_3}\phi_{l_1}\meth^{i_4}\phi_{l_2}
\meth^{i_5}\mu\meth^{i_6}\Gamma||^2_{L^2_{sc}(\mathcal{S}_{u,v'})} \\
\leq&\int_{u_{\infty}}^u\frac{a}{|u'|^2}\frac{1}{|u'|^6}a^2\phi[\phi_j]^4
\frac{|u'|^4}{a^2}\Gamma[\mu]^2\frac{|u|^2}{a}\Gamma[\lambda]^2
\lesssim\frac{1}{|u|}\mathcal{O}^8.
\end{align*}


\end{document}